\documentclass[11pt,letter]{article}
\usepackage{graphicx} 
\usepackage{mathtools}
\usepackage[utf8]{inputenc}

\allowdisplaybreaks
\usepackage{fullpage}

\usepackage[utf8]{inputenc}
\usepackage[T1]{fontenc}
\usepackage{amsmath}
\usepackage{mathtools}
\usepackage{amsthm,amsfonts,amstext}
\usepackage{amssymb}

\usepackage{enumerate}

\usepackage[linesnumbered,ruled,vlined]{algorithm2e}

\usepackage{paralist}
\usepackage{bm}
\usepackage[dvipsnames]{xcolor}
\usepackage{todonotes}

\usepackage[hidelinks]{hyperref}
\usepackage[nameinlink]{cleveref}
\usepackage{multicol}

\usepackage{thmtools,thm-restate}
\usepackage{complexity}

\newtheorem{theorem}{Theorem}
\newtheorem{lemma}{Lemma}

\newtheorem{fact}{Fact}
\newtheorem{definition}{Definition}

\newtheorem{proposition}{Proposition}

\newcommand{\PAP}{\textsc{PAP}\xspace}

\newcommand{\tecpc}{\textsc{$2$ECPC}\xspace}
\newcommand{\tpp}{\textsc{TPP}\xspace}
\newcommand{\OPT}{\textsc{OPT}\xspace}
\newcommand{\opt}{{\rm opt}}
\newcommand{\credits}{{\rm credits}\xspace}
\newcommand{\RED}{\textsc{RED}\xspace}
\newcommand{\red}{{\rm red}\xspace}
\newcommand{\ALG}{\textsc{ALG}\xspace}
\newcommand{\eps}{\varepsilon}
\newcommand{\apxrn}{\ensuremath{1.9412}}
\newcommand{\apxr}{\ensuremath{1.9412}}
\newcommand{\apxrf}{\ensuremath{1.9955}}
\newcommand{\cost}{{\rm cost}}

\newcommand{\mcP}{\ensuremath{\mathcal{P}}\xspace}
\newcommand{\mcp}{\ensuremath{\mathcal{P}}\xspace}
\title{Improved Approximation Algorithms for Path and Forest Augmentation via a Novel Relaxation}
\author{
Felix Hommelsheim \\ University of Bremen \\ \texttt{fhommels@uni-bremen.de}
}
\date{}

\begin{document}
\pagenumbering{gobble}
\maketitle

\begin{abstract}
The Forest Augmentation Problem (FAP) asks for a minimum set of additional edges (links) that make a given forest 2-edge-connected while spanning all vertices. 
A key special case is the Path Augmentation Problem (PAP), where the input forest consists of vertex-disjoint paths.
Grandoni, Jabal Ameli, and Traub [STOC'22] recently broke the long-standing 2-approximation barrier for FAP, achieving a 1.9973-approximation. A crucial component of this result was their 1.9913-approximation for PAP; the first better-than-2 approximation for PAP.

In this work, we improve these results and provide a \apxr-approximation for PAP, which implies a \apxrf-approximation for FAP.
One of our key innovations is a $(\frac{7}{4} + \varepsilon)$-approximation preserving reduction to so-called structured instances, which simplifies the problem and enables our improved approximation.
Additionally, we introduce a new relaxation inspired by 2-edge covers and analyze it via a corresponding packing problem, where the relationship between the two problems is similar to the relationship between 2-edge covers and 2-matchings. 
Using a factor-revealing LP, we bound the cost of our solution to the packing problem w.r.t.\ the relaxation and derive a strong initial solution.
We then transform this solution into a feasible PAP solution, combining techniques from FAP and related connectivity augmentation problems, along with new insights. 
A key aspect of our approach is leveraging the properties of structured PAP instances to achieve our final approximation guarantee.
Our reduction framework and relaxation may be of independent interest in future work on connectivity augmentation problems.

\end{abstract}

\clearpage

\tableofcontents

\newpage

\pagenumbering{arabic}

\section{Introduction}

Designing reliable networks that can withstand failures is a fundamental challenge in real-world applications. 
The field of \emph{survivable network design} (SND) aims to construct cost-efficient networks that maintain connectivity even when some edges fail. 
Many natural SND problems are \NP-hard (even \APX-hard), justifying the study of approximation algorithms.
A landmark result in this area is Jain's 2-approximation algorithm~\cite{jain2001factor}, which applies to a broad class of connectivity problems. 
However, despite extensive research, breaking the 2-approximation barrier remains a major open question, even for many special cases of SND.

In this paper, we focus on two fundamental network augmentation problems: The Forest Augmentation Problem (FAP) and the Path Augmentation Problem (PAP). 
In the Forest Augmentation Problem we are given an undirected graph $G = (V, F \cup L)$, where $F$ is a forest and $L \subseteq \binom{V}{2}$ is a set of additional edges (links) and the goal is to find a minimum-cardinality subset $S \subseteq L$ such that $(V, F \cup S)$ is 2-edge-connected\footnote{A graph is $k$-edge-connected (kEC) if it remains connected after the removal of an arbitrary subset of $k-1$ edges.}. 
The assumption that $F$ is a forest is no restriction, as one can obtain an equivalent instance by contracting each 2-edge-connected component of $F$ into a single vertex.
The Path Augmentation Problem (PAP) is a special case of FAP where the given forest is composed of vertex-disjoint paths. 
We formally denote an instance of PAP by $G=(V, E(\mathcal{P}) \cup L)$, where $\mathcal{P}$ is the set of vertex-disjoint paths and $E(\mathcal{P})$ denotes its corresponding edge set.
The goal is to find a minimum-cardinality subset $S \subseteq L$ such that $(V, E(\mcP) \cup S)$ is 2-edge-connected. 
Both problems are \APX-hard~\cite{CL99,F98}, ruling out the possibility of a Polynomial-Time Approximation Scheme (PTAS), unless $\P = \NP$. 
Until recently, no polynomial-time algorithm achieving a better-than-2 approximation was known for either problem.

Grandoni, Jabal Ameli, and Traub~\cite{grandoni2022breaching} were the first to break this barrier for PAP and FAP, achieving a 1.9913-approximation for PAP\footnote{In~\cite{grandoni2022breaching} a different result is shown with a bi-criteria approximation ratio. However, the result stated here for \PAP is given implicitly in~\cite{grandoni2022breaching}.}. 
This result, combined with another of their techniques, led to a 1.9973-approximation for FAP.

\subsection{Our Contribution}
We present improved approximation algorithms for FAP and PAP.
Our main result is a substantial improvement for PAP.

\begin{theorem}
\label{thm:PAP:main}
    There is a polynomial-time $\apxr$-approximation algorithm for PAP.
\end{theorem}

To achieve this improvement, we introduce a new $\smash{(\frac{7}{4} + \varepsilon)}$-approximation preserving reduction to structured instances, which is key to our approach. 
Additionally, we develop a new relaxation of PAP, called 2-Edge Cover with Path Constraints (2ECPC), in which we have to compute a minimum-cardinality set of links $X \subseteq L$ such that $(V, E(\mcP) \cup X)$ is a 2-edge cover with the additional property that each path $P \in \mcP$ has two outgoing links.
Instead of directly obtaining an approximative solution for 2ECPC, we consider a novel packing problem, called the Track Packing Problem (TPP).
TPP is in some sense the corresponding packing problem to the covering problem 2ECPC, where the relationship between the two problems is inspired by the relationship between 2-edge covers and 2-matchings (precise definitions follow).
We compute up to $|\mcP|+2$ many solutions to TPP and analyze their cost using a factor-revealing linear program (LP). 
The resulting approximate solution serves as a suitable starting point for our framework. 
Later, this solution is turned into a feasible solution, highly relying on the fact that the input instance is structured.
Our techniques also combine insights from~\cite{grandoni2022breaching} with methods used in other augmentation problems, refined with new ideas tailored for structured instances of PAP. 
We believe that both the reduction to structured instances and the new relaxation 2ECPC as well as the corresponding packing problem TPP may be of independent interest for other connectivity augmentation problems.

By combining \Cref{thm:PAP:main} with a result from~\cite{grandoni2022breaching}, we also give an improved approximation algorithm for FAP. 

\begin{theorem}
\label{thm:FAP:main}
    There is a polynomial-time \apxrf-approximation algorithm for FAP.
\end{theorem}

\subsection{Previous and Related Work}

FAP generalizes several well-known network design problems, many of which have seen improvements beyond the 2-approximation barrier. One important special case is the Tree Augmentation Problem (TAP), where $F$ is a tree. Over the past decade, a sequence of results~\cite{adjiashvili2018beating, cheriyan2018approximating, cheriyan2018approximating2, cheriyan2008integrality, cohen20131+, even20091, fiorini2018approximating, frederickson1981approximation,  grandoni2018improved, khuller1993approximation,    kortsarz2015simplified, kortsarz2018lp, nagamochi2003approximation, nutov2021tree} has led to a 1.393-approximation by Cecchetto, Traub, and Zenklusen~\cite{cecchetto2021bridging}. 
The weighted version of TAP has also been extensively studied, culminating in a $1.5 + \varepsilon$ approximation by Traub and Zenklusen~\cite{traub2022better, traub2022local, traub2022A}, which even holds for the more general connectivity augmentation problem.
In the connectivity augmentation problem we are given a \mbox{$k$-edge-connected} graph~$G$ and additional links $L$ and the task is to compute a set of links $X \subseteq L$ of minimum size such that $G \cup X$ is $(k+1)$-edge-connected.
Before, the barrier of $2$ for connectivity augmentation was breached by Byrka, Grandoni, and Jabal Ameli~\cite{DBLP:journals/siamcomp/ByrkaGA23}.

Another relevant special case of FAP is the unweighted 2-edge-connected spanning subgraph problem (unweighted 2-ECSS), where $F = \emptyset$.
Several works~\cite{cheriyan2001improving, garg2023improved, hunkenschroder20194,  khuller1994biconnectivity,  kobayashi2023approximation,   sebHo2014shorter} have led to better-than-2 approximations.
Recently, Bosch-Calvo, Grandoni, and Jabal Ameli~\cite{bosch20245} and Garg, Hommelsheim, and Lindermayr~\cite{garg2024two} improved the ratio to $\frac{5}{4}$ and $\frac{5}{4} + \varepsilon$, respectively; see also \cite{2ECSS-merge}.

The Matching Augmentation Problem (MAP) is the special case of PAP in which $F$ is a matching, i.e., each path consists of a single edge. 
Cheriyan, Dippel, Grandoni, Khan, and Narayan introduced this problem and gave a $\frac{7}{4}$-approximation~\cite{cheriyan2020matching}.
Cheriyan, Cummings, Dippel, and Zhu~\cite{cheriyan2023improved} improved this to $\frac{5}{3}$ and Garg, Hommelsheim, and Megow improved this further to $\frac{13}{8}$~\cite{garg2023matching}.
An LP-based approximation algorithm for MAP with worse guarantee yet simpler analysis has been developed by Bamas, Drygala, and Svensson~\cite{DBLP:conf/ipco/BamasDS22}.

Beyond edge connectivity, many of these problems have natural vertex connectivity counterparts, where the objective is to compute a 2-vertex-connected spanning subgraph, ensuring connectivity despite a single vertex failure. 
For the vertex-connected version of unweighted 2-ECSS, several algorithms obtain better-than-2 approximations~\cite{CT00, GVS93, HV17,  khuller1994biconnectivity}. 
The current best-known result is a $\frac{4}{3}$-approximation by Bosch-Calvo, Grandoni, and Jabal Ameli~\cite{BGJ23}. 
Similarly, better-than-2 approximations have been achieved for the vertex-connectivity version of TAP, with recent improvements in~\cite{AHS23,N20waoa}.

\subsection{Organization of the Paper}

The paper is organized as follows. 
In \Cref{sec:overview}, we provide an overview of our approach, outlining the key techniques used to achieve our improved approximation guarantees. 
Sections~\ref{sec:structured}-\ref{sec:gluing} present the four main components of our algorithm in detail. 
Finally, in \Cref{sec:preliminaries:FAP-PAP}, we describe how our results for PAP imply the improved approximation for FAP, i.e., \Cref{thm:FAP:main}.

\subsection{Notation}
\label{sec:preliminaries:notation}

We assume that all instances of optimization problems studied in this paper admit a feasible solution, which can be tested in polynomial time.
Throughout this paper we consider an instance of PAP given by $G=(V, E(\mathcal{P}) \cup L)$,
where $\mcP$ is a collection of disjoint paths, $E(\mcP)$ denotes their edge set, and $L \subseteq \left(\begin{smallmatrix}V \\ 2\end{smallmatrix}\right)$ is a set of additional links. 
We define $E = E(\mathcal{P}) \cup L$ as the set of edges of $G$.
For $\mcp' \subseteq \mcp$ we write $V^*(\mathcal{P}')$ to denote the set of endpoints of all paths in $\mathcal{P}'$, i.e., vertices that have degree 1 w.r.t.\ $E(\mcp')$.
We further define the following subsets of links: $L^* \subseteq L$ is the set of links $uv$ such that $u, v \in V^*(\mcP)$. 
We use $\hat{L} \subseteq L^*$ to denote the set of links of $L^*$ whose endpoints belong to the same path, while $\bar{L} = L^* \setminus \hat{L}$ contains links of $L^*$ that connect different paths.
For a subgraph $H$ of $G$ we write $L(H)$ to denote the set of links contained in $H$.
For $P \in \mcP$ and $u, v \in V(P)$, let $P^{uv}$ be the path from $u$ to $v$ on $P$.
We denote by $\OPT(G)$ some optimal solution of $G$ and by $\opt(G)$ the cost of the optimal solution, i.e., the number of links in $\OPT(G)$. If $G$ is clear from context, we simply write $\OPT$ and $\opt$, respectively.

Further, we use standard graph notation.
Given a set of vertices $U \subseteq V$, we let $G|U$ denote the multi-graph obtained from $G$ by contracting $U$ to a single vertex. 
For a subgraph $H$ of $G$ we simply write $G|H$ instead of $G|V(H)$.
More generally, given a set $\mathcal{T}= \{T_1, ..., T_k \}$ of vertex disjoint sets $T_1, ..., T_k \subseteq V(G)$, $G|\{T_1, ..., T_k\}$ (or $G|\mathcal{T}$ in short) denotes the graph obtained by contracting each $T_i$ into a single vertex.
We say that $X \subseteq V$ is a separator if $G \setminus X$ is disconnected.
For $V' \subseteq V(G)$ we define $N_G(V')$ as the neighborhood of $V'$ in $G$ and $\delta(V')$ refers to the edge set with exactly one endpoint in $V'$.
A graph is $k$-edge-connected if it remains connected after the removal of any subset of edges of size at most $k-1$ and abbreviate it by $k$EC.
A bridge in $G$ is an edge whose removal from $G$ increases the number of connected components.
A block is a maximal 2-edge-connected subgraph containing at least 2 vertices.
A block $B$ is called \emph{simple} if there is some $P \in \mathcal{P}$ such that $V(B) \subseteq V(P)$, it is called \emph{small} if there are two distinct paths $P_1, P_2 \in \mathcal{P}$ such that $V(B) = V(P_1) \cup V(P_2)$, and \emph{large} if there are three distinct paths $P_1, P_2, P_3 \in \mathcal{P}$ such that $V(P_1) \cup V(P_2) \cup V(P_3) \subseteq V(B)$.
A connected component that is 2EC is called a 2EC \emph{component}.
We call 2EC components small, if they contain exactly 2 links and large if they contain at least 3 links.
A subgraph $H$ is a 2-edge cover of $G$ if $|\delta_H(v)| \geq 2$ for each $v \in V(G)$.

\section{Overview of Our Approach}
\label{sec:overview}
\begin{figure}[t]
	\begin{center}
 
\resizebox{.8\linewidth}{!}{
\begin{tikzpicture}
    \node[shape=circle,draw=black, inner sep=3pt, label = left:$a_1$] (A1) at (1,0.5) {};
    \node[shape=circle,draw=black, fill=black, inner sep=2pt, label = left:$a_2$] (A2) at (1,1.5) {};
    \node[shape=circle,draw=black, fill=black, inner sep=2pt, label = below right:$a_3$] (A3) at (1,2.5) {};
    \node[shape=circle,draw=black, inner sep=3pt, label = left:$a_4$] (A4) at (1,3.5) {};

    \node[shape=circle,draw=black, inner sep=3pt, label = below left:$b_1$] (B1) at (3,0) {};
    \node[shape=circle,draw=black, fill=black, inner sep=2pt, label = left:$b_2$] (B2) at (3,1) {};
    \node[shape=circle,draw=black, fill=black, inner sep=2pt, label = right:$b_3$] (B3) at (3,2) {};
    \node[shape=circle,draw=black, inner sep=3pt, label = above left:$b_4$] (B4) at (3,3) {};

    \node[shape=circle,draw=black, inner sep=3pt, label = left:$c_1$] (C1) at (6,1) {};
    \node[shape=circle,draw=black, fill=black, inner sep=2pt, label = below:$c_2$] (C2) at (7.5,1) {};
    \node[shape=circle,draw=black, inner sep=3pt, label = right:$c_3$] (C3) at (9,1) {};
    
    \node[shape=circle,draw=black, inner sep=3pt, label = below:$i_1$] (I1) at (5.5,4) {};
    \node[shape=circle,draw=black, fill=black, inner sep=2pt, label = below right:$i_2$] (I2) at (7,4) {};
    \node[shape=circle,draw=black, inner sep=3pt, label = below left:$i_3$] (I3) at (8.5,4) {};
    
    \node[shape=circle,draw=black, inner sep=3pt, label = above left:$j_1$] (J1) at (5.5,5) {};
    \node[shape=circle,draw=black, fill=black, inner sep=2pt, label = below:$j_2$] (J2) at (7,5) {};
    \node[shape=circle,draw=black, inner sep=3pt, label = above right:$j_3$] (J3) at (8.5,5) {};
    
    \node[shape=circle,draw=black, inner sep=3pt, label = below left:$d_1$] (D1) at (6,0) {};
    \node[shape=circle,draw=black, fill=black, inner sep=2pt, label = above left:$d_2$] (D2) at (7,0) {};
    \node[shape=circle,draw=black, fill=black, inner sep=2pt, label = above right:$d_3$] (D3) at (8,0) {};
    \node[shape=circle,draw=black, inner sep=3pt, label = right:$d_4$] (D4) at (9,0) {};

    \node[shape=circle,draw=black, inner sep=3pt, label = above:$f_1$] (F1) at (12,1) {};
    \node[shape=circle,draw=black, inner sep=3pt, label = above:$f_2$] (F2) at (14,1) {};
    \node[shape=circle,draw=black, inner sep=3pt, label = below:$e_1$] (E1) at (12,0) {};
    \node[shape=circle,draw=black, fill=black, inner sep=2pt, label = below:$e_2$] (E2) at (13,0) {};
    \node[shape=circle,draw=black, inner sep=3pt, label = below:$e_3$] (E3) at (14,0) {};
    
    \node[shape=circle,draw=black, inner sep=3pt, label = above:$g_1$] (G1) at (15,4) {};
    \node[shape=circle,draw=black, inner sep=3pt, label = above:$g_2$] (G2) at (16,4) {};
    \node[shape=circle,draw=black, inner sep=3pt, label = right:$h_1$] (H1) at (16.5,3) {};
    \node[shape=circle,draw=black, fill=black, inner sep=2pt, label = right:$h_2$] (H2) at (16,2) {};
    \node[shape=circle,draw=black, inner sep=3pt, label = right:$h_3$] (H3) at (15.5,1) {};

    \path [dashed] (A1) edge node[left] {} (A2);
    \path [dashed] (A2) edge node[left] {} (A3);
    \path [dashed] (A3) edge node[left] {} (A4);
    
    \path [dashed] (B1) edge node[left] {} (B2);
    \path [dashed] (B2) edge node[left] {} (B3);
    \path [dashed] (B3) edge node[left] {} (B4);
    
    \path [dashed] (C1) edge node[left] {} (C2);
    \path [dashed] (C2) edge node[left] {} (C3);
    
    \path [dashed] (D1) edge node[left] {} (D2);
    \path [dashed] (D2) edge node[left] {} (D3);
    \path [dashed] (D3) edge node[left] {} (D4);
    
    \path [dashed] (E1) edge node[left] {} (E2);
    \path [dashed] (E2) edge node[left] {} (E3);
    
    \path [dashed] (F1) edge node[left] {} (F2);

    \path [dashed] (G1) edge node[left] {} (G2);
    
    \path [dashed] (H1) edge node[left] {} (H2);
    \path [dashed] (H2) edge node[left] {} (H3);
    
    \path [dashed] (I1) edge node[left] {} (I2);
    \path [dashed] (I2) edge node[left] {} (I3);
    
    \path [dashed] (J1) edge node[left] {} (J2);
    \path [dashed] (J2) edge node[left] {} (J3);
    
    \path [-] (A2) edge node[left] {} (B3);
    
    \path[-] (A4)  edge [bend left=45] node {} (A3);
    \path[-] (A4)  edge [bend right=35] node {} (A2);
    
    \path [-] (A1) edge node[left] {} (B1);
    
    \path [-] (A3) edge node[left] {} (B3);
    \path [-] (A3) edge node[left] {} (B4);
    
    \path [-] (C3) edge node[left] {} (D4);
    \path[-] (B1)  edge [bend right=35] node {} (D2);
    \path [-] (B2) edge node[left] {} (D1);
    \path [-] (B4) edge node[left] {} (C1);
    
    \path[-] (C1)  edge [bend left=35] node {} (C3);
    \path[-] (C2)  edge [bend left=35] node {} (F1);
    
    \path[-] (D1)  edge [bend right=35] node {} (E3);
    \path[-] (D3)  edge [bend right=35] node {} (E1);
    
    \path [-] (F1) edge node[left] {} (E1);
    \path [-] (F2) edge node[left] {} (E3);
    
    \path [-] (F1) edge node[left] {} (G1);
    \path [-] (E3) edge node[left] {} (H3);
    
    \path [-] (H2) edge node[left] {} (G1);
    \path [-] (H1) edge node[left] {} (G2);
    
    \path [-] (I1) edge node[left] {} (B4);
    \path [-] (I2) edge node[left] {} (C1);
    \path [-] (I3) edge node[left] {} (C3);
    
    \path[-] (J1)  edge [bend left=35] node {} (J3);
    \path [-] (J1) edge node[left] {} (I1);
    \path [-] (I3) edge node[left] {} (J3);
    \path [-] (I2) edge node[left] {} (J3);
    \path[-] (J2)  edge [bend right=15] node {} (B4);
    \path[-] (J2)  edge [bend right=15] node {} (C1);

\end{tikzpicture}
}
	\end{center}
 \vspace*{-0.5cm}
	\caption{Non-structured instance of PAP. Dashed edges are in $E(\mathcal{P})$, solid edges are links of $L$.
    Filled vertices are interior vertices of paths from $\mcP$, empty vertices are endvertices of paths in $\mcP$.
    The path $P = b_1 b_2 b_3 b_4 \in \mathcal{P}$ is a separator, the cycle $a_2 a_3 a_4 a_2$ is a 1-contractible subgraph, the path $c_1 c_2 c_3 d_4 d_3 d_2 d_1$ is a $P^2$-separator, the cycle $e_1 e_2 e_3 f_2 f_1 e_1$ is a $C^2$-separator, and the path $j_1 j_2 j_3$ is a degenerate path.}
	\label{fig:structured}
\end{figure}
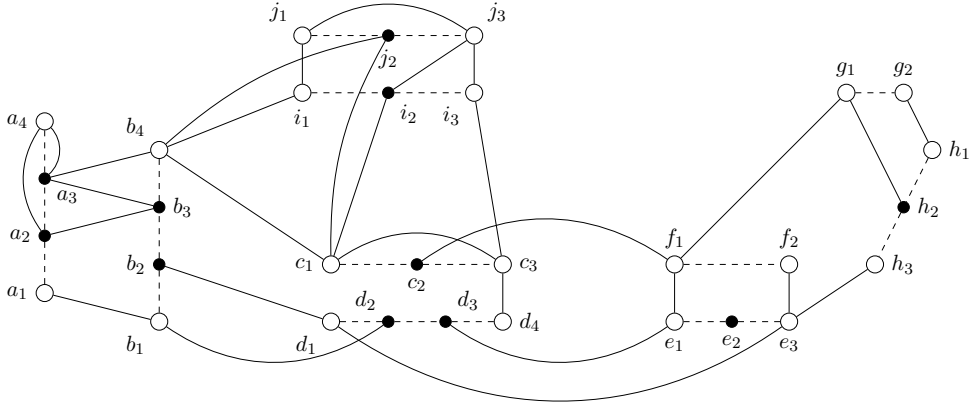

In this section, we provide an overview of our algorithmic approach to prove Theorem~\ref{thm:PAP:main}.
At a high level, our approach is vaguely similar to approaches recently introduced for FAP and PAP~\cite{grandoni2022breaching}, unweighted 2-ECSS~\cite{2ECSS-merge, bosch20245, garg2023improved, garg2024two} and MAP~\cite{cheriyan2020matching,cheriyan2023improved,garg2023matching}.
Our approach has four main steps.

\textbf{1. Reduction to Structured Instances.}
The first step is a reduction to structured instances. 
Specifically, we show that for any $\alpha \geq \frac{7}{4}$ and any small constant $\varepsilon > 0$, if an $\alpha$-approximation exists for structured instances, then a $(\alpha + \varepsilon)$-approximation exists for arbitrary instances. These structured instances possess key properties that enable and simplify our subsequent algorithmic steps. This step has not been utilized in the previous work on FAP and PAP. 

\textbf{2. Computing a Low-Cost Starting Solution.}
The second step includes one of our key innovations. 
Here, we compute a \emph{low-cost} starting solution $S \subseteq L$ such that the graph \mbox{$H = (V, E(\mcP) \cup S)$} is not necessarily feasible but has important structural properties. 
In later steps, we will turn this starting solution into a feasible solution.
To guide this process, we introduce a credit scheme that assigns to any spanning subgraph $H'$ of $G$ credits to vertices, edges, blocks, and components of $H'$. 
The total sum of credits of $H'$ is denoted by $\credits(H')$ and we define $\cost(H') = |H' \cap L| + \credits(H)$, i.e., the number of links in $H'$ plus the total sum of credits assigned to $H'$.
This scheme helps to maintain a cost measure that accounts not only for the number of selected links but also for structural penalties.
Roughly speaking, the number of credits assigned to a solution reflects its distance to a feasible solution.
To compute the starting solution, we introduce a novel relaxation of \PAP, which we use as a lower bound on $\opt$ and which is related to 2-edge covers. 
Instead of computing a solution to the relaxation directly, we consider a closely related packing problem, where the relationship between the packing problem and the relaxation of \PAP is similar to the relationship between 2-matchings and 2-edge covers\footnote{For a graph $G=(V, E)$ a 2-matching $M \subseteq E$ satisfies $\delta_M(v) \leq 2$ for all $v \in V$. A classical result states that one can compute a minimum 2-edge cover by computing a maximum 2-matching and greedily adding edges to each vertex $v \in V$ with $\delta_M(v) < 2$; see~\cite{schrijver2003combinatorial} for more details.}.
Although it turns out that the packing problem and the relaxation are still \NP-hard,
we can compute a solution $S$ to the packing problem such that $H_0=(V, E(\mcP) \cup S)$ satisfies $\cost(H_0) \leq \apxr \cdot \opt$.
In fact, we compute up to $|\mcP| +2$ many solutions $H$ and output the one minimizing $\cost(H)$. 
We analyze the cost of this solution using a factor-revealing LP.
This (most likely infeasible) solution is used as a starting solution.

In the remainder, we step-by-step turn $H_0$ into solutions $H_1, H_2, ...$ that are 'closer to being feasible', while always keeping the invariant that the new solution $H_{i+1}$ satisfies \mbox{$\cost(H_{i+1}) \leq \cost(H_{i})$}. 
Since the credits are non-negative and the starting solution satisfies \mbox{$\cost(H_0) \leq \apxr \cdot \opt$}, eventually we will have a feasible solution using at most $\apxr \cdot \opt$ many links.
There are two reasons why some solution $H$ is not yet feasible: 1) there is a bridge in $H$, and 2) there are several connected components in $H$.
The next two steps deal with these reasons for infeasibility.

\textbf{3. Bridge-Covering.}
The third step, called bridge-covering, ensures that the intermediate solution $H'$ is bridgeless.
Our goal is to add (and sometimes remove) links to (from) $H$ such that the resulting graph $H'$ is bridgeless, while keeping the invariant $\cost(H') \leq \cost(H)$. 
At this stage, the structured nature of~$G$ plays a crucial role, allowing us to argue the existence of cheap augmenting link sets needed to eliminate bridges. 

\textbf{4. Gluing.}
In the fourth and final step, called gluing, we transform $H$ into a fully 2-edge-connected spanning subgraph. 
At this point, $H$ consists of multiple 2-edge-connected components that must be merged. 
We achieve this by iteratively adding well-chosen cycles in the \emph{component graph}, ensuring that each step maintains feasibility and preserves the approximation guarantee. 
The component graph $\Tilde{G}$ arises from $G$ by contracting each component of $H$ to a single vertex.
The structured nature of $G$ again plays a key role in proving that such cycles always exist.

The remainder of this section provides a more detailed overview of these four steps, along with lemmas formalizing each transformation. 
The subsequent subsections are dedicated to the individual steps, leading to the proof of Theorem~\ref{thm:PAP:main}, which follows from the four main lemmas established before.

\subsection{Reduction to Structured Instances}
\label{sec:overview:preprocessing}

Unlike previous algorithms for FAP and PAP~\cite{grandoni2022breaching}, we introduce a reduction to structured instances, which plays a crucial role in our approach. 
Informally, a structured instance is one that avoids certain degenerate configurations that would complicate the design of good approximation algorithms in our approach. 
We formally define structured instances and establish that solving PAP on structured instances suffices to achieve a good approximation algorithm for general instances.
We prove that for all $\alpha \geq \frac{7}{4}$ and any small $\varepsilon > 0$, if there exists an $\alpha$-approximation for PAP on structured instances, then there is an $(\alpha + \varepsilon)$-approximation for arbitrary instances. 
This guarantees that improvements on structured instances directly translate into improvements for the general case.

To define structured instances, we introduce several graph properties, see \Cref{fig:structured}.
First, structured instances do not contain contractible subgraphs.
Roughly speaking, a subgraph is (strongly) $\alpha$-contractible if at least a $\frac{1}{\alpha}$ fraction of its links must appear in any feasible solution, allowing us to safely contract it without a loss in the approximation guarantee. 
We ensure that such subgraphs can be identified and handled efficiently.
Formally, we define $\alpha$-contractible subgraphs as follows.

\begin{restatable}[$\alpha$-contractable subgraph]{definition}{defContractibleSubgraph}
    Let $\alpha \geq 1$ and $t \geq 1$ be fixed constants. 
    Given a $2$-edge-connected graph $G$, a collection $\mathcal{H}$ of vertex-disjoint $2$-edge-connected subgraphs $H_1, H_2, ..., H_k$ of $G$, such that $G|\mathcal{H}$ is also an instance of \PAP, is called \emph{weakly $(\alpha, t , k)$-contractible} if $1 \leq |L(H_i)| \leq t$ for every $i \in [k]$ and there exists some optimal solution that contains at least $\frac{1}{\alpha} |\bigcup_{i \in [k]} L(H_i)|$ links from $\bigcup_{i \in [k]} L(G[V(H_i)])$.
    Furthermore, $\mathcal{H}$ is \emph{strongly} $(\alpha, t , k)$-contractible if additionally every 2EC spanning subgraph of $G$ contains at least $\frac{1}{\alpha} |\bigcup_{i \in [k]} L(H_i)|$ links from $\bigcup_{i \in [k]} L(G[V(H_i)])$.
\end{restatable}

Note that every strongly $(\alpha, t , k)$-contractible subgraph is also a weakly $(\alpha, t , k)$-contractible subgraph.
We say that $\mathcal{H}$ is an $(\alpha, t , k)$-contractible subgraph if it is a weakly or strongly $(\alpha, t , k)$-contractible subgraph.
If $\alpha, t, k$ are clear from the context, we simply refer to $\alpha$-contractible subgraphs or even contractible subgraphs.

Second, structured instances avoid certain small separators. 
Specifically, we prevent the existence of path-separators, $P^2$-separators, and $C^2$-separators. 
If an instance contains such a separator, we decompose it into smaller subinstances, solve each recursively, and merge their solutions while maintaining the desired approximation guarantee.
Formal definitions are given in \Cref{sec:structured}.

A final obstacle to structuring instances is the presence of too many degenerate paths, which are paths $P \in \mcP$ with limited connectivity to the rest of the graph. 
When the fraction of degenerate paths exceeds a threshold $\varepsilon |\mcP|$, we identify certain well-connected subgraphs that allow us to perform controlled contractions and recursive decomposition, ensuring that the structured instance assumption holds.
Formally, degenerate paths are defined as follows.

\begin{restatable}[Degenerate Path]{definition}{defDegeneratePathSeparator}
    Let $P', P \in \mathcal{P}$ be two distinct paths with endpoints $u', v'$ and $ u,  v$, respectively, such that $u' u \in L$, $ v' v \in L$ and $u' v' \in L$. 
    We say that $P'$ is a degenerate path if $N(u'), N(v') \subseteq V(P') \cup V(P)$.
\end{restatable}

Using these definitions, we define structured instances.

\begin{restatable}[($\alpha, \eps)$-structured Instance]{definition}{defStructuredGraph}
\label{def:structured-graph}
We call an instance $G = (V, \mathcal{P} \cup L)$ of PAP $(\alpha, \eps)$-structured if the following properties are satisfied:
\begin{itemize}
    \item[(P0)] \label{prop:structured:size} $ \opt (G) > \frac{3}{\eps}$.  
    \item[(P1)] \label{prop:structured:contractible}
     $G$ does not contain strongly $(\alpha, t, 1)$-contractible subgraphs with $t \leq 10$ many links.
    \item[(P2)] \label{prop:structured:outgoing-1} Every endpoint of a path $P \in \mathcal{P}$ has a link to some other path $P' \in \mathcal{P}$, $P' \neq P$. 
    \item[(P3)] \label{prop:structured:path}
    $G$ does not contain path-separators.
    \item[(P4)] \label{prop:structured:double-path}
    $G$ does not contain $P^2$-separators.
    \item[(P5)] \label{prop:structured:cycle} 
     $G$ does not contain $C^2$-separators.
     \item[(P6)] \label{prop:structured:degenerate-path} 
     $G$ contains at most $\eps \cdot |\mathcal{P}|$ many degenerate paths.
    \item[(P7)] \label{prop:structured:isolatednodes} 
     Every vertex $v \in V$ is on some path $P \in \mathcal{P}$ that contains at least 2 vertices.
\end{itemize}
\end{restatable}
If $\alpha$ and $\varepsilon$ are clear from the context, we simply say $G$ is structured.
Our main lemma in this section is the following.

\begin{restatable}[]{lemma}{lemstructuredmain}
\label{lem:structured:main}
    For all $\alpha \geq \frac{7}{4}$ and constant $\eps \in \big(0, \frac{1}{12}\big]$, if there exists a polynomial-time $\alpha$-approximation algorithm for \PAP on ($\alpha, \eps)$-structured instances, then there exists a polynomial-time $(\alpha + 3 \eps)$-approximation algorithm for \PAP on arbitrary instances.
\end{restatable}

We give a rough overview on how we obtain \Cref{lem:structured:main}. 
Assume that $G$ is not structured. 
In this overview, we focus only on three obstacles that prevent $G$ from being structured:
1) $G$ contains a contractible subgraph, 2) $G$ contains a forbidden separator, or 3) $G$ contains more than $\varepsilon |\mcP|$ many degenerate paths. 
If $G$ contains one of these subgraphs, we can decompose $G$ into at most 2 smaller subinstances, solve each recursively, and merge their solutions while maintaining the approximation guarantee.
We outline our approach for each of these three cases.

\textbf{Property (P1) is not met.}
We can safely contract each component of an $\alpha$-contractible subgraph $\mathcal{H}$ to a single vertex and work with $G|\mathcal{H}$, as long as we only aim for an $\alpha$-approximation. 
Indeed, any $\alpha$-approximate solution for $G|\mathcal{H}$ can be turned to an $\alpha$-approximate solution for $G$ by adding the links of $\mathcal{H}$.
In general it is not straightforward to argue that a 2EC subgraph is weakly $\alpha$-contractible.
However, under certain conditions we can argue that this is indeed the case.
We will show that we can find strongly $(\alpha, t , 1)$-contractible subgraphs in polynomial time if $t$ is constant, i.e., if there is only one component and there are only a constant number of links in the contractible subgraph. 
Hence, after our preprocessing, we can assume that $G$ does not contain any strongly $(\alpha, t , 1)$-contractible subgraphs for some constant $t$ (here $t = 10$ suffices).

\textbf{Property (P3), (P4), or (P5) is not met.}
Assume that $G$ contains a forbidden separator $X$, resulting in $G \setminus X$ to decompose into two disjoint vertex sets $V_1$ and $V_2$ and there are no edges between $V_1$ and $V_2$.
Here, we construct two sub-instances $G_1 = (G \setminus V_2)|X$ and $G_2 = (G \setminus V_1)|X$ arising from $G$ by removing either $V_1$ or $V_2$ and contracting $X$.
We recursively solve these sub-instances approximately and combine both solutions while preserving the desired approximation guarantee.
By introducing a proper notion of \emph{size} of an instance, one can ensure that the sum of the sizes of the two sub-instances is smaller than the size of the original instance and hence this step runs in polynomial time.

\textbf{Property (P6) is not met.}
One of the main innovations in our reduction to structured graphs compared to previous work deals with degenerate paths.
For degenerate paths, we can not do the same strategy as for separators and instead do the following: 
Let $\mathcal{P}' = \{P_1', ..., P_\ell' \}$ be a set of degenerate paths.
We define two collections $\mathcal{K}$ and $\mathcal{C}$ of $\ell$ vertex-disjoint 2-edge-connected subgraphs, each subsuming $V(\mcP')$.
We show that either $\mathcal{C}$ is weakly $(\alpha, 2, \ell)$-contractible or $\mathcal{K}$ is weakly $(\alpha, 1, \ell)$-contractible for any $\alpha \geq \frac 74$.
Hence, either an $\alpha$-approximate solution for $G|\mathcal{K}$ or an $\alpha$-approximate solution for $G|\mathcal{C}$ can be turned into an $\alpha$-approximate solution for $G$ by adding the edges contained in $\mathcal{K}$ or $\mathcal{C}$, respectively.

Unfortunately, in general this recursive procedure leads to an exponential running time if $|\mcP'|$ is small: We need to recursively solve the two sub-instances $G|\mathcal{K}$ and $G|\mathcal{C}$ that both may have roughly the same size as the original instance. 
Since we might need to apply this reduction again in later steps in our reduction to structured graphs, this procedure can result in an exponential running time. 
However, we circumvent this by only applying the above procedure if some constant fraction~$\varepsilon$ of the paths in~$\mathcal{P}$ is degenerate, and hence each sub-instance is significantly smaller than the original one, which overall results in a polynomial running time.

In light of \Cref{lem:structured:main}, it remains to give a $\apxr$-approximation for \PAP on structured instances.
Hence, for the remaining steps, we assume the instance $G=(V, E(\mathcal{P}) \cup L)$ is structured.

\subsection{Credit Scheme, Relaxation, and Starting Solution}

\label{sec:overview:credit-starting}

Throughout our algorithm, we maintain a temporary solution $S \subseteq L$, which is not necessarily a feasible solution. 
To track progress and guide modifications, we assign certain credits to parts of the current graph $H = (V, E(\mcP) \cup S)$. 
These credits are used to ensure that every transformation preserves an approximation bound while gradually converting $H$ into a feasible solution.
We first define key structural components and introduce a cost function that incorporates both the number of selected links and assigned credits.

Let $H$ be a spanning subgraph of $G$, and let $C$ be a connected component of $H$. 
A component is called \emph{complex} if it contains at least one bridge. 
For any such $C$, we define the graphs $G_C$ and $H_C$ as follows:
$G_C$ is the graph obtained from $G$ by contracting every connected component of $H$ except for $C$, as well as every block of $C$.
$H_C$ is the corresponding contraction of $H$.
The structure of $H_C$ consists of a tree $T_C$ along with singleton vertices. 
A vertex in $G_C$ or $H_C$ that corresponds to a block in $H$ is called a \emph{block vertex}. 
If a block vertex corresponds to a simple block, it is called a \emph{simple block vertex}.
We identify the edges/links in $G$ and the corresponding edges/links in $G^C$ and $H^C$, respectively.

A component $C$ is called \emph{trivial} if it consists of a single path $P \in \mcP$ (i.e., $V(C) = V(P)$). 
Otherwise, it is called \emph{nontrivial}. 
A vertex $v \in V$ is called \emph{lonely} if it does not belong to any 2-edge-connected component or block of $H$. 
For two paths $P, P' \in P$ with endpoints $u,v$ and $u', v'$, respectively, we say that $uu' \in L$ is an \emph{expensive link} if $N_G(v) \subseteq (V(P) \cup V(P')) \setminus \{v'\}$. 
If additionally $uu' \in L(H)$ is a bridge in $H$ and $v$ is a lonely leaf in $H$, then $v$ is called an \emph{expensive leaf vertex}.

For a given $H = (V, E(\mathcal{P}) \cup S)$, we assign credits according to the following rules.

\begin{itemize}
\item[(A)] Every lonely vertex $v$ that is a leaf of some complex component $C$ receives 1 credit.
\begin{itemize}
\item[(i)] It receives an additional $\frac{1}{4}$ credit if $v$ is an expensive leaf vertex.
\item[(ii)] It receives an additional $\frac{1}{4}$ credit if there is a path in $H_C$ from $v$ to a simple block vertex of $C$ without passing through a non-simple block vertex.
\end{itemize}
\item[(B)] Every bridge $\ell \in S$ (i.e., a link that is not part of any 2-edge-connected component of $H$) receives $\frac{3}{4}$ credits.
\item[(C)] Every complex component receives 1 credit.
\item[(D)] Every non-simple 2-edge-connected block receives 1 credit.
\item[(E)] Each 2-edge-connected component receives:
\begin{itemize}
\item 2 credits if it is large or contains a degenerate path.
\item $\frac{3}{2}$ credits if it is small and does not contain a degenerate path.
\end{itemize}
\end{itemize}

For any spanning subgraph $H = (V, E(\mathcal{P}) \cup S)$, we define the total credit as $\credits(H)$, which is the sum of all assigned credits under these rules.
We define $\cost(H) = \credits(H) + |S|$.
To prove the approximation guarantee, we show that we maintain the following credit invariant.
    
\begin{restatable}[]{invariant}{invariantcredit}
    \label{invariant:credit}
    $\cost(H) \leq \apxrn \cdot \opt(G)$.
\end{restatable}

Furthermore, we maintain the following two invariants.

\begin{restatable}[]{invariant}{invariantlonely}
    \label{invariant:degree-lonely-vertex}
    Every lonely vertex of $H$ has degree at most $2$ in $H$. Every simple 2EC block $B$ is part of some complex component $C$ and has degree exactly $2$ in $H^C$. 
    Furthermore, for any complex component $C$, no two simple block vertices of $C$ are adjacent to each other in $H^C$.
\end{restatable}

\begin{restatable}[]{invariant}{invariantblock}
    \label{invariant:block-size}
    For every non-simple 2EC block $B$ of $H$ there are two distinct paths $P, P' \in \mathcal{P}$ such that $V(P) \cup V(P') \subseteq V(B)$.
\end{restatable}

Our goal is to design a polynomial-time algorithm that computes a solution $S \subseteq L$ such that $H= (V, E(\mcp) \cup S)$ satisfies all invariants above. This is formalized in the following lemma.
\begin{restatable}[]{lemma}{lemmamainstarting}
    \label{lem:main:starting-solution-satisfies-invariants}
    In polynomial time we can compute a set of links $S \subseteq L$ such that $H= (V, E(\mcp) \cup S)$ satisfies the Invariants~\ref{invariant:credit}-\ref{invariant:block-size}.
\end{restatable}

Next, we outline how to compute such a solution.
Typical approaches for special cases of \PAP, such as unweighted $2$-ECSS and MAP, use some relaxation of the problem, for which an optimum solution can be computed in polynomial time and such a solution serves as a good starting solution. 
A typical relaxation is the $2$-edge cover problem, in which the task is to find a minimum-cost edge set such that each vertex is incident to at least 2 edges. 
For \emph{structured instances} of these special cases of \PAP, optimum 2-edge covers are usually not too far away from an optimum solution of the original problem, and hence they serve as a good lower bound. 
However, in the general case of \PAP, a $2$-edge cover can be quite far away from an optimum solution.
For example, if the 2-edge cover of an instance of \PAP only contains edges of $E(\mathcal{P})$ and $\hat{L}$ (recall that $\hat{L}$ is the set of links that connect endpoints of the same path), then this edge set only 2-edge connects the paths individually, but it does not help to establish overall $2$-edge-connectivity. 
This motivates us to define a more suitable relaxation for \PAP, called 2-Edge Cover with Path Constraints (\tecpc).

\paragraph{A relaxation of \PAP: The 2-Edge Cover with Path Constraints.}

\begin{figure}[t]
\begin{minipage}{.49\textwidth}
	\begin{center}
 
\resizebox{\linewidth}{!}{
\begin{tikzpicture}
    \node[shape=circle,draw=black, inner sep=3pt, label = left:$u_1$] (A1) at (4,1) {};
    \node[shape=circle,draw=black, fill=black, inner sep=2pt, label = left:$w_1$] (A2) at (4,2) {};
    \node[shape=circle,draw=black, inner sep=3pt, label = left:$v_1$] (A3) at (4,3) {};

    \node[shape=circle,draw=black, inner sep=3pt, label = below:$u_2$] (B1) at (5,0.5) {};
    \node[shape=circle,draw=black, fill=black, inner sep=2pt, label = left:$w_2$] (B2) at (5,1.5) {};
    \node[shape=circle,draw=black, inner sep=3pt, label = above:$v_2$] (B3) at (5,2.5) {};

    \node[shape=circle,draw=black, inner sep=3pt, label = left:$u_3$] (C1) at (6,1) {};
    \node[shape=circle,draw=black, fill=black, inner sep=2pt, label = above:$w_3$] (C2) at (7.5,1) {};
    \node[shape=circle,draw=black, inner sep=3pt, label = right:$v_3$] (C3) at (9,1) {};
    
    \node[shape=circle,draw=black, inner sep=3pt, label = below:$u_8$] (G1) at (5.5,4) {};
    \node[shape=circle,draw=black, fill=black, inner sep=2pt, label = below:$w_8$] (G2) at (7,4) {};
    \node[shape=circle,draw=black, inner sep=3pt, label = below:$v_8$] (G3) at (8.5,4) {};
    
    \node[shape=circle,draw=black, inner sep=3pt, label = above:$u_7$] (H1) at (5.5,5) {};
    \node[shape=circle,draw=black, fill=black, inner sep=2pt, label = above:$w_7$] (H2) at (7,5) {};
    \node[shape=circle,draw=black, inner sep=3pt, label = above:$v_7$] (H3) at (8.5,5) {};
    
    \node[shape=circle,draw=black, inner sep=3pt, label = left:$u_4$] (D1) at (6,0) {};
    \node[shape=circle,draw=black, fill=black, inner sep=2pt, label = below:$w_4$] (D2) at (7.5,0) {};
    \node[shape=circle,draw=black, inner sep=3pt, label = right:$v_4$] (D3) at (9,0) {};

    \node[shape=circle,draw=black, inner sep=3pt, label = right:$u_6$] (F1) at (9.5,3) {};
    \node[shape=circle,draw=black, fill=black, inner sep=2pt, label = right:$w_6$] (F2) at (9.5,4) {};
    \node[shape=circle,draw=black, inner sep=3pt, label = right:$v_6$] (F3) at (9.5,5) {};
    \node[shape=circle,draw=black, inner sep=3pt, label = right:$u_5$] (E1) at (10,0) {};
    \node[shape=circle,draw=black, fill=black, inner sep=2pt, label = right:$w_5$] (E2) at (10,1) {};
    \node[shape=circle,draw=black, inner sep=3pt, label = right:$v_5$] (E3) at (10,2) {};
    
    \node[shape=circle,draw=black, inner sep=3pt, label = left:$u_9$] (J1) at (12,0) {};
    \node[shape=circle,draw=black, fill=black, inner sep=2pt, label = left:$w_9$] (J2) at (12,1) {};
    \node[shape=circle,draw=black, inner sep=3pt, label = left:$v_9$] (J3) at (12,2) {};
    
    \node[shape=circle,draw=black, inner sep=3pt, label = right:$u_{10}$] (K1) at (13,0) {};
    \node[shape=circle,draw=black, fill=black, inner sep=2pt, label = right:$w_{10}$] (K2) at (13,1) {};
    \node[shape=circle,draw=black, inner sep=3pt, label = right:$v_{10}$] (K3) at (13,2) {};

    \node[shape=circle,draw=black, inner sep=3pt, label = above:$u_{11}$] (L1) at (11,4) {};
    \node[shape=circle,draw=black, fill=black, inner sep=2pt, label = above:$w_{11}$] (L2) at (12,4) {};
    \node[shape=circle,draw=black, inner sep=3pt, label = above:$v_{11}$] (L3) at (13,4) {};
    
    \path [dashed] (A1) edge node[left] {} (A2);
    \path [dashed] (A2) edge node[left] {} (A3);
    
    \path [dashed] (B1) edge node[left] {} (B2);
    \path [dashed] (B2) edge node[left] {} (B3);
    
    \path [dashed] (C1) edge node[left] {} (C2);
    \path [dashed] (C2) edge node[left] {} (C3);
    
    \path [dashed] (D1) edge node[left] {} (D2);
    \path [dashed] (D2) edge node[left] {} (D3);
    
    \path [dashed] (E1) edge node[left] {} (E2);
    \path [dashed] (E2) edge node[left] {} (E3);
    
    \path [dashed] (F1) edge node[left] {} (F2);
    \path [dashed] (F2) edge node[left] {} (F3);

    \path [dashed] (G1) edge node[left] {} (G2);
    \path [dashed] (G2) edge node[left] {} (G3);
    
    \path [dashed] (H1) edge node[left] {} (H2);
    \path [dashed] (H2) edge node[left] {} (H3);
    
    \path [dashed] (I1) edge node[left] {} (I2);
    \path [dashed] (I2) edge node[left] {} (I3);
    
    \path [dashed] (J1) edge node[left] {} (J2);
    \path [dashed] (J2) edge node[left] {} (J3);
    
    \path [dashed] (K1) edge node[left] {} (K2);
    \path [dashed] (K2) edge node[left] {} (K3);
    
    \path [dashed] (L1) edge node[left] {} (L2);
    \path [dashed] (L2) edge node[left] {} (L3);
    
    \path[-] (A1)  edge [bend left=75] node {} (A3);
    \path[-] (C1)  edge [bend left=45] node {} (C3);
    \path[-] (D1)  edge [bend right=45] node {} (D3);
    \path[-] (L1)  edge [bend left=55] node {} (L3);

    \path[-] (A2)  edge [bend left=25] node {} (B3);
    \path[-] (A2)  edge [bend right=25] node {} (B1);
    
    \path [-] (J3) edge node[left] {} (L2);
    \path [-] (K3) edge node[left] {} (L2);

    \path [-] (G1) edge node[left] {} (H1);
    \path [-] (C2) edge node[left] {} (D2);
    \path[-] (C2)  edge [bend left=0] node {} (F1);
    \path[-] (D2)  edge [bend right=25] node {} (E1);
    \path [-] (E3) edge node[left] {} (J1);

    \path [-] (H3) edge node[left] {} (F2);
    \path [-] (G3) edge node[left] {} (H2);
    \path [-] (F3) edge node[left] {} (G2);
    \path [-] (K1) edge node[left] {} (J2);

\end{tikzpicture}
}
	\end{center}
 \vspace*{-0.5cm}

    \end{minipage}
    \hspace*{0.3cm}
\begin{minipage}{.49\textwidth}
	\begin{center}
 
\resizebox{\linewidth}{!}{
\begin{tikzpicture}
    \node[shape=circle,draw=black, inner sep=3pt, label = left:$u_1$] (A1) at (4,1) {};
    \node[shape=circle,draw=black, fill=black, inner sep=2pt, label = left:$w_1$] (A2) at (4,2) {};
    \node[shape=circle,draw=black, inner sep=3pt, label = left:$v_1$] (A3) at (4,3) {};

    \node[shape=circle,draw=black, inner sep=3pt, label = below:$u_2$] (B1) at (5,0.5) {};
    \node[shape=circle,draw=black, fill=black, inner sep=2pt, label = left:$w_2$] (B2) at (5,1.5) {};
    \node[shape=circle,draw=black, inner sep=3pt, label = above:$v_2$] (B3) at (5,2.5) {};

    \node[shape=circle,draw=black, inner sep=3pt, label = left:$u_3$] (C1) at (6,1) {};
    \node[shape=circle,draw=black, fill=black, inner sep=2pt, label = above:$w_3$] (C2) at (7.5,1) {};
    \node[shape=circle,draw=black, inner sep=3pt, label = right:$v_3$] (C3) at (9,1) {};
    
    \node[shape=circle,draw=black, inner sep=3pt, label = below:$u_8$] (G1) at (5.5,4) {};
    \node[shape=circle,draw=black, fill=black, inner sep=2pt, label = below:$w_8$] (G2) at (7,4) {};
    \node[shape=circle,draw=black, inner sep=3pt, label = below:$v_8$] (G3) at (8.5,4) {};
    
    \node[shape=circle,draw=black, inner sep=3pt, label = above:$u_7$] (H1) at (5.5,5) {};
    \node[shape=circle,draw=black, fill=black, inner sep=2pt, label = above:$w_7$] (H2) at (7,5) {};
    \node[shape=circle,draw=black, inner sep=3pt, label = above:$v_7$] (H3) at (8.5,5) {};
    
    \node[shape=circle,draw=black, inner sep=3pt, label = left:$u_4$] (D1) at (6,0) {};
    \node[shape=circle,draw=black, fill=black, inner sep=2pt, label = below:$w_4$] (D2) at (7.5,0) {};
    \node[shape=circle,draw=black, inner sep=3pt, label = right:$v_4$] (D3) at (9,0) {};

    \node[shape=circle,draw=black, inner sep=3pt, label = right:$u_6$] (F1) at (9.5,3) {};
    \node[shape=circle,draw=black, fill=black, inner sep=2pt, label = right:$w_6$] (F2) at (9.5,4) {};
    \node[shape=circle,draw=black, inner sep=3pt, label = right:$v_6$] (F3) at (9.5,5) {};
    \node[shape=circle,draw=black, inner sep=3pt, label = right:$u_5$] (E1) at (10,0) {};
    \node[shape=circle,draw=black, fill=black, inner sep=2pt, label = right:$w_5$] (E2) at (10,1) {};
    \node[shape=circle,draw=black, inner sep=3pt, label = right:$v_5$] (E3) at (10,2) {};
    
    \node[shape=circle,draw=black, inner sep=3pt, label = left:$u_9$] (J1) at (12,0) {};
    \node[shape=circle,draw=black, fill=black, inner sep=2pt, label = left:$w_9$] (J2) at (12,1) {};
    \node[shape=circle,draw=black, inner sep=3pt, label = left:$v_9$] (J3) at (12,2) {};
    
    \node[shape=circle,draw=black, inner sep=3pt, label = right:$u_{10}$] (K1) at (13,0) {};
    \node[shape=circle,draw=black, fill=black, inner sep=2pt, label = right:$w_{10}$] (K2) at (13,1) {};
    \node[shape=circle,draw=black, inner sep=3pt, label = right:$v_{10}$] (K3) at (13,2) {};

    \node[shape=circle,draw=black, inner sep=3pt, label = above:$u_{11}$] (L1) at (11,4) {};
    \node[shape=circle,draw=black, fill=black, inner sep=2pt, label = above:$w_{11}$] (L2) at (12,4) {};
    \node[shape=circle,draw=black, inner sep=3pt, label = above:$v_{11}$] (L3) at (13,4) {};
    
    \path [dashed] (A1) edge node[left] {} (A2);
    \path [dashed] (A2) edge node[left] {} (A3);
    
    \path [dashed] (B1) edge node[left] {} (B2);
    \path [dashed] (B2) edge node[left] {} (B3);
    
    \path [dashed] (C1) edge node[left] {} (C2);
    \path [dashed] (C2) edge node[left] {} (C3);
    
    \path [dashed] (D1) edge node[left] {} (D2);
    \path [dashed] (D2) edge node[left] {} (D3);
    
    \path [dashed] (E1) edge node[left] {} (E2);
    \path [dashed] (E2) edge node[left] {} (E3);
    
    \path [dashed] (F1) edge node[left] {} (F2);
    \path [dashed] (F2) edge node[left] {} (F3);

    \path [dashed] (G1) edge node[left] {} (G2);
    \path [dashed] (G2) edge node[left] {} (G3);
    
    \path [dashed] (H1) edge node[left] {} (H2);
    \path [dashed] (H2) edge node[left] {} (H3);
    
    \path [dashed] (I1) edge node[left] {} (I2);
    \path [dashed] (I2) edge node[left] {} (I3);
    
    \path [dashed] (J1) edge node[left] {} (J2);
    \path [dashed] (J2) edge node[left] {} (J3);
    
    \path [dashed] (K1) edge node[left] {} (K2);
    \path [dashed] (K2) edge node[left] {} (K3);
    
    \path [dashed] (L1) edge node[left] {} (L2);
    \path [dashed] (L2) edge node[left] {} (L3);
    
    \path[-] (A1)  edge [bend left=75] node {} (A3);
    \path[-] (C1)  edge [bend left=45] node {} (C3);
    \path[-] (D1)  edge [bend right=45] node {} (D3);
    \path[-] (L1)  edge [bend left=55] node {} (L3);

    \path[-] (A2)  edge [bend left=25] node {} (B3);
    \path[-] (A2)  edge [bend right=25] node {} (B1);
    
    \path [-] (J3) edge node[left] {} (L2);
    \path [-] (K3) edge node[left] {} (L2);

    \path [-] (G1) edge node[left] {} (H1);
    \path [-] (C2) edge node[left] {} (D2);
    \path[-] (C2)  edge [bend left=0] node {} (F1);
    \path[-] (D2)  edge [bend right=25] node {} (E1);
    \path [-] (E3) edge node[left] {} (J1);

\end{tikzpicture}
}
	\end{center}
 \vspace*{-0.5cm}
    \end{minipage}
    \caption{Dashed edges are in $E(\mathcal{P})$, solid edges are links of $L$.
    Filled vertices are interior vertices of paths from $\mcP$, empty vertices are endvertices of paths in $\mcP$.
    Left: Feasible solution $Y$ for 2ECPC, since $E(\mcP) \cup Y$ is a 2-edge cover with $\delta_Y(P) \geq 2$ for each $P \in \mcP$. Right: Solution $X = \{T_1, T_2, T_3, T_4, T_5\}$ for the corresponding TPP instance with two tracks containing 1 link ($T_1 = \{ u_7u_8 \}$ and $T_2 = \{ v_5 u_9 \}$), two tracks containing 3 links ($T_3 = \{u_2w_1, u_1v_1, w_1 v_2\}$ and $T_4=\{v_9 w_{11}, u_{11} v_{11}, w_{11} v_{10}\}$) and one track containing 5 links ($T_5=\{u_5 w_4, u_4 v_4, w_4 w_3, u_3 v_3, w_3 u_6\}$).
    Note that $|Y| = 17 = 2 |\mcP| - 5 = 2|\mcP| - |X|$.}
	\label{fig:s2ecpc-tpp}
\end{figure}
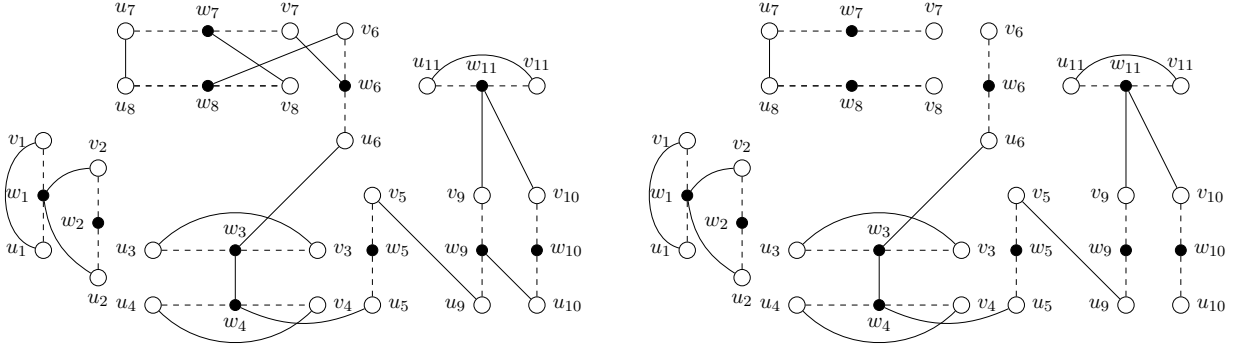

The input of an instance of 2ECPC consists of a structured instance of PAP where each path $P_i \in \mcP$ has exactly three vertices: two endpoints $u_i, v_i \in V^*(P_i)$ and an interior vertex $w_i$. 
The goal is to compute a minimum-cardinality link set $X \subseteq L$ such that:
\begin{itemize}
\item Every endpoint $u \in V^*(P_i)$ is incident to at least one link in $X$, i.e., $\delta_X(u) \geq 1$.
\item Each path $P_i \in P$ has at least two outgoing edges in $X$, i.e., $\delta_X(P_i) \geq 2$.
\end{itemize}
In other words,  $E(\mcP) \cup X$ is a 2-edge cover with the additional constraint that each path $P_i$ is incident to at least two outgoing links in $X$.
One can easily obtain an instance of \tecpc from an instance of $\PAP$ by doing the following for each path: 
if the path contains at least 4 vertices, we contract all internal vertices to a single vertex, and if the path contains 2 vertices, we add a dummy vertex as an inner vertex.
There is a one-to-one correspondence between links in the graph of the instance of \PAP and links in the graph of the instance of \tecpc.
We denote by $\OPT$ and $\OPT_C$ the optimum solutions for the \PAP and the \tecpc instance, respectively, and define $\opt$ and $\opt_C$ to be their objective values.
One can verify that any feasible solution to the \PAP instance is also feasible for its corresponding \tecpc instance (i.e., $\opt_C \leq \opt$), and hence \tecpc is indeed a relaxation for \PAP.

Similar to other connectivity augmentation problems, we assume that the instance of \tecpc is \emph{shadow-complete}, allowing us to assume that certain \emph{weaker} links (shadows) are present, implying more structure on our instance. Adding these shadows does not affect feasibility or the value of an optimum solution to 2ECPC. 
Instead of computing an (approximate) solution to \tecpc, we consider a packing problem, which we call the Track Packing Problem (\tpp).

\paragraph{The Track Packing Problem.}
TPP is motivated by the relationship between 2-matchings and 2-edge covers, where the covering constraints are expressed as a packing problem.
Our goal is to obtain a similar packing problem from the covering problem 2ECPC.
Formally, we construct an instance of \tpp from a shadow-complete instance of \tecpc.
For some link set $A \subseteq L$ let $V^*(A) = V^*(\mcP) \cap A$ be the set of vertices that are endpoints of some path and are incident to some link in $A$.
We define a \emph{track} as follows.
A track $T \subseteq L$ with $T = Q(T) \cup I(T)$ is a subset of links such that the following conditions hold.
\begin{itemize}
    \item $Q(T)$ is a simple path in $G$  of length at least one, where the start vertex $u$ and the end vertex $v$ of $Q(T)$ are endpoints of some paths $P, P' \in \mcP$ (not necessarily $P \neq P'$) such that $u \neq v$,
    \item any interior vertex of $Q(T)$, i.e., a vertex from $V(Q(T)) \setminus \{u, v \}$, is the interior vertex of some path $P \in \mcP$, 
    \item $I(T) \subseteq L$ contains exactly the following set of links: for each interior vertex $w_i \in V(Q(T)) \setminus \{u, v \}$, which is an interior vertex of some path $P_i \in \mcP$, 
    there is a link $u_i v_i \in I(T)$, where $u_i, v_i \in V^*(P_i)$ are the endvertices of $P_i$, and
    \item if $|T| = 1$, i.e., if $|Q(T)|=1$ and $|I(T)| = 0$, then the start vertex $u$ of $Q(T)$ and the end vertex $v$ of $Q(T)$ belong to distinct paths.
\end{itemize}
For any track $T$ we refer by $Q(T)$ and $I(T)$ as the set of links as above.
Observe that the number of links in a track is odd, since $|Q(T)| = |I(T)| + 1$. 
Further, note that a track of length 1 is simply a link in $\bar{L}$, i.e., a link between two endvertices of two distinct paths. 

The instance of \tpp is now defined as follows.
The input is a (shadow-complete) instance of 2ECPC of some structured graph $G$.
Let $\mathcal{T}$ be the set of all tracks in $G$.
The task is to select a maximum number of disjoint tracks from $\mathcal{T}$, where two tracks $T, T'$ are disjoint if $V(T) \cap V(T') = \emptyset$.

We establish a close connection between solutions of \tecpc and solutions of the corresponding \tpp instance, which is similar to the connection between 2-matchings and 2-edge covers. In particular, an optimum solution to one of one of the problems can be turned into an optimum solution to the other problem in polynomial time. 
The intuition behind the relationship of TPP and 2ECPC is demonstrated in \Cref{fig:s2ecpc-tpp}.

We will later show that TPP is NP-hard, and by the above close relationship, this also implies
NP-hardness for 2ECPC. 
Despite this hardness, in the remaining part we show how to compute a solution to TPP which serves as a good starting solution for our algorithmic framework.

\paragraph*{Computing a good starting solution.}
Our goal is to compute a solution to \tpp. By its close relationship to \tecpc, we can bound the cost of this solution in terms $\opt_C$, and since \tecpc is a relaxation of \PAP, we can in turn bound its cost in terms of $\opt$.
We do not compute an approximate solution to \tpp relative to its optimum solution value; instead we do the following:
For a solution $X$ to \tpp we compute $H_X = (V, E(\mcP) \cup (X \cap L))$ and we want to find a solution $X$ to \tpp such that $\cost(H_X) = |X \cap L| + \credits(H_X) \leq \apxr \cdot \opt$.

To find such a solution, we view \tpp as a Set-Packing problem.
In a Set-Packing problem, we are given a set of elements $U$ and a set of subsets $\mathcal{S} \subseteq 2^U$ of $U$ and the task is to compute a set $X \subseteq \mathcal{S}$ of maximum cardinality such that for any two sets $S, S' \in X$ we have $S \cap S' = \emptyset$.
The state-of-the-art polynomial-time approximation algorithms for Set-Packing achieve an approximation factor of $\frac{3}{1 + k + \varepsilon}$ for any constant $\varepsilon > 0$, where $k \coloneqq \max_{S \in \mathcal{S}} |S|$ is the maximum size of a set in $\mathcal{S}$~\cite{cygan2013improved, furer2014approximating}.
\tpp can be reduced to Set-Packing by viewing each $v \in V^*(\mcP)$ as an element and each track as a set, which exactly contains those vertices $v \in V^*(\mcP)$ to which the track is incident to.

Using such set-packing algorithms, we define three different algorithms that compute sets of disjoint tracks, where the maximum size of a track we consider contains at most 4 vertices of~$V^*(\mcP)$.
In total, we compute up to $|\mcP|+2$ different solutions and output the solution $X$ minimizing $\cost(H_X)$.
It can be easily shown that $H_X$ satisfies \Cref{invariant:degree-lonely-vertex} and \Cref{invariant:block-size}.
To show that $H_X$ satisfies \Cref{invariant:credit}, that is, $\cost(H_X) \leq \apxr \cdot \opt$, we bound the cost of $H_X$ in terms of variables that depend on the three algorithms and certain parameters of an optimum solution. 
We then derive a factor-revealing LP, which proves $\cost(H_X) \leq \rho \cdot \opt$ if the polyhedron is empty for a specific value of $\rho$. 
Indeed, for $\rho= \apxr$ this is the case, implying \Cref{invariant:credit}.

\subsection{Bridge-Covering}

\label{sec:overview:bridge-covering}

After computing a good starting solution, our next goal is to ensure that the current partial solution $H = (V, E(\mathcal{P}) \cup S)$ is bridgeless. 
This process, called \emph{bridge-covering}, eliminates all bridges in $H$ while preserving Invariants~\ref{invariant:credit}-\ref{invariant:block-size}.
In particular, we show the following lemma.

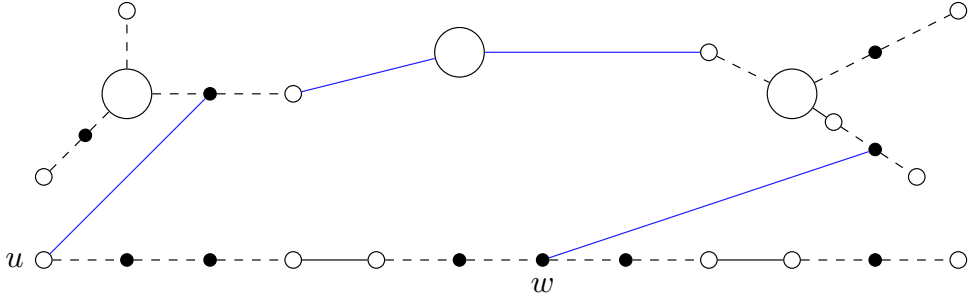
\begin{figure}[t]
	\begin{center}
 
\resizebox{.8\linewidth}{!}{
\begin{tikzpicture}
    \node[shape=circle,draw=black, inner sep=2pt, label = left:$u$] (A0) at (0,0) {};
    \node[shape=circle,draw=black, fill=black, inner sep=1.5pt] (A1) at (1,0) {};
    \node[shape=circle,draw=black, fill=black, inner sep=1.5pt] (A2) at (2,0) {};
    
    \node[shape=circle,draw=black, inner sep=2pt] (A3) at (3,0) {};
    \node[shape=circle,draw=black, inner sep=2pt] (A4) at (4,0) {};
    \node[shape=circle,draw=black, fill=black, inner sep=1.5pt] (A5) at (5,0) {};
    \node[shape=circle,draw=black, fill=black, inner sep=1.5pt, label = below:$w$] (A6) at (6,0) {};
    \node[shape=circle,draw=black, fill=black, inner sep=1.5pt] (A7) at (7,0) {};
    
    \node[shape=circle,draw=black, inner sep=2pt] (A8) at (8,0) {};
    \node[shape=circle,draw=black, inner sep=2pt] (A9) at (9,0) {};
    
    \node[shape=circle,draw=black, fill=black, inner sep=1.5pt] (A10) at (10,0) {};
    \node[shape=circle,draw=black, inner sep=2pt] (A11) at (11,0) {};

    \node[shape=circle,draw=black, inner sep=2pt] (B1) at (0,1) {};
    \node[shape=circle,draw=black, fill=black, inner sep=1.5pt] (B2) at (0.5,1.5) {};
    
    \node[shape=circle,draw=black, inner sep=6pt] (L1) at (1,2) {};
    
    \node[shape=circle,draw=black, inner sep=2pt] (B3) at (1,3) {};
    
    \node[shape=circle,draw=black, fill=black, inner sep=1.5pt] (B4) at (2,2) {};
    \node[shape=circle,draw=black, inner sep=2pt] (B5) at (3,2) {};
    
    \node[shape=circle,draw=black, inner sep=6pt] (L2) at (5,2.5) {};

    \node[shape=circle,draw=black, inner sep=2pt] (C1) at (8,2.5) {};
    \node[shape=circle,draw=black, inner sep=6pt] (L3) at (9,2) {};
    \node[shape=circle,draw=black, fill=black, inner sep=1.5pt] (C2) at (10,2.5) {};
    \node[shape=circle,draw=black, inner sep=2pt] (C3) at (11,3) {};
    \node[shape=circle,draw=black, inner sep=2pt] (C4) at (9.5,1.66) {};
    \node[shape=circle,draw=black, fill=black, inner sep=1.5pt] (C5) at (10,1.33) {};
    \node[shape=circle,draw=black, inner sep=2pt] (C6) at (10.5,1) {};

    \path [dashed] (A0) edge node[left] {} (A1);
    \path [dashed] (A1) edge node[left] {} (A2);
    \path [dashed] (A2) edge node[left] {} (A3);
    \path [-] (A3) edge node[left] {} (A4);
    \path [dashed] (A4) edge node[left] {} (A5);
    \path [dashed] (A5) edge node[left] {} (A6);
    \path [dashed] (A6) edge node[left] {} (A7);
    \path [dashed] (A7) edge node[left] {} (A8);
    \path [-] (A8) edge node[left] {} (A9);
    \path [dashed] (A9) edge node[left] {} (A10);
    \path [dashed] (A10) edge node[left] {} (A11);
    
    \path [dashed] (B1) edge node[left] {} (B2);
    \path [dashed] (B2) edge node[left] {} (L1);
    \path [dashed] (L1) edge node[left] {} (B3);
    \path [dashed] (L1) edge node[left] {} (B4);
    \path [dashed] (B4) edge node[left] {} (B5);

    \path [dashed] (C1) edge node[left] {} (L3);
    \path [dashed] (L3) edge node[left] {} (C2);
    \path [-] (L3) edge node[left] {} (C4);
    \path [dashed] (C2) edge node[left] {} (C3);
    \path [dashed] (C4) edge node[left] {} (C5);
    \path [dashed] (C5) edge node[left] {} (C6);
    
    \path [-, color=blue!90!white] (A0) edge node[left] {} (B4);
    \path [-, color=blue!90!white] (B5) edge node[left] {} (L2);
    \path [-, color=blue!90!white] (L2) edge node[left] {} (C1);
    \path [-, color=blue!90!white] (C5) edge node[left] {} (A6);

\end{tikzpicture}
}
	\end{center}
 \vspace*{-0.5cm}
	\caption{Example for the bridge-covering process: a pseudo-ear covering the witness path from $u$ to $w$. $H$ consists of the dashed edges (in $E(\mathcal{P})$) and solid black edges (links of $H$). Filled vertices are interior vertices of paths from $\mcP$, empty vertices are endvertices of paths in $\mcP$. The big circles represent 2EC blocks or components. The blue edges form a pseudo-ear $Q^{uw}$ from $u$ to $w$. All bridges, vertices and blocks on the path from $u$ to $w$ in $H$ will be in a 2EC block in $H \cup Q^{uw}$. }
	\label{fig:bridge-covering}
\end{figure}

\begin{restatable}[]{lemma}{lembridgecoveringmain}
    \label{lem:bridge-covering:main}
    Let $G$ be some structured instance of \PAP and $H=(V, E(\mcP) \cup S)$ be a solution satisfying the Invariants~\ref{invariant:credit}-\ref{invariant:block-size}.
    In polynomial-time we can compute a solution $H'=(V, E(\mcP) \cup S')$ that satisfies Invariants~\ref{invariant:credit}-\ref{invariant:block-size} and contains no bridges.
\end{restatable}

Throughout this section we assume that $H$ contains some complex component $C$, as otherwise the lemma holds trivially.
We prove \Cref{lem:bridge-covering:main} by showing that we can iteratively turn $H$ into a solution $H'$ such that $H'$ also satisfies the Invariants~\ref{invariant:credit}-\ref{invariant:block-size} and additionally $H'$ contains fewer bridges than $H$ (except for one case that we need for technical reasons).
Applying this procedure exhaustively proves \Cref{lem:bridge-covering:main}.

Intuitively, as long as our current solution $H$ contains complex components, i.e., it contains bridges, our goal is to add certain links to $H$ in order to 'cover' the bridges and, at the same time, maintain the Invariants~\ref{invariant:credit}-\ref{invariant:block-size}.
To do so, it will be convenient to add so-called pseudo-ears to $H$, defined as follows. Similar definitions have been used in~\cite{cheriyan2023improved} and~\cite{garg2023matching}.

\begin{restatable}[Pseudo-ear and Witness Path]{definition}{defpseudoear}
    Let $C$ be a complex component of $H$ and $u, v \in V(C)$ be distinct vertices in $G^C$. 
    A \emph{pseudo-ear} $Q^{u v}$ from $u$ to $v$ is a simple path in $G^C$ that starts in $u$ and ends in $v$ and does not visit any vertex of $C$ except $u$ and $v$.
    The unique path $W^{uv}$ from $u$ to $v$ in $H^C$ is called the \emph{witness} path of $Q^{uv}$.
    Note that the edges of $W^{uv}$ consist of bridges of $C$.
\end{restatable}

An example for a pseudo-ear is given in \Cref{fig:bridge-covering}.
We first show that a pseudo-ear $Q^{uv}$ can be found in polynomial time, if one exists.
Furthermore, pseudo-ears are useful in the following sense.
If, for some pseudo-ear $Q^{uv}$, its witness path $W^{uv}$ contains 'enough' credit, e.g., if $W^{uv}$ contains 2 block-nodes or 2 lonely vertices, then we can add $Q^{uv}$ to $H$ to obtain the graph $H'$. 
Afterwards, $H'$ is a graph that contains fewer bridges than $H$, the number of connected components in $H'$ is at most the number of connected components in $H$, and $H'$ satisfies all invariants.

Using pseudo-ears, we then show that in polynomial time we can find link sets $R \subseteq L \setminus S$ and $F \subseteq S$, such that $(H \cup R) \setminus F$ has fewer bridges than $H$ and also satisfies all invariants.
This then proves \Cref{lem:bridge-covering:main}.
It is crucial for this step that $G$ is structured: We usually argue that $G$ must have a desired pseudo-ear, as otherwise the absence of it implies that $G$ is not structured, a contradiction.

\subsection{Gluing Algorithm}

\label{sec:overview:gluing}
After the bridge-covering step, the graph $H = (V, P \cup S)$ consists of multiple 2-edge-connected components that remain disconnected from each other. The final step of our algorithm, called \emph{gluing}, ensures that these components are merged into a single 2-edge-connected graph while maintaining the approximation guarantee.
We show the following.

\begin{figure}[t]
	\begin{center}
 
\resizebox{.5\linewidth}{!}{
\begin{tikzpicture}

    \node[shape=circle,draw=black, inner sep=6pt] (L1) at (0,1) {};
    
    \node[shape=circle,draw=black, inner sep=6pt] (L2) at (1,3) {};
    
    \node[shape=circle,draw=black, inner sep=6pt] (L3) at (6,1) {};
    
    \node[shape=circle,draw=black, inner sep=2pt] (A1) at (2,0) {};
    \node[shape=circle,draw=black, fill=black, inner sep=1.5pt] (A2) at (2,0.5) {};
    \node[shape=circle,draw=black, fill=black, inner sep=1.5pt] (A3) at (2,1) {};
    \node[shape=circle,draw=black, inner sep=2pt] (A4) at (2,1.5) {};
    
    \node[shape=circle,draw=black, inner sep=2pt] (B1) at (3,0) {};
    \node[shape=circle,draw=black, fill=black, inner sep=1.5pt] (B2) at (3,0.75) {};
    \node[shape=circle,draw=black, inner sep=2pt] (B3) at (3,1.5) {};

    \node[shape=circle,draw=black, inner sep=2pt] (C1) at (3.5,2.5) {};
    \node[shape=circle,draw=black, fill=black, inner sep=1.5pt] (C2) at (4,2.5) {};
    \node[shape=circle,draw=black, fill=black, inner sep=1.5pt] (C3) at (4.5,2.5) {};
    \node[shape=circle,draw=black, inner sep=2pt] (C4) at (5,2.5) {};
    
    \node[shape=circle,draw=black, inner sep=2pt] (D1) at (3.5,3.5) {};
    \node[shape=circle,draw=black, inner sep=2pt] (D2) at (5,3.5) {};

    \path [dashed] (A1) edge node[left] {} (A2);
    \path [dashed] (A2) edge node[left] {} (A3);
    \path [dashed] (A3) edge node[left] {} (A4);
    
    \path [dashed] (B1) edge node[left] {} (B2);
    \path [dashed] (B2) edge node[left] {} (B3);
    
    \path [dashed] (D1) edge node[left] {} (D2);
    
    \path [dashed] (C1) edge node[left] {} (C2);
    \path [dashed] (C2) edge node[left] {} (C3);
    \path [dashed] (C3) edge node[left] {} (C4);

    \path [-, color=blue!90!white] (L1) edge node[midway, label= left:\textcolor{black}{\small{$e_5$}}] {} (L2);
    \path [-, color=blue!90!white] (L1) edge node[midway, label= below:\textcolor{black}{\small{$e_6$}}] {} (A1);
    
    \path [-, color=blue!90!white] (L2) edge node[midway, label= left:\textcolor{black}{\small{$e_4$}}] {} (A4);
    \path [-, color=blue!90!white] (L2) edge node[midway, label= above:\textcolor{black}{\small{$e_8$}}] {} (D1);
    \path [-, color=blue!90!white] (L2) edge node[midway, label= below:\textcolor{black}{\small{$e_1$}}] {} (C1);
    \path [-, color=blue!90!white] (C3) edge node[midway, label= below:\textcolor{black}{\small{$e_3$}}] {} (L3);
    \path [-, color=blue!90!white] (B1) edge node[midway, label= above:\textcolor{black}{\small{$e_2$}}] {} (L3);
    \path [-, color=blue!90!white] (A1) edge node[midway, label= below:\textcolor{black}{\small{$g$}}] {} (B3);
    
    \path [-, color=blue!90!white] (B1) edge node[midway, label= right:\textcolor{black}{\small{$e_7$}}] {} (C1);

    \path [-] (A1) edge node[midway] {} (B1); 
    \path [-] (A4) edge node[midway] {} (B3); 
    \path [-] (C1) edge node[midway] {} (D1);
    \path [-] (C4) edge node[midway] {} (D2);

\end{tikzpicture}
}
	\end{center}
 \vspace*{-0.5cm}
	\caption{Example for the gluing algorithm and a good cycle. The bridgeless 2-edge-cover consists of large components (big circles) and small components (short cycles). $H$ consists of the dashed edges (in $E(\mathcal{P})$) and solid black edges (links of $H$). Blue edges are edges in $E(G) \setminus E(H)$. The cycle $e_4 e_5 e_6$ is a good cycle in the components graph as it contains 2 large components. The cycle $e_4 e_1 e_7$ is also a good cycle as it shortcuts the lower small components by using the Hamiltonian Path through~$g$.}
	\label{fig:gluing}
\end{figure}
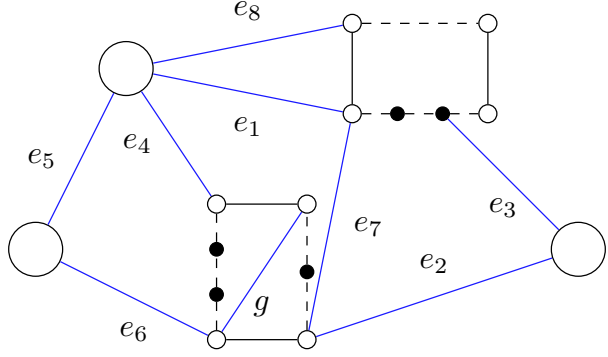

\begin{restatable}[]{lemma}{lemgluingmain}
    \label{lem:gluing:main}
    Let $G$ be some structured instance of \PAP and $H=(V, E(\mcP) \cup S)$ be a solution satisfying the Invariants~\ref{invariant:credit}-\ref{invariant:block-size} that contains no bridges.
    In polynomial-time we can compute a solution $H'=(V, E(\mcP) \cup S')$ satisfying the Invariants~\ref{invariant:credit}-\ref{invariant:block-size} such that $H'$ is a 2-edge-connected spanning subgraph.
\end{restatable}

Throughout this section we are given a solution $S$ such that $H=(V, E(\mcP) \cup S)$ consists of several 2-edge-connected components, and each component contains at least the vertex set of two distinct paths $P, P' \in \mcP$. 
According to the credit scheme (part (E)), each component has a credit of 2 if it is large or it is small and contains a degenerate path. Otherwise, it has a credit of $\frac{3}{2}$.

Our approach here is to find \emph{good} cycles in the component graph $\Tilde{G}$, that is, the graph arising from $G$ by contracting each 2-edge-connected component of $H$ to a single vertex.
Adding cycles from the component graph to $H$ reduces the number of components and does not introduce any bridges.
By iteratively adding such cycles, we eventually end up with a single 2-edge-connected component and we are done.

Good cycles are cycles that either contain 2 components with credit 2 each, or a small component that can be \emph{shortcut}.
A small component $C$ can be shortcut w.r.t.\ a cycle $K$ of $\Tilde{G}$ if we can replace $E(C)$ by a Hamiltonian path $Q \subseteq G[V(C)]$ such that $|Q \cap L| < |E(C) \cap L|$ and the vertices of $C$ that are incident to $K$ are the endvertices of $Q$.
If a good cycle contains a shortcut, we then update $H$ to $H' = (H \setminus E(C)) \cup K \cup Q$ and refer to this operation as \emph{augmenting} $H$ with $K$.
Two examples of good cycles can be found in \Cref{fig:gluing}.
We then show that 1) augmenting $H$ with a good cycle results in a graph with fewer components that is bridgeless and satisfies all invariants, 2) we can find good cycles in polynomial time if one exists, and 3) the component graph of a structured graph always admits a good cycle.
This then proves \Cref{lem:gluing:main}.

\subsection{Proof of \Cref{thm:PAP:main}}

\label{sec:overview:proof-thm1}
\Cref{thm:PAP:main} now easily follows from Lemmas~\ref{lem:structured:main}-\ref{lem:gluing:main}.

\begin{proof}[Proof of \Cref{thm:PAP:main}.]
  For structured graphs $G$, we first use \Cref{lem:main:starting-solution-satisfies-invariants} to compute a solution $S$ such that $H= (V, E(\mcP) \cup S)$ satisfies Invariants~\ref{invariant:credit}-\ref{invariant:block-size}.
  \Cref{lem:bridge-covering:main} now guarantees that we can transform $H$ into a solution $H'$ that is bridgeless and satisfies Invariants~\ref{invariant:credit}-\ref{invariant:block-size}.
  Finally, \Cref{lem:gluing:main} ensures that that we can obtain in polynomial time a 2-edge-connected spanning subgraph $H'' = (V, E(\mcP) \cup S'')$ satisfying Invariants~\ref{invariant:credit}-\ref{invariant:block-size}.
  In particular, $H''$ is a feasible solution and satisfies $\cost(H'') = |S''| + \credits(H'') \leq \apxr \cdot \opt$. 
  Since $\credits(H'') \geq 0$, we have $|S''| \leq \apxr \cdot \opt$, implying that our algorithm is a polynomial-time $\apxr$-approximation for structured graphs.
  Finally, \Cref{lem:structured:main} then guarantees that this leads to a $(\apxr + 3 \varepsilon)$-approximation for arbitrary instances.
  However, we mention here that the guarantee of $\apxr$ from \Cref{lem:main:starting-solution-satisfies-invariants} has some slack, and hence we can choose $\varepsilon$ small enough (still constant) so that we have a $\apxr$-approximation for \PAP.
\end{proof}

The four subsequent sections provide the proofs of Lemmas~\ref{lem:structured:main}-\ref{lem:gluing:main}.

\section{Reduction to Structured Instances}

\label{sec:structured}

In this section we prove \Cref{lem:structured:main}, which we restate for convenience.

\lemstructuredmain*

We outlined the main definitions and proof of \Cref{lem:structured:main} in \Cref{sec:overview:preprocessing}. 
For completeness, we repeat the definitions given in \Cref{sec:overview:preprocessing}.
Throughout this section we assume that $\varepsilon \in (0, \frac{1}{12}]$ is some fixed constant.
In order to define structured graphs, we need the following definitions.
We will show that any structured graph does not contain any of these subgraphs. 

\defContractibleSubgraph*

Note that every strongly $(\alpha, t , k)$-contractible subgraph is also a weakly $(\alpha, t , k)$-contractible subgraph.
We say that $\mathcal{H}$ is an $(\alpha, t , k)$-contractible subgraph if it is a weakly or strongly $(\alpha, t , k)$-contractible subgraph.
If $\alpha, t, k$ are clear from the context, we simply refer to $\alpha$-contractible subgraphs or even contractible subgraphs.

We use $(\alpha, t , k)$-contractible subgraphs in the following sense: We can safely contract each component of an $\alpha$-contractible subgraph $\mathcal{H}$ to a single vertex and work with $G|\mathcal{H}$ (if we only aim for an $\alpha$-approximation). 
Indeed, any $\alpha$-approximate solution for $G|\mathcal{H}$ can be turned to an $\alpha$-approximate solution for $G$ by adding the links of $\mathcal{H}$.

In general it is not straightforward to argue that a 2EC subgraph is weakly $\alpha$-contractible.
However, under certain conditions we can argue that this is indeed the case.
We will later show that we can find strongly $(\alpha, t , 1)$-contractible subgraphs in polynomial time if $t$ is constant, i.e., if there is only one component and there are only a constant number of links in the contractible subgraph. 
Hence, after our preprocessing we can assume that $G$ does not contain any $(\alpha, t , 1)$-contractible subgraphs for some constant $t$.

We now define certain separators.

\begin{restatable}[Path-separator]{definition}{defPathSeparator}
    A path $P \in \mathcal{P}$ is called a path-separator if $P$ is a separator, i.e., $G \setminus V(P)$ is disconnected.
\end{restatable}

\begin{restatable}[$P^2$-separator]{definition}{defPTwoSeparator}
    Let $P_1, P_2\in \mathcal{P}$ be two distinct paths with endpoints $v_1, w_1$ and $v_2, w_2$, respectively, such that $w_1 v_2 \in L$ and assume that $V(P_1) \cup V(P_2)$ is a separator.
    Let $Q$ be the shortest (w.r.t.\ the number of edges) sub-path of $P_1 \cup P_2 \cup \{ w_1 v_2\}$ such that $Q$ is a separator.
    We call $Q$ a $P^2$-separator if there is a partition $(V_1, V_2)$ of $G \setminus V(Q)$ such that there are no edges between $V_1$ and $V_2$ and either $\opt((G \setminus V_i)|Q) \geq 4$ for $i = 1, 2$ or $\opt((G \setminus V_2)|Q) \geq 4$, $\opt((G \setminus V_1)|Q) \geq 3$ and $|\delta(Q) \cap E(\mathcal{P})| \geq 1$. 
\end{restatable}
Note that by the definition of a $P^2$-separator $Q$, if the last condition $|\delta(Q) \cap E(\mathcal{P})| \geq 1$ is satisfied, this means that there is an edge $e \in E(\mathcal{P})$ such that $e \in \delta(Q)$ and $e$ is incident to an endpoint of $Q$.

\begin{restatable}[$C^2$-separator]{definition}{defCTwoSeparator}
    Let $P_1, P_2\in \mathcal{P}$ be two distinct paths with endpoints $v_1, w_1$ and $v_2, w_2$, respectively, and assume $w_1 v_2 \in L$ and $w_2 v_1 \in L$.
    Let $V(P_1) \cup V(P_2)$ be a separator. 
    We call $C \coloneqq P_1 \cup P_2 \cup w_1 v_2 \cup w_2 v_1$ a $C^2$-separator if there is a partition $(V_1, V_2)$ of $G \setminus V(C)$ such that there are no edges between $V_1$ and $V_2$ and $\opt((G \setminus V_i)|C) \geq 3$ for $i = 1, 2$.
    Observe that $C$ is a cycle with exactly two links visiting all vertices of $V(P_1) \cup V(P_2)$.
\end{restatable}

Whenever $G$ contains one of these separators, we show that we can essentially divide $G$ into two subinstances $G_1, G_2$ of smaller size such that we can combine solutions to these subinstances to a solution to $G$ by adding only few edges. 

Next, we define degenerate paths.

\defDegeneratePathSeparator*

We outlined the ideas used to handle degenerate paths in \Cref{sec:overview:preprocessing}.
Note that the condition $v' w' \in L$ is always satisfied for a degenerate path if $G$ does not contain contractible subgraphs, as otherwise $C= (V(P') \cup V( P), E(P') \cup E( P) + w' v + w v')$ is a $(1, 2, 1)$-contractible subgraph.

Now we are ready to define $(\alpha, \eps)$-structured instances.

\defStructuredGraph*

If $\alpha$ and $\varepsilon$ are clear from the context, we simply say $G$ is structured.
Examples of different separators, contractible subgraphs and degenerate paths are given in \Cref{fig:structured}.

 The remaining part of this section is dedicated to proving \Cref{lem:structured:main}.

\subsection{Important Facts, Definitions and Lemmas}

It will be convenient to use the following facts and lemmas throughout this section.
Similar facts have been introduced in~\cite{garg2023improved}.

\begin{fact}
\label{fact:structured:2EC-after-contraction}
    Let $G$ be a 2EC graph and $U \subseteq V(H)$. Then $G|U$ is 2EC.
\end{fact}

\begin{fact}
\label{fact:structured:2EC-contraction+2EC}
    Let $H$ be a 2EC subgraph of some 2EC graph $G$. Let $S$ and $S'$ be 2EC spanning subgraphs of $H$ and $G|H$, respectively. Then $S \cup S'$ is a 2EC spanning subgraph of $G$, where we interpret $S'$ as edges of $G$ after decontracting $V(H)$.
\end{fact}

We now show that we can find strongly $(\alpha, t , 1)$-contractible subgraphs in polynomial time if $t$ is constant, i.e., if there is only one component and there are only a constant number of links in the contractible subgraph. 

\begin{lemma}
    \label{lem:structured:finding-small-contractible-polytime}
    If $G$ contains a strongly $(\alpha, t , 1)$-contractible subgraph $H$ for some constant $t$, then we can find $H$ in polynomial-time.
\end{lemma}

\begin{proof}
    We first describe the algorithm and then prove that it is indeed correct.
    In the algorithm, we first enumerate all polynomially many $L' \subseteq L$ of size at most $t$.
    For some given $L'$ we give a polynomial-time algorithm that checks whether a certain subgraph induced by $L'$ is in fact strongly $(\alpha, t , 1)$-contractible. This then proves the lemma.
    
    Let us fix one such~$L'$.
    We first check whether $G|L'$ is an instance of \PAP, which can be checked efficiently.
    We assume that, for each $P \in \mathcal{P}$, $L'$ is incident to $0$ or at least $2$ vertices of $P$, as otherwise $L'$ together with any subset of edges of $\mathcal{P}$ can not be 2-edge-connected (as it contains bridges).
    For $P \in \mathcal{P}$ with endpoints $u$ and $w$ let $u' \in V(L') \cap V(P)$ be the vertex closest to $u$ on $P$ that is incident to some link in $L'$, if such a vertex exist. 
    Similarly, let $w' \in V(L') \cap V(P)$ be the vertex closest to $w$ on $P$ that is incident to some link in $L'$, if such a vertex exist.  
    We define $L'_P$ to be the sub-path of $P$ from $u'$ to $w'$ and $L'_\mathcal{P} \coloneqq \bigcup_{P \in \mathcal{P}} L'_P$.
    Let $G_{L'} \coloneqq (V(L'_\mathcal{P}), L'_\mathcal{P} \cup L')$.
    
    If $G_{L'}$ is a 2EC spanning subgraph, $G_{L'}$ is a candidate to be a  strongly $(\alpha, t , 1)$-contractible subgraph $H$ for some constant $t$.
    To verify if this is indeed the case, we proceed as follows.
    Let $V' = V \setminus V(L'_\mathcal{P})$ be the vertices not contained in $G_{L'}$.
    Let $X \subseteq E$ be a maximum cardinality subset of edges with at least one endpoint in $V'$.
    If any feasible solution to $G$ contains at least $\frac{1}{\alpha} |L'|$ many links from within $G|V(X)$, then we know that $G_{L'}$ forms an $\alpha$-contractible subgraph.
    It remains to give a polynomial-time algorithm to decide if this is the case or not.
    Since $G_{L'}|V(X)$ is a 2EC spanning subgraph of $G|V(X)$ with a constant number of links, we know that $\opt(G|V(X)) \leq |L'| \leq t$ and thus constant. 
    Hence, we can compute $\opt(G|V(X))$ in polynomial-time.
    If $\opt(G|V(X)) \leq \frac{1}{\alpha} |L'|$, then $G_{L'}$ is strongly $(\alpha, t , 1)$-contractible; otherwise it is not.

    It remains to provide a proof that the algorithm above finds an $\alpha$-contractible subgraph $H$ with at most $t$ links, if such a subgraph exists.
    Let $H=(U, E(\mathcal{P'}) \cup L')$ be an $\alpha$-contractible subgraph with link set $L' \subseteq L$ of size at most $t$, where $E(\mathcal{P'}) \subseteq E(\mathcal{P})$ and $U \subseteq V(G)$.
    We assume that $L'$ is minimal, i.e., no proper subset of $L'$ together with some $Z \subseteq E(\mathcal{P})$ induces an $\alpha$-contractible subgraph.
    Furthermore, let $H$ be maximal in the sense that there is no edge set $Y \subseteq E(\mathcal{P}) \setminus E(\mathcal{P'})$  such that $(U \cup V(Y), Y \cup E(\mathcal{P'}) \cup L')$ is also a 2-edge-connected spanning subgraph.
    In particular, note that this implies that for any two vertices $u, v \in V(H)$ such that $u, v \in V(P)$ for some $P \in \mathcal{P}$, the subpath $P^{uv}$ of $P$ from $u$ to $v$ must be in $H$. 
    If not, then we could safely add $P^{uv}$ to $H$ such that the resulting graph is still 2-edge-connected, and therefore also $\alpha$-contractible.
    Hence, if $L'$ is enumerated in the above algorithm, then also $G_{L'} = H$ must be enumerated by the above algorithm.

    Finally, let $V' = V \setminus V(H)$ be the vertices not contained in $H$.
    Let $X \subseteq E$ be a maximum cardinality subset of edges with at least one endpoint in $V'$.
    Since $H$ is $\alpha$-contractible, any 2EC spanning subgraph of $G$ must contains at least $\frac{1}{\alpha} |L'|$ many links from within $G|V(X)$.
    Since the algorithm checks this, it will find correctly an $\alpha$-contractible subgraph if one exists.
\end{proof}

The next two lemmas give examples of $\frac{7}{4}$-contractible subgraphs, which we use later.

\begin{lemma}
\label{lem:structured:endpoint-adj-this-path}
    If $G$ contains an endpoint $u$ of some path $P \in \mathcal{P}$ such that $N(u) \subseteq V(P)$, then there is some edge $e = uv$ such that $(V(P^{uv}), E(P^{uv}) + e)$ is strongly $1$-contractible. 
    Moreover, $e$ can be computed in polynomial time.
\end{lemma}

\begin{proof}
Let $P$ and $u$ be as in the lemma statement.
Let the endpoints of $P$ be $u$ and $w$ and let $e=uv \in L$ be the link incident to $u$ such that $v$ is closest to $w$ on the path $P$.
We claim that $C=(V(P^{uv}), E(P^{uv}) + e)$ is $1$-contractible.
Note that any feasible solution must contain an edge incident to~$u$.
By our assumption that $N(u) \subseteq V(P)$ and the definition of $e$, we conclude that any feasible solution must contain at least one link from within $G[C]$.
Further, note that $G|C$ is an instance of \PAP.
Then $C$ is a 1-contractible subgraph and note that $e$ can be found in polynomial time.
\end{proof}

We assume that the requirements of \Cref{lem:structured:endpoint-adj-this-path} are not met, i.e., each endpoint $u$ of some path $P \in \mcP$ has an edge to some other path $P' \in \mcP$ with $P' \neq P$.

\begin{lemma}
\label{lem:structured:path}
    If $G$ contains a path $P \in \mathcal{P}$ such that $N(V(P)) \subseteq V(P) \cup V(P')$ for some $P \neq P' \in \mathcal{P}$ and both endpoints of $P$ have an edge to $V(P')$, then there is some edge set $R$ and $V' \subseteq V(P) \cup V(P')$ with $V(P) \subseteq V'$ such that $(V', R)$ is strongly $\frac{3}{2}$-contractible. Moreover, $R$ and $V'$ can be computed in polynomial time.
\end{lemma}
\begin{proof}
Let $P$ and $P'$ as in the lemma statement.
Let $u$ and $w$ be the endpoints of $P$ and let $a$ and $b$ be the endpoints of $P'$.
We distinguish two cases whether $u w \in L$ or not.
First, assume that $e = u w \in L$.
Let $x \in N(P)$ be some neighbor of $V(P)$ that is closest to $a$ and let $y \in N(P)$  be some neighbor of $V(P)$ that is closest to $b$. 
Let $e_x$ and $e_y$ be the corresponding edges to $x$ and $y$ from $V(P)$, respectively.
Now observe that $H \coloneqq (V(P) \cup V(P'^{xy}), P \cup P'^{xy} + e_x + e_y + e)$ is 2EC.
Furthermore, note that any optimal solution must include at least $2$ links from within $G[V(H)]$, as $V(P)$ must be connected to $V \setminus V(P)$ by at least $2$ links in $\OPT(G)$ and since $N(P) \subseteq V(P) \cup V(P'^{xy})$.
Further note that $G|H$ is an instance of \PAP.
Then $H$ is a $\frac{3}{2}$-contractible subgraph and we have found the desired sets $R$ and $V'$.

Second, assume that $u w \notin L$.
Let $x \in N(u + w)$ be some neighbor of $u$ or $w$ that is closest to $a$ and let $y \in N(u + w)$ be some neighbor of $u$ or $w$ that is closest to $b$. 
Let $e_x$ and $e_y$ be the corresponding edges to $x$ and $y$ from $u$ or $w$, respectively.
If both edges are outgoing from one of the two vertices $u$ or $w$, say $u$, then let $e'$ be an arbitrary edge from $v$ to $V(P')$ (which must exist).
Now observe that $H \coloneqq (V(P) \cup V(P'^{xy}), P \cup P'^{xy} + e_x + e_y + e')$ is 2EC.
Furthermore, note that any optimal solution must include at least $2$ links from within $G[V(H)]$, as $\OPT(G)$ must contain two links from $u$ and $w$ (since $uw \notin L$) and since $N(u+w) \subseteq V(P) \cup V(P'^{xy})$.
Observe that $G|H$ is an instance of \PAP.
Then $H$ is a $\frac{3}{2}$-contractible subgraph and we have found the desired sets $R$ and $V'$.
Note that all computations can be implemented to run in polynomial time.
\end{proof}

Next, we show that given a separator $X \subseteq V$ of $G$, we can find in polynomial time a certain partition of $G \setminus X$.

\begin{lemma}
    \label{lem:structured:finding-connected-components}
    Let $X \subseteq V$ be a path-separator, $P^2$-separator, or a $C^2$-separator of $G$. 
    If there is a partition of $G \setminus V(X)$ into components $(V_1, V_2)$ such that there are no edges between $V_1$ and $V_2$ in $G$ and $\opt((G \setminus V_1)|P) \geq k_1$ and $\opt ((G \setminus V_2) | P) \geq k_2$, $k_1, k_2 \in \{3, 4\}$, then we can find such a partition in polynomial time. 
\end{lemma}
\begin{proof}
    Let $C_1, ..., C_j$ be the connected components of $G \setminus V(X)$.
    Consider some $P \in \mathcal{P}$. If $v \in V(P)$ is contained in some component $C_i$, then so is $V(P)$ (except for vertices in $X$).
    Let $(V_1, V_2)$ be a partition of $G \setminus V(X)$ such that there are no edges between $V_1$ and $V_2$ in $G$.
    Note that if $u \in C_i$ is in $V_1$ (or $V_2$), then all $V(C_i)$ must be in $V_1$ (or $V_2$), as otherwise there would be an edge between $V_1$ and $V_2$.
    
    If there are two connected components $C_1$, $C_2$ such that $\opt(G[V(C_i) \cup X]|X) \geq 4$, $i = 1, 2$ (which we can check in polynomial time), then set $V'_1 = C_1$ and $V'_2 = \bigcup_{i=2}^{j} V(C_i)$, and observe that $(V_1', V_2')$ satisfies the conditions of the lemma.
    Let us now assume that there is exactly one component, say $C_1$, such that $\opt(G[V(C_1) \cup X]|X) \geq 4$.
    Also in this case we set $V_1' = C_1$ and $V_2' = \bigcup_{i=2}^{j} V(C_i)$, which must satisfy the requirements of the lemma (if such a partition exists).
    Hence, assume that $\opt(G[V(C_i) \cup X]|X) \leq 3$ for all $i \in [j]$.
    Then, in polynomial time we can order the connected components $C_1, ..., C_j$ such that $\opt(G[V(C_i) \cup X]|X)$ is non-increasing.
    It can be easily checked that $\opt(G[V(C_i) \cup X]|X) \geq 1$ for all $i \in [j]$ since $X$ is a path-separator, $P^2$-separator, or a $C^2$-separator of $G$, by assumption.
    
    Hence, if $j \geq 8$, we can set $V_1' = \bigcup_{i=1}^{j - 4} V(C_i)$ and $V_2' = \bigcup_{i=j-3}^{j} V(C_i)$, which is the desired partition.
    Else, assume that $j \leq 7$.
    By our previous observations, if there is a partition $(V_1, V_2)$ as in the lemma, then there must be some $I \subseteq [j]$ such that $V_1 = \bigcup_{i \in I} V(C_i)$ and $V_2 = \bigcup_{i \in [j] \setminus I} V(C_i)$.
    Since $j \leq 7$, we can find such a partition in polynomial time by enumeration.
\end{proof}

Next, we show that if there is a path separator $P$ for some $P \in \mathcal{P}$, a $P^2$-separator or a $C^2$-separator, then there is a link-set $F$ that 'glues' together the two solutions of the subproblems, i.e., the two solutions of the subproblem together with $F$ form a 2EC connected spanning subgraph of $G$.

\begin{lemma}
    \label{lem:structured:finding-F-1}
    Let $P \in \mathcal{P}$ be a separator of $G$ and let $P'$ be a minimum-length sub-path of $P$ such that $P'$ is a separator of $G$. 
    Assume that there is a partition of $V(G) \setminus V(P')$ into components $(V_1, V_2)$ such that there are no edges between $V_1$ and $V_2$ in $G[V(G) \setminus V(P')]$.
    Let $H_1$ be a 2EC spanning subgraph of $(G \setminus V_2)|P'$ and $H_2$ be a 2EC spanning subgraph of $(G \setminus V_1)|P'$.
    Then we can find an edge set $F \subseteq L$ with $|F| \leq 2$ such that $H = H_1 \cup H_2 \cup P \cup F$ is a 2EC spanning subgraph of $G$.
    \end{lemma}

\begin{proof}
    Consider $H'= H_1 \cup H_2 \cup P$. 
    Observe that $H'$ is a connected spanning subgraph of $G$ and that the bridges in $H'$ (if any) are a subset of the edges in $P'$.
    Now observe that $P'$ is minimal in the sense that no sub-path of $P'$ is a separator in $G$. 
    Hence, the endpoints $v, w$ of $P'$ must have edges to both $V_1$ and $V_2$.
    Let $e_1$ and $e_2$ be such edges from $V_2$ to $v$ and $w$, respectively.
    By setting $F = \{e_1, e_2 \}$ observe that $H \cup P \cup F$ is a 2EC spanning subgraph of $G$.
\end{proof}

\begin{lemma}
    \label{lem:structured:finding-F-2}
    Let $Q$ be a $P^2$-separator of $G$ with paths $P_1, P_2 \in \mathcal{P}$ and edge $e \in L$.
    Assume that there is a partition of $G \setminus V(Q)$ into components $(V_1, V_2)$ such that there are no edges between $V_1$ and $V_2$ in $E$.
    Let $H_1$ be a 2EC spanning subgraph of $(G \setminus V_2)|Q$ and $H_2$ be a 2EC spanning subgraph of $(G \setminus V_1)|Q$.
    Then we can find an edge set $F$ with $|F| \leq 2$ such that $H = H_1 \cup H_2 \cup P_1 \cup P_2 \cup F \cup \{ e \}$ is a 2EC spanning subgraph of $G$.
    Furthermore, if there is an edge $e \in E(\mathcal{P})$ with $e \in \delta(Q)$ and $e$ is incident to one of the endpoints of $Q$, then we can find an edge set $F$ with $|F| \leq 1$ such that $H = H_1 \cup H_2 \cup P_1 \cup P_2 \cup F \cup \{ e \}$ is a 2EC spanning subgraph of $G$.
    \end{lemma}

\begin{proof}
    The proof is similar to the proof of \Cref{lem:structured:finding-F-1}. 
    Consider $H'= H_1 \cup H_2 \cup P_1 \cup P_2 \cup \{ e \}$. 
    Observe that $H'$ is a connected spanning subgraph of $G$ and that the bridges in $H'$ (if any) are a subset of the edges in $Q$.
    Now observe that since $Q$ is a $P^2$-separator, $Q$ is minimal in the sense that no sub-path of $Q$ is a separator in $G$. 
    Hence, the endpoints $v, w$ of $Q$ must have edges to both $V_1$ and $V_2$.
    Let $e_1$ and $e_2$ be such edges from $V_2$ to $v$ and $w$, respectively.
    By setting $F = \{e_1, e_2 \}$ observe that $H \cup Q \cup F$ is a 2EC spanning subgraph of $G$.
    Furthermore, if one of the edges incident to $v$ or $w$, respectively, is an edge of $E(\mathcal{P})$, then this edge serves as one of the edges $e_1$ or $e_2$ from above and hence we only need one of these additional links.
\end{proof}

\begin{lemma}
    \label{lem:structured:finding-F-3}
    Let $Q$ be a $C^2$-separator of $C$ with paths $P_1, P_2 \in \mathcal{P}$ and edges $e_1, e_2 \in L$. 
    Assume that there is a partition of $G \setminus V(C)$ into components $(V_1, V_2)$ such that there are no edges between $V_1$ and $V_2$ in $E$.
    Let $H_i$ be a 2EC spanning subgraph of $G_i|C$ for $i \in \{1, 2\}$.
    Then $H = H_1 \cup H_2 \cup C$ is a 2EC spanning subgraph of $G$.
    \end{lemma}

\begin{proof}
    Consider $H= H_1 \cup H_2 \cup C$. 
    Observe that $H$ is a connected spanning subgraph of $G$.
    Furthermore, since $C$ is a cycle, \Cref{fact:structured:2EC-contraction+2EC} implies that $H$ is 2EC.
    Hence, $H$ is a 2EC spanning subgraph of~$G$.
\end{proof}

\subsection{The Algorithm}

We are now ready to describe the reduction that is used for \Cref{lem:structured:main}, which is formally given in \Cref{alg:reduce}.
We assume that we are given a polynomial-time $\alpha$-approximation algorithm \ALG for structured instances of \PAP.
We consider the recursive procedure \RED described in \Cref{alg:reduce}, together with the subroutines presented in Algorithms~\ref{alg:handle-path-separator}-\ref{alg:handle-degenerate-paths}.
Given an instance $G = (V, \mathcal{P} \cup L)$ of \PAP, where $E = E(\mathcal{P}) \cup L$, \RED returns a solution $\RED(G) \subseteq E$ to \PAP on $G$.
Our algorithm \RED uses \ALG as a black-box subroutine whenever the instance is structured.
Whenever we consider edges from a contracted subgraph and then 'decontract' this subgraph, we interpret these edges as edges from the original graph.
We will show that whenever one of the if-conditions of \Cref{alg:reduce} is satisfied, the returned subgraph is a 2EC spanning subgraph.

The algorithm proceeds as follows: 
In the first step of the algorithm in \Cref{line:alg:isolated:nodes} we obtain an instance without isolated nodes by replacing each isolated node with a path of length 2. This step is only used to simplify notation in this section and the subsequent sections.
If the optimal solution of $G$ uses only constant many links (which we can check in polynomial time), we simply output the optimal solution.
Furthermore, if $G$ contains a $(\frac 74, t, 1)$-contractible subgraph $H$ for some constant $t$ (which we can also check in polynomial time), then it is safe to contract $H$ and consider $G|H$.
If $G$ contains one of the 'forbidden' structures, i.e., a path separator $P$ for some $P \in \mathcal{P}$, a $P^2$-separator, a $C^2$-separator or a set of degenerate paths of size at least $\eps \cdot |\mathcal{P}|$ (which we can also check in polynomial time), then we consider two disjoint subproblems of \PAP and 'glue' them together with only a few additional links.

In the next subsection, we show that the overall algorithm has a polynomial running-time for any constant $\varepsilon >0$ given that \ALG has a polynomial running-time.  
Furthermore, we show that the returned subgraph is an $(\alpha + 3 \eps)$-approximation for \PAP whenever \ALG returns an $\alpha$-approximation for structured instances of \PAP, as long as $\alpha \geq \frac 74$.
Finally, we prove that the output graph $G'$ in Line~16 has the desired properties.

\begin{algorithm}
\DontPrintSemicolon
\caption{Reduce}
\label{alg:reduce}
\If{$G$ contains isolated nodes}{
apply the reduction given in Lemma~29 of the full version of~\cite{grandoni2022breaching} to obtain an instance $G$ without isolated nodes: 
Replace each isolated node $v \in V$ with two new vertices $v_1, v_2$, add the edge $v_1 v_2$ to $\mathcal{P}$ as a new path and for each $u v \in L$ add $u v_1$ and $u v_2$ to $L$. \label{line:alg:isolated:nodes}
}
\If{$|\OPT(G)| \leq \frac{3}{\eps}$}{
compute $\OPT(G)$ by brute force and \Return $\OPT(G)$ \label{reduce:bruteforce}
\label{line:alg:reduce:constant}
}
\If{$G$ contains a strongly $(\frac 74, t, 1)$-contractible subgraph $H$ for $t \leq 10$
}{
\Return $\RED(G|H) + E(H)$
\label{line:alg:reduce:contractible}
}
\If{$G$ contains a path $P \in \mathcal{P}$ that is a separator such that $G \setminus V(P)$ has a partition into $(V_1, V_2)$ (according to \Cref{lem:structured:finding-connected-components}) with $\opt((G \setminus V_1)|P) \geq \opt ((G \setminus V_2) | P) \geq 3$}{
\Return output of \Cref{alg:handle-path-separator} (Handle Path Separator) with input $G, P, V_1, V_2$.
\label{line:alg:reduce:main:path-separator}
}
\If{$G$ contains a $P^2$-separator $Q$ with paths $P_1, P_2 \in \mathcal{P}$ and link $e \in L$ such that $G \setminus V(Q)$ has a partition into $(V_1, V_2)$ (according to \Cref{lem:structured:finding-connected-components}) with $\opt((G \setminus V_1)|Q) \geq \opt ((G \setminus V_2) | Q)$}{
    \Return output of \Cref{alg:handle-P2-separator} (Handle $P^2$ Separator) with input $G, Q, V_1, V_2$.
    \label{line:alg:reduce:main:P2-separator}
}
\If{$G$ contains a $C^2$-separator $C$ with paths $P_1, P_2 \in \mathcal{P}$ and links $e_1, e_2 \in L$ such that $G \setminus V(C)$ has a partition into $(V_1, V_2)$ (according to \Cref{lem:structured:finding-connected-components}) with $\opt((G \setminus V_1)|C) \geq \opt ((G \setminus V_2) | C) \geq 3$}{
    \Return output of \Cref{alg:handle-C2-separator} (Handle $C^2$ Separator) with input $G, C, V_1, V_2$.
    \label{line:alg:reduce:main:C2-separator}
}
\If{$G$ contains a collection $\mathcal{P'} \subseteq \mathcal{P}$ of degenerate paths such that $|\mathcal{P'}| \geq \eps \cdot |\mathcal{P}|$}{
    \Return output of \Cref{alg:handle-degenerate-paths} (Handle Degenerate Paths) with input $G, \mathcal{P'}$.
    \label{line:alg:reduce:main:degenerate-path-separator}
}
\Else{
\Return $\ALG(G)$
\label{line:alg:reduce:main:ALG}
}
\end{algorithm}

\begin{algorithm}
\DontPrintSemicolon
\caption{Handle Path Separator}
\label{alg:handle-path-separator}
 \textbf{Input:} Instance $G$, path $P$, sets $V_1, V_2$ \\
    Let $P' \subseteq P$ be the shortest sub-path of $P$ such that $P'$ is a separator in $G$ and let $P_i \subseteq P$, $i \in \{1, 2\}$ be the longest sub-path of $P$ such that the endpoints of $P_i$ are connected to $V_i$ in $G$. Let $G_i \coloneqq G[V_i \cup P_i]$.\\
    \If{ $\opt(G_i|P') \geq \frac{1}{\eps}$ for $i \in \{1, 2\}$}{
        Let $H'_i = \RED(G_i|P')$ for $i \in \{1, 2\}$,
        let $F' \subseteq L$, $|F'| \leq 2$ such that $H' \coloneqq H'_1 \cup H'_2 \cup P \cup F'$ is 2EC.   
    \label{line:alg:reduce:handle:P-separator-1}
    \Return $H'$
    }
    \Else{
        Let $H_1' \coloneqq \OPT(G_1|P')$, $H'_2 = \RED(G_2|P')$ and let $F' \subseteq L$, $|F'| \leq 2$ such that $H' \coloneqq H'_1 \cup H'_2 \cup P \cup F'$ is 2EC.   
    \label{line:alg:reduce:handle:P-separator-2}
    \Return $H'$
    }  

\end{algorithm}

\begin{algorithm}
\DontPrintSemicolon
\caption{Handle $P^2$ Separator}
\label{alg:handle-P2-separator}
 \textbf{Input:} Instance $G$, $P^2$-separator $Q$, sets $V_1, V_2$ \\
    Let $G_1 \coloneqq G[V \setminus V_2]$ and $G_2 \coloneqq G[V \setminus V_1]$.\\
    \label{line:alg:reduce:P2-separator}
    \If{ $\opt(G_i|Q) \geq \frac{1}{\eps}$ for $i \in \{1, 2\}$}{
        Let $H'_i = \RED(G_i|Q)$ for $i \in \{1, 2\}$,
        let $F' \subseteq L$, $|F'| \leq 2$ such that $H' \coloneqq H'_1 \cup H'_2 \cup E(\mathcal{P}) \cup F' \cup \{ e \}$ is 2EC. 
    \label{line:alg:reduce:handle:P2-separator-1}
    \Return $H'$
    }
    \Else{
        Let $H_1' \coloneqq \OPT(G_1|Q)$, $H'_2 = \RED(G_2|Q)$ and let $F' \subseteq L$, $|F'| \leq 2$, $e$ (according to \Cref{lem:structured:finding-F-2}) such that $H' \coloneqq H'_1 \cup H'_2 \cup E(\mathcal{P}) \cup F' \cup \{ e \}$ is 2EC. 
    \label{line:alg:reduce:handle:P2-separator-2}
    \Return $H'$
    }
\end{algorithm}

\begin{algorithm}
\DontPrintSemicolon
\caption{Handle $C^2$ Separator}
\label{alg:handle-C2-separator}
 \textbf{Input:} Instance $G$, $C^2$-separator $C$, sets $V_1, V_2$ \\
    Let $G_i \coloneqq G[V_i \cup C]|C$.\\
    \label{line:alg:reduce:C2-separator}
    \If{ $\opt(G_i|C) \geq \frac{1}{\eps}$ for $i \in \{1, 2\}$}{
        Let $H'_i = \RED(G_i|C)$ for $i \in \{1, 2\}$ and
        let $H' \coloneqq H'_1 \cup H'_2 \cup C$. 
    \label{line:alg:reduce:handle:C2-separator-1}
    \Return $H'$
    }
    \Else{
        Let $H_1' \coloneqq \OPT(G_1|C)$, $H'_2 = \RED(G_2|C)$ and let $H' \coloneqq H'_1 \cup H'_2 \cup C$. 
    \label{line:alg:reduce:handle:C2-separator-2}
    \Return $H'$
    }
\end{algorithm}

\begin{algorithm}
\DontPrintSemicolon
\caption{Handle Degenerate Paths}
\label{alg:handle-degenerate-paths}
 \textbf{Input:} Instance $G$, set of degenerate paths $\mathcal{P}' = \{P_1', ..., P_\ell' \}$. \\
    For each $P'_i \in \mathcal{P}'$ let $u'_i, w'_i$ be the endpoints of $P'_i$ and let $P_i \in \mathcal{P} \setminus \mathcal{P}'$ (with endpoints $u_i, w_i$) be the unique path to which $u'_i$, $w'_i$ are incident to s.t.\ $u_i' u_i \in L$ and $w_i' w_i \in L$ (according to the definition of a degenerate path).\\ 
    Define $C_i = (V(P_i') \cup V(P_i), E(P_i') \cup E(P_i) + u_i' u_i + w_i' w_i)$ and let $\mathcal{C}= \{ C_i | P_i' \in \mathcal{P}' \}$.\\
    Let $H_1 = \RED(G|\mathcal{P}')$ and $H_2 = \RED(G|\mathcal{C})$.\\
    \label{line:alg:reduce:degenerate}
    \If{ $|H_1 \cap L| + \ell \leq |H_2 \cap L| + 2\ell$ }{
        \Return $H_1 \cup \bigcup_{i=1}^\ell E(P_i') + u_i' w_i'$.
    \label{line:alg:reduce:handle:degenerate-1}
    }
    \Else{
        \Return $H_2 \cup \bigcup_{i=1}^\ell E(C_i)$.
    \label{line:alg:reduce:handle:degenerate-2}
    }
\end{algorithm}

\subsection{Proof of Correctness of the Algorithm}

\begin{lemma}
    \label{lem:reduce:running-time}
    \RED runs in polynomial time in $|V(G)|$ if \ALG does so.
\end{lemma}

\begin{proof}
    We say that a path $P \in \mcP$ is symmetric if $N(u) = N(w)$ for all $u, w \in V(P)$, i.e., all vertices of $P$ have the same neighbors. We define $\mcP_S$ to be the set of symmetric paths and define $\mcP_n = \mcP \setminus \mcP_S$ to be the set of non-symmetric paths.  
    Observe that all paths that are constructed from isolated nodes in \Cref{line:alg:isolated:nodes} are symmetric paths.
    Let $n \coloneqq |V(G)|$, $p = |\mcP_n|$ and define $s(G) \coloneqq n^2 + M^2 \cdot p^q$ as the size of the problem on the instance $G$, where $M$ is a sufficiently large integer and $q \geq 2$ is a constant to be specified later. 
    For example, $M = N$ suffices, where $N$ is the number of vertices of the original instance before we applied \RED.

    During the algorithm we apply some step which can be performed in polynomial time and potentially also apply \RED recursively to sub-instances of the problem.
    Note that we can check all if-conditions in polynomial time: checking whether the optimum solution has a constant number of links can be done in polynomial time and checking whether $G$ contains a strongly $(\frac 74, t, 1)$-contractible subgraph $H$ for some constant $t$ can be done in polynomial time due to \Cref{lem:structured:finding-small-contractible-polytime}. 
    Checking whether one of the separator conditions is satisfied can also be done in polynomial time.
    Finally, we can identify all degenerate paths in polynomial time and thus also check whether there are at least $\varepsilon |\mathcal{P}|$ many such paths.
    
    Hence, in order to bound the running-time, it suffices to show that whenever \Cref{line:alg:isolated:nodes} is executed (except for the very first time), then we have reduced the total size of the constructed subproblems.
    Then, by the Master Theorem (see e.g.~\cite{bentley1980general}) it follows that we obtain a polynomial running-time.
    
    This is clearly true if \Cref{line:alg:reduce:constant} is applied, as we then output the optimum solution, which can be done in polynomial time since $\opt$ is constant.
    Similarly, if \Cref{line:alg:reduce:contractible} is executed then $s(G|H) < s(G)$ and we are done.
    Note that we can not enter an infinite loop of alternatingly \Cref{line:alg:isolated:nodes} and \Cref{line:alg:reduce:contractible}, as the application of the reduction in \Cref{line:alg:isolated:nodes} does not create new strongly $\alpha$-contractible subgraphs.

    Now consider the subproblems created in Lines~\ref{line:alg:reduce:main:path-separator}, \ref{line:alg:reduce:main:P2-separator}, and \ref{line:alg:reduce:main:C2-separator}. 
    Afterwards, we consider \Cref{line:alg:reduce:main:degenerate-path-separator}.
    If we only create one subproblem $G_1$ from $G$ then it is easy to see that $s(G_1) < s(G)$ and we are done in any of the if conditions.
    If we create two subproblems $G_1, G_2$, note that in any of the above if-conditions we have that $s(G_1) + s(G_2) < s(G)$, where $s(G_i)$ is the size of $G_i$ after the application of \Cref{line:alg:isolated:nodes}, for $i = \{1, 2\}$.

    Now let us turn to \Cref{line:alg:reduce:main:degenerate-path-separator}.
    Here, we create two subproblems. 
    In this case, it is important that the number of degenerate paths $\mathcal{P}'$ is at least $\varepsilon | \mcP |$. 
    First, note that any $P \in \mathcal{P}'$ is non-symmetric, as otherwise all edges from $P$ go to some unique path $P'$ and hence the subgraph induced by the vertices of $P$ and $P'$ is 1-contractible (similar to \Cref{lem:structured:path}). 
    Now observe that in $H_1$ and $H_2$ generated in \Cref{alg:handle-degenerate-paths} (after the application of \Cref{line:alg:isolated:nodes}) each degenerate path is symmetric. 
    Hence, in each of the subproblems defined on $H_1$ and $H_2$ we have reduced the number of non-symmetric paths by at least a fraction of $\varepsilon$, which is constant. 
    Therefore, by choosing $q$ to be a sufficiently large (yet constant) integer, it follows that $s(H_1) + s(H_2) < s(G)$.
    
    Finally, observe that if we execute \Cref{line:alg:reduce:main:ALG}, clearly we have a polynomial running time by our assumption.
\end{proof}

Using the lemmas that we have established before, it is straightforward to verify that \RED returns a feasible solution.

\begin{lemma}
    \label{lem:reduce:feasibility}
    If $G$ is a 2EC spanning subgraph, then \RED returns a 2EC spanning subgraph for $G$ if \ALG returns a 2EC spanning subgraph for structured instances.
\end{lemma}

\begin{proof}
We prove the lemma inductively for increasing values of $|V(G)|$.
The base cases are given in \Cref{line:alg:reduce:constant} and \Cref{line:alg:reduce:main:ALG}, in which the statement clearly holds and by the correctness of \ALG, respectively.
If \Cref{line:alg:reduce:contractible} is executed, the statement follows from the definition of contractibility.
If one of the Lines~\ref{line:alg:reduce:main:path-separator}, \ref{line:alg:reduce:main:P2-separator}, or \ref{line:alg:reduce:main:C2-separator} is executed, then the statement follows from Lemmas \ref{lem:structured:finding-F-1}, \ref{lem:structured:finding-F-2}, or \ref{lem:structured:finding-F-3}, respectively.
Finally, if \Cref{line:alg:reduce:main:degenerate-path-separator} is executed, it is clearly true that both solutions are feasible.
\end{proof}

It remains to upper bound the approximation factor of \RED. Let $\red(G)$ be the number of links output by \RED. 

\begin{lemma}
    \label{lem:reduce:approx}
    The number of links in the solution output by the algorithm is bounded by
    \begin{equation*}
                    \red(G) \leq \begin{cases}
                        \opt(G) & \text{ if $\opt(G) \leq \frac{3}{\eps}$}, \\
                        \alpha \cdot \opt(G) + 3 \eps \opt(G) -3  & \text{ if $\opt(G) > \frac{3}{\eps}$.}
                    \end{cases} 
    \end{equation*}
 \end{lemma}

\begin{proof}
Similarly to the previous lemma, we prove the lemma inductively for increasing values of $|V(G)|$.
If \Cref{line:alg:reduce:constant} is executed, the statement clearly holds.
If \Cref{line:alg:reduce:contractible} is executed, the statement follows since it is an $\alpha$-contractible subgraph:
For a strongly $(\alpha, t, 1)$-contractible subgraph $H$ we have $\opt(G) = \opt(G) \cap G[V \setminus V(H)] + \opt(G) \cap G[V(H)] \geq \opt(G|H) + \frac{1}{\alpha}\opt(H)$, by the definition of $\alpha$-contractible subgraphs and since $\opt(G) \cap G[V \setminus V(H)]$ is a 2EC spanning subgraph of $G|H$.
If $\opt(\red(G|H)) \leq \frac{3}{\eps}$, then $\RED(G) = \opt(G|H) + \frac{1}{\alpha}\opt(H) \leq \alpha \opt(G) \leq \alpha \opt(G) + 3 \eps \cdot \opt(G) -3$, since $\opt(G) \geq \frac{3}{\eps}$.
If $\opt(\red(G|H)) > \frac{3}{\eps}$, we have $\opt(G) \leq \frac{1}{\alpha}\opt(H) + \alpha \opt(G|H) + 3 \eps \opt(G|H) -3 \leq \alpha \opt(G) + 3 \eps \opt(G) -3$, since $\opt(G) \geq \opt(G|H)$.

Next, assume that \Cref{line:alg:reduce:main:P2-separator} is executed.
Lines~\ref{line:alg:reduce:main:path-separator} and \ref{line:alg:reduce:main:C2-separator} are analogous.
First, assume that \Cref{line:alg:reduce:handle:P2-separator-1} of \Cref{alg:handle-P2-separator} is executed.
We have $\opt(G) \geq \opt(G_1|Q) + \opt(G_2|Q)$ and therefore $\RED(G) = \RED(G_1|Q) + \RED(G_1|Q) + 3 \leq \alpha \opt(G_1|Q) + 3 \eps \opt(G_1|Q) -3 + \alpha \opt(G_2|Q) + 3 \eps \opt(G_2|Q) -3 + 3 \leq \alpha \opt(G) + 3 \eps \opt(G) -3$, where the '$+3$' after the first equation is due to \Cref{lem:structured:finding-F-2}.
Next, assume that \Cref{line:alg:reduce:handle:P2-separator-2} of \Cref{alg:handle-P2-separator} is executed.
As before, we have $\opt(G) \geq \opt(G_1|Q) + \opt(G_2|Q)$. 
Additionally, by the definition of a $P^2$-separator we have that $\opt(G_2|Q) \geq 4$ and $\opt(G_1|Q) \geq 3$. 
If also $\opt(G_1|Q) \geq 4$, then we have $\alpha \opt(G_i|Q) \geq \opt(G_i|Q) + 3$ for $i \in \{1, 2\}$, as $\alpha \geq \frac{7}{4}$.
Therefore, we have $\RED(G) = \opt(G_1|Q) + \RED(G_1|Q) + 3 \leq \alpha \opt(G_1|Q) -3 + \alpha \opt(G_2|Q) + 3 \eps  \opt(G_2|Q) -3 + 3 = \alpha \opt(G) + 3 \eps \opt(G) -3$, since $\opt(G) \geq \opt(G_2|Q)$.
Again, the '$+3$' after the first equation is due to \Cref{lem:structured:finding-F-2}.
Else, assume that $\opt(G_1|Q) = 3$.
First, note that $\alpha \opt(G_1|Q) \geq \opt(G_1|Q) + 2$, since $\alpha \geq \frac{7}{4}$.
Furthermore, by the definition of a $P^2$-separator, there must be an edge $e \in E(\mathcal{P})$ from one of the endpoints of $Q$ to $V_2$.
Then \Cref{lem:structured:finding-F-2} implies that actually $|F| \leq 1$ and hence we have $\RED(G) = \opt(G_1|Q) + \RED(G_1|Q) + 2 \leq \alpha \opt(G_1|Q) - 2 + \alpha \opt(G_2|Q) + 3 \eps  \opt(G_2|Q) -3 + 2 = \alpha \opt(G) + 3 \eps \opt(G) -3$, since $\opt(G) \geq \opt(G_2|Q)$.

Finally, assume that \Cref{line:alg:reduce:main:degenerate-path-separator} is executed.
We use the notation as in \Cref{alg:handle-degenerate-paths}.
For $1 \leq i \leq \ell$ let $K_i = (V(P_i), E(P_i) + u_i' w_i')$ and define $\mathcal{K} = \{ K_i \}_{1 \leq i \leq \ell}$.
We show the following: If $\mathcal{C}$ is not weakly $(\alpha, 2, \ell)$-contractible, then $\mathcal{K}$ is weakly $(\alpha, 1, \ell)$-contractible.
From this statement it clearly follows that the better of the two solutions $H_1$ and $H_2$ satisfies the statement of the lemma.

It remains to prove the above statement.
Fix some optimum solution $\opt$.
Observe that any feasible solution must contain at least one link from within $G[C_i]$ for each $1 \leq i \leq \ell$, as $u_i'$ and $w_i'$ are only adjacent to vertices of $P_i'$ and $P_i$, by the definition of a degenerate path.
Furthermore, if $u_i' w_i' \notin \opt$, then it follows that $\opt$ has at least 2 links from within $G[V(C_i)]$, by the same argument as above.
Let $\beta$ be the number of paths of $\mathcal{P}'$ for which $\opt$ only has one link from within $G[V(C_i)]$, $ 1 \leq i \leq \ell$.
Then we have $\opt \cap \mathcal{C} \cap L \geq \beta \ell + (1 - \beta) \cdot 2 \ell = (2 - \beta) \ell$.
Hence, since $\mathcal{C} \cap L = 2 \ell$, we have that if $\beta \leq \frac 67$, then it follows that $\mathcal{C}$ is a $(\alpha, 2, \ell)$-contractible subgraph for any $\alpha \geq \frac 74$.
Hence, let us assume $\beta \geq \frac 67$.
Now, $\opt$ takes at least $\frac 67 \ell$ many links from within $G[V(\mathcal{K})]$.
Therefore, since in $\mathcal{K}$ there are exactly $\ell$ many links, $\mathcal{K}$ forms a $(\alpha, 1, \ell)$-contractible subgraph for any $\alpha \geq \frac 74$.
\end{proof}

Finally, we need to show that whenever we call \RED or \ALG, then indeed we call an instance of PAP.

\begin{lemma}
    \label{lem:structured:remains-PAP}
    Whenever \RED or \ALG is called recursively in \RED on some graph $G'$, then $G'$ is an instance of PAP.
\end{lemma}

\begin{proof}
    We check each line individually.
    If \Cref{line:alg:reduce:contractible} is executed, then by definition the remaining instance is a \PAP instance.
    If \Cref{line:alg:reduce:main:path-separator} is called, the path $P$ is only truncated in each sub-instance, resulting in an instance of \PAP.
    If \Cref{line:alg:reduce:main:P2-separator} is executed, we have exhaustively applied \Cref{line:alg:reduce:main:path-separator} and hence there are no path-separators any more. 
    Therefore, no subset of the considered $P^2$-separator is a path-separator and hence contracting the $P^2$-separator results in an instance of \PAP.
    If \Cref{line:alg:reduce:main:C2-separator} is executed, note that we shrink 2 paths into one vertex in each sub-instance. Hence, also in this case both sub-instances are indeed instances of \PAP.
    Finally, if \Cref{line:alg:reduce:main:degenerate-path-separator} is executed, in both subinstances we essentially contracted several paths into single vertices each, or several cycles consisting of 2 paths to single vertices each. 
    Since this is exactly what we also do in \Cref{line:alg:reduce:main:path-separator} and \Cref{line:alg:reduce:main:C2-separator}, respectively, the resulting instances are also instances of \PAP
\end{proof}

It remains to prove that when we reach \Cref{line:alg:reduce:main:ALG} of \RED, the graph $G$ has the properties of \Cref{def:structured-graph}.

\begin{lemma}
\label{lem:reduce:properties}
    Consider the graph $G$ when reaching \Cref{line:alg:reduce:main:ALG} of \RED. Then $G$ satisfies the properties of \Cref{def:structured-graph}.
\end{lemma}

\begin{proof}
We prove the properties one by one. 
Properties P0, P1, P3, P4, P5, P6, and P7 are clear as otherwise this would contradict our algorithm.
Property P2 follows from \Cref{lem:structured:endpoint-adj-this-path}, since we exclude all strongly $(\alpha, t, 1)$-contractible subgraphs for some $\alpha \geq \frac 74$ and $t \leq 10$.
\end{proof}

Hence, \Cref{lem:reduce:properties} implies that when \ALG is executed, then $G$ is structured, i.e., $G$ has the properties of \Cref{def:structured-graph}. 
Therefore, \Cref{lem:structured:main} follows by the Lemmas~\ref{lem:reduce:running-time}, \ref{lem:reduce:feasibility}, \ref{lem:reduce:approx}, \ref{lem:structured:remains-PAP} and~\ref{lem:reduce:properties}.

\section{Credit Scheme, Relaxation, and Starting Solution}
\label{sec:initial-solution}

This section is dedicated to proving \Cref{lem:main:starting-solution-satisfies-invariants}, which we outlined in \Cref{sec:overview:credit-starting}.
For convenience, in this section we repeat the definitions introduced in \Cref{sec:overview:credit-starting}.

Throughout our algorithm we maintain a partial solution, which is not necessarily a feasible solution.
Furthermore, we assign certain credits to parts of the current graph $H = (V, E(\mcp) \cup S)$, that is, to vertices, links, components and blocks of $H$.
After first computing a starting solution, we then add (and sometimes remove) links to/from $H$ such that first we obtain a partial solution that does not contain bridges. 
Afterwards, we add (and sometimes remove) links to/from $H$ to glue the 2-edge-connected component together, such that eventually we obtain a single 2EC component, which is then a feasible solution to our problem.
While doing so, for our partial solution $H$ we maintain the invariant that the number of credits for $H$ plus the number of links $|S|$ in $H$ is no more than $\apxr \cdot \opt$.
By observing that the number of credits is non-negative, we obtain our desired guarantee on the approximation ratio.

In the remaining part of this section we first define the credit assignment rule and three invariants. 
The rest of this section is then dedicated to proving \Cref{lem:main:starting-solution-satisfies-invariants}.

\subsection{Credit Scheme}

In this section we describe how we distribute the credits.
It will we convenient to use the following definitions.
Let $H$ be a spanning subgraph of $G$, and let $C$ be a connected component of $H$. 
A component is called \emph{complex} if it contains at least one bridge. 
For any such $C$, we define the graphs $G_C$ and $H_C$ as follows:
$G_C$ is the graph obtained from $G$ by contracting every connected component of $H$ except for $C$, as well as every block of $C$.
$H_C$ is the corresponding contraction of $H$.
The structure of $H_C$ consists of a tree $T_C$ along with singleton vertices. 
A vertex in $G_C$ or $H_C$ that corresponds to a block in $H$ is called a \emph{block vertex}. 
If a block vertex corresponds to a simple block, it is called a \emph{simple block vertex}.
We identify the edges/links in $G$ and the corresponding edges/links in $G^C$ and $H^C$, respectively.

A component $C$ is called \emph{trivial} if it consists of a single path $P \in \mcP$ (i.e., $V(C) = V(P)$). 
Otherwise, it is called \emph{nontrivial}. 
A vertex $v \in V$ is called \emph{lonely} if it does not belong to any 2-edge-connected component or block of $H$. 
For two paths $P, P' \in P$ with endpoints $u,v$ and $u', v'$, respectively, we say that $uu' \in L$ is an \emph{expensive link} if $N_G(v) \subseteq (V(P) \cup V(P')) \setminus \{v'\}$. 
If additionally $uu' \in L(H)$ is a bridge in $H$ and $v$ is a lonely leaf in $H$, then $v$ is called an \emph{expensive leaf vertex}.

For $H = (V, E(\mathcal{P}) \cup S)$, we assign credits according to the following rules.

\begin{itemize}
\item[(A)] Every lonely vertex $v$ that is a leaf of some complex component $C$ receives 1 credit.
\begin{itemize}
\item[(i)] It receives an additional $\frac{1}{4}$ credit if $v$ is an expensive leaf vertex.
\item[(ii)] It receives an additional $\frac{1}{4}$ credit if there is a path in $H_C$ from $v$ to a simple block vertex of $C$ without passing through a non-simple block vertex.
\end{itemize}
\item[(B)] Every bridge $\ell \in S$ (i.e., a link that is not part of any 2-edge-connected component of $H$) receives $\frac{3}{4}$ credits.
\item[(C)] Every complex component receives 1 credit.
\item[(D)] Every non-simple 2-edge-connected block receives 1 credit.
\item[(E)] Each 2-edge-connected component receives:
\begin{itemize}
\item 2 credits if it is large or contains a degenerate path.
\item $\frac{3}{2}$ credits if it is small and does not contain a degenerate path.
\end{itemize}
\end{itemize}

Note that according to the definition, a lonely leaf vertex can have credit $1$, $\frac{5}{4}$ or $\frac{3}{2}$. 
For $H = (V, E(\mathcal{P}) \cup S)$, we write $\credits(H)$ to denote the total number of credits used according to the rules (A)-(E).
We define $\cost(H) = \credits(H) + |S|$.

To prove the approximation guarantee, we show that we maintain the following credit invariant.

    \invariantcredit*

Furthermore, we maintain the following two invariants.

\invariantlonely*

\invariantblock*

In the next subsection we introduce a novel relaxation of \PAP and compute a solution to it using some approximation algorithm for some other problem.
This solution might not be feasible to \PAP, but will serve as a starting solution.
We then show that this solution satisfies all invariants. Formally, we will show the following lemma.

\lemmamainstarting*

\subsection{A relaxation of \PAP: The 2-Edge Cover with Path Constraints and the Track Packing Problem}
Typical approaches for special cases of \PAP, such as unweighted $2$-ECSS and MAP, use some relaxation of the problem, for which one can compute an optimum solution in polynomial time and which serves as an (infeasible) starting solution. 
Then the goal is to add and remove edges to/from this relaxation to obtain a feasible solution.
A typical relaxation is the $2$-edge-cover problem, in which one seeks for a minimum-cardinality edge set such that each vertex is incident to at least 2 edges. 
In these special cases of \PAP, after certain preprocessing steps, optimum 2-edge-covers are usually not too far away from an optimum solution of the original problem and hence they serve as a good lower bound. 
However, in the case of \PAP, such a 2-edge-cover can be quite far away from an optimum solution.
For example, if the 2-edge-cover of a \PAP instance only contains the edges of $E(\mcP)$ together with links from $\hat{L}$, i.e., links that connect endpoints of the same path, then this edge set only 2-edge connects the paths individually, but it does not help to establish overall $2$-edge-connectivity. Next, we introduce a new relaxation, called 2-Edge Cover with Path Constraints.

\paragraph{A relaxation of \PAP: The 2-Edge Cover with Path Constraints.}
In this section we consider a different relaxation of \PAP, which we call 2-Edge Cover with Path Constraints (\tecpc), and is defined as follows.
An instance of \tecpc has the same input as a structured instance of \PAP (i.e., the input graph $G$ is structured), but with the additional constraint that each vertex is part of some path $P_i \in \mathcal{P}$ and that each path $P_i \in \mcP$ has exactly three vertices: the two endvertices $u_i, v_i \in V^*(P_i)$ of $P_i$ and the interior node $w_i$ of $P_i$.
Recall that for some subset $\mathcal{P}' \subseteq \mcP$ the set $V^*(\mathcal{P'})$ is the set of endpoints of $\mathcal{P'}$.
The goal is to compute a minimum cardinality link set $X \subseteq L$ such that $\delta_X(u) \geq 1$ for all $u \in V^*(\mcP)$ and $\delta_X(P_i) \geq 2$ for each $P_i \in \mcP$, i.e., each vertex is incident to at least two edges in $E(\mcP) \cup X$ and each path $P_i$ is incident to at least two outgoing edges in $X$.
An example of a feasible solution to a 2ECPC instance is given in \Cref{fig:s2ecpc-tpp}.

From a structured instance of \PAP we construct an instance of \tecpc as follows:
There are two additional constraints on the \tecpc instance compared to a \PAP instance.
The first condition is satisfied for any structured instance of \PAP. 
The second condition can be obtained as follows: if $|V(P_i)| \geq 4$ we contract all interior vertices of some path $P_i \in \mcP$ to a single vertex $w_i$. 
If $|V(P_i)| = 2$, we remove $u_i v_i$ from $P_i$, where $u_i, v_i \in V^*(P_i)$ are the endvertices of $P_i$ and add the edges $u_i w_i$ and $w_i v_i$ to $P_i$, where $w_i$ is a dummy node.
Note that any structured instance of \PAP satisfies $|V(P_i)| \geq 2$.
There is a one-to-one correspondence between edges/links in the graph of the instance of \PAP and edges/links in the graph of the instance of \tecpc.
Throughout this section we fix a structured instance of \PAP, which we also fix for the input of the instance of \tecpc. We denote by $\OPT$ and $\OPT_C$ the optimum solutions for the \PAP and the \tecpc instance, respectively, and define $\opt$ and $\opt_C$ to be their objective values.
It is easy to verify that any feasible solution to the \PAP instance is also feasible for its corresponding \tecpc instance, and hence \tecpc is indeed a relaxation for \PAP. We summarize this as follows.

\begin{proposition}
    \label{prop:pap-relaxation}
    Let $G$ be a structured instance of \PAP. Then $\opt \geq \opt_C$.
\end{proposition}

It is easy to see that $|\mcP| \leq \opt_C \leq 2 |\mcP|$:
Observe that since the input graph $G$ is structured, each endpoint $u$ of some path $P \in \mcP$ is incident to a link $e \in \delta(P)$.
Therefore, by picking such an edge for each endpoint $u$, there is a straightforward solution containing at most $2 |\mcP|$ many links.
Furthermore, note that any feasible solution must pick at least $|\mcP|$ many links.
Therefore we obtain the inequalities above.

Similar to other connectivity augmentation problems, we assume that the instance of \tecpc is \emph{shadow-complete}. 
In our setting this means the following:
Let $P_i, P_j \in P$ such that $P_i \neq P_j$. 
If there is a link $u_i u_j$ from an endpoint $u_i$ of $P_i $ to some other endpoint $u_j$ of $P_j$, then we also add the \emph{weaker} links $u_i w_j$, $w_i u_j$ and $w_i w_j$, where $w_i, w_j$ are the interior vertices of the paths $P_i$ and $P_j$, respectively.
Furthermore, if there is a link $u_i w_j$ from an endpoint $u_i$ of $P_i$ to the interior vertex $w_j$ of $P_j$, then we also add the \emph{weaker} link $w_i w_j$, where $w_i$ is the interior vertex of $P_i$.
These weaker links are then also called the \emph{shadows} of their original links.
An instance that contains all such shadows is called shadow-complete.

It is not hard to see that for any feasible instance of \tecpc, a feasible solution to the shadow-complete instance can be turned into a feasible solution to the original instance of the same cost.

\begin{restatable}[]{lemma}{lemshadowcomplete}
    \label{lem:shadow-complete-no-less-cost}
    For any feasible instance of \tecpc arising from a structured graph, a feasible solution $Y$ to the shadow-complete instance can be turned into a feasible solution $Y'$ to the original instance such that $|Y'| \leq |Y|$.
\end{restatable}

\begin{proof}
    Consider a feasible solution $Y$ to the shadow-complete instance.
    Without loss of generality we assume that $Y$ is minimal, i.e., $Y-e$ is not feasible for each $e \in Y$. 
    We show how to transform $Y$ into a feasible solution $Y'$ such that the number of links in $Y'$ is at most the number of links in $Y$ and $Y'$ contains less shadows than $Y$. 
    By iteratively doing this procedure, we are done.
    For some $e \in L$ let $S(e)$ be the set of shadows of $e$ including $e$ itself.
    Observe that $|S(e)| \leq 4$, by the definition of shadows.
    If $Y$ contains for each original link $e$ only one link of $S(e)$, then it is straightforward to see that by replacing each shadow with its original link, we obtain a feasible solution to the original instance.
    Hence, let us assume there is some link $e \in L$ such that $Y$ contains at least two links from $S(e)$.
    Observe that if $Y$ contains at least three such links, then there must be at least one link $f \in Y$ such that $Y-f$ is also feasible, a contradiction as $Y$ is minimal.
    Hence, let us assume that $Y$ contains exactly two such links from $S(e)$. 
    Let $u, x'$ be the endpoints of $e$ with $u \in V(P)$ and $x' \in V(P')$, $P, P' \in \mcP$, $P \neq P'$.
    
    We consider two cases: First, assume that $u$ is an endvertex of $P$ and $x'$ is the interior vertex of~$P'$.
    Then the unique shadow of $e$ is $w x'$, where $w$ is the interior vertex of $P$.
    Hence, $u x' \in Y$ and $w x' \in Y$.
    Let $v$ be the other endvertex of $P$ (i.e., $v \in V^*(P)$, $v \neq u$). 
    If $v$ is incident to a link $f$ in $Y$ such that $f \in \delta_Y(P)$, then we can replace $wx'$ with an arbitrary link from $\delta(P') \setminus Y$ that is not a shadow, which must exists (as the instance arises from a strictured graph $G$). 
    Observe that afterwards the solution is still feasible and contains fewer shadows.
    Otherwise, if $v$ is only incident in $Y$ to the link $v u$, then we first replace $v u$ with any link in $\delta(v) \cap \delta(P)$ that is not a shadow, which must exist since $G$ is structured. 
    Afterwards we are in the previous case and we do the same as before. 
    Observe that the resulting solution is still feasible and contains fewer shadows.

    In the second case, we assume that $u$ is an endvertex of $P$ and $x'$ is an endvertex of $P'$.
    Let $v, w \in V(P)$ such that $w$ is the interior vertex of $P$ and $v \neq u$ is the other endvertex of $P$.
    Similarly, let $y', z' \in V(P')$ such that $y'$ is the interior vertex of $P'$ and $z'\neq x'$ is the other endvertex of~$P'$.
    W.l.o.g.\ we can assume that $u x' \in Y$, as otherwise we can include $ux'$ to $Y$ and remove one of its shadows and observe that the resulting solution is still feasible and contains fewer shadows.
    We distinguish two sub-cases:
    First, assume that $ux' \in Y$ and $uy' \in Y$.
    Assume that $v$ is incident to a link in $Y$ of $\delta(P)$. 
    Then we can replace $uy'$ by any link of $\delta(z') \cap \delta(P')$ that is not a shadow and which is not in $Y$ (which must exist). Observe that afterwards the solution is still feasible and contains fewer shadows.
    Similarly, we do the same if $z'$ is incident to a link in $Y$ of $\delta(P')$.
    Hence, in the remaining case we have that $uv \in Y$ and $x'z' \in Y$.
    Now observe that by replacing $uv$ and $x'z'$ by links from $\delta(v) \cap \delta(P)$ and $\delta(z') \cap \delta(P')$ that are not shadows, respectively, and additionally removing $uy'$, we have that the resulting solution is still feasible and contains fewer links and shadows.
    
    In the remaining case, if $ux' \in Y$ and $wy' \in Y$, we do the same as in the previous sub-case and observe that this also yields a feasible solution with fewer shadow links. This finishes the proof.
\end{proof}

Hence, from now on we assume that our instance of \tecpc is shadow-complete.

\paragraph{The Track Packing Problem.}
Instead of computing an (approximate) solution to \tecpc, we consider a packing problem, which we call the Track Packing Problem (\tpp).
We construct an instance of \tpp from a shadow-complete instance of \tecpc.
In order to do this, we establish some notation.
For some link set $A \subseteq L$ let $V^*(A) = V^*(\mcP) \cap A$ be the set of vertices that are endpoints of some path and are incident to some link in $A$.
We define a \emph{track} as follows.
A track $T \subseteq L$ with $T = Q(T) \cup I(T)$ is a subset of links such that the following conditions hold.
\begin{itemize}
    \item $Q(T)$ is a simple path in $G$ that contains at least one link, where the start vertex $u$ and the end vertex $v$ of $Q(T)$ are endpoints of some paths $P, P' \in \mcP$ (not necessarily $P \neq P'$) such that $u \neq v$,
    \item any interior vertex of $Q(T)$, i.e., a vertex from $V(Q(T)) \setminus \{u, v \}$, is the interior vertex of some path $P \in \mcP$, 
    \item $I(T) \subseteq L$ contains exactly the following set of links: for each interior vertex $w_i \in V(Q(T)) \setminus \{u, v \}$, which is an interior vertex of some path $P_i \in \mcP$, 
    there is a link $u_i v_i \in I(T)$, where $u_i, v_i \in V^*(P_i)$ are the endvertices of $P_i$, and
    \item if $|T| = 1$, i.e., if $|Q(T)|=1$ and $|I(T)| = 0$, then the start vertex $u$ of $Q(T)$ and the end vertex $v$ of $Q(T)$ belong to distinct paths.
\end{itemize}
For any track $T$ we refer by $Q(T)$ and $I(T)$ as the set of links as above.
Observe that the number of links in a track is odd, since $|Q(T)| = |I(T)| + 1$. 
Further note that a track of length 1 is simply a link in $\bar{L}$, i.e., a link between two endvertices of two distinct paths $P, P' \in \mcP$. 

The instance of \tpp is now defined as follows.
The input is simply a (shadow-complete) instance of 2ECPC with structured graph $G$.
Let $\mathcal{T}$ be the set of all tracks in $G$.
The task is to select a maximum number of disjoint tracks from $\mathcal{T}$, where two tracks $T, T'$ are disjoint if $V(T) \cap V(T') = \emptyset$.
An example of a feasible solution to a TPP instance is given in \Cref{fig:s2ecpc-tpp}.

In the following we establish a relationship between solutions of \tecpc and its corresponding \tpp instance.
This connection is similar to the connection between 2-matchings and 2-edge covers. In particular, an optimum solution to one of one of the problems can be turned into an optimum solution to the other problem in polynomial time.
We note that the additional condition required in \Cref{lem:relaxation:relationship:packing-to-covering} and \Cref{lem:relaxation:relationship:covering-to-packing} is always satisfied for structured graphs.
The relationship between 2ECPC and TPP is highlighted in \Cref{fig:s2ecpc-tpp}.

\begin{restatable}[]{lemma}{lemPPPtoPCP}
    \label{lem:relaxation:relationship:packing-to-covering}
    Given a solution $X$ for a shadow-complete instance of \tpp, in which each endvertex of a path is connected to a vertex of a different path, we can compute in polynomial time a solution $Y$ to the corresponding \tecpc instance of value $2 |\mcP| - |X|$.
\end{restatable}

\begin{proof}
    We construct a solution $Y$ to \tecpc as follows.
    \begin{enumerate}
        \item For each $T \in X$ we add $T$ to $Y$.
        \item For each $a \in V^*(\mcP)$ that is not incident to some link in $Y$ so far, add an arbitrary link $a b$ such that $b$ does not belong to the same path as $a$.
    \end{enumerate}
    First, note that such a link as in the second step always exists since $G$ is structured.
    We now prove that $Y$ is feasible and then bound the number of links.
    Let $Y_1$ be the set of links added in the first step of the above algorithm and $Y_2$ be the set of links added in the second step.
    Consider some track $T \in X$ and let $u, v$ be the endpoints of the path $Q(T)$.
    Consider some vertex $y \in V(Q(T)) \setminus \{u, v\}$, which is an interior vertex of some path $P \in \mcP$.
    Observe that $|\delta(P) \cap T| = 2$ and $|\delta(x) \cap T| = 1$ for each $x \in V^*(P)$, i.e., $T$ satisfies all conditions for feasibility for the instance of \tecpc w.r.t.~$P$.
    Furthermore, observe that the start and end vertex $u$ and $v$ of $Q(T)$ are incident to exactly one edge in~$T$. Additionally, by definition, these edges are outgoing edges w.r.t.\ their path.
    Hence, by adding to $Y_1$ one link $a b$ to each vertex $a$ that is not incident to some link in $Y_1$ such that $b$ does not belong to the same path as $a$, the resulting link set satisfies $|\delta_Y(P)| \geq 2$ for each $P \in \mcP$ and $|\delta_Y(a)| \geq 1$ for each $a \in V^*(\mcP)$. 
    That is, $Y$ is feasible for the instance of \tecpc.

    It remains to prove that $|Y| = 2|\mcP| - |X|$.
    For some $T \in X$, observe that $|Q(T)| = |I(T)| +1$
    and that $|T| = |Q(T)| + |I(T)| = |V^*(T)| -1$.
    Define $A = \{ u \in V^*(\mcP) | \delta_{Y_1}(u) = 0 \}$ to be the set of endpoints of $\mcP$ that are not incident to any link of $Y_1$.
    Observe that $|A| = 2 |\mcP| - \sum_{T \in X} |V^*(T)|$.
    Therefore, we have that $|Y_2| = 2 |\mcP| - \sum_{T \in X} |V^*(T)|$.
    Using that $|T| = |V^*(T)| -1$ for each $T \in X$ we obtain $|Y| = |Y_1| + |Y_2| = 2 |\mcP| - |X|$.
\end{proof}

\begin{restatable}[]{lemma}{lemPCPtoPPP}
    \label{lem:relaxation:relationship:covering-to-packing}
    Consider a shadow-complete instance $I$ of \tecpc, in which each endvertex of a path is connected to a vertex of a different path.
    Given a minimal solution $Y$ to $I$ such that $|Y| \leq 2|\mcp|$, we can compute in polynomial time a solution $X$ to the corresponding \tpp instance of value at least $2 |\mcP| - |Y|$.
\end{restatable}

\begin{proof}
    The idea of the proof is the following. We remove certain links from $Y$ and modify $Y$ such that the remaining set of links $X'$ corresponds to a collection of disjoint tracks $X$ of the corresponding instance of \tpp.
    Furthermore, for each removed link from $Y$ in the first step, we assign a token to a specific vertex $v \in V$.
    By $t(v)$ we denote the number of tokens assigned to $v$ and further we denote by $t(P)$ the sum of tokens assigned to all vertices of $P$, i.e., $t(P) = \sum_{v \in V(P)} t(v)$ for some path $P \in \mcP$.
    Initially we have $t(v)=0$ for all $v \in V$.
    We then show that $|X'| + \sum_{P \in \mcP} t(P) = |Y|$, which we can use to show that $|X| \geq 2 |\mcP| - |Y|$.

    We start by giving an algorithm that modifies $Y$ and outputs $X$, the desired solution to the instance of \tpp.
    Let $Y' \subseteq L$ be some set of links, which is not necessarily a feasible solution to the instance of \tecpc.
    Let $e = uv \in Y'$ and let $P, P' \in \mcP$ be such that $u \in V(P)$ and $v \in V(P')$.
    We say that $e$ is \emph{greedy} w.r.t.\ $u$ if $|\delta_{Y'}(P')| + t(P') \geq 3$ holds, and, if $v$ is an endvertex of $P'$, then additionally it must hold that $|\delta_{Y'}(v)| + t(v) \geq 2$.
    Furthermore, we say that $e$ is \emph{quasi-greedy} w.r.t.\ $u$ if $|\delta_{Y'}(P')| = 1 = t(P')$ holds and $xy \in Y_1$, where $x,y \in V^*(P')$ are the endvertices of $P'$.
       We modify $Y$ as follows, where initially we have $t(v)=0$ for all $v \in V$:
    \begin{itemize}[1.]
        \item[1.] As long as $Y$ contains a link $e= uv$ such that $u \in V(P)$ and $v \in V(P')$ for some $P, P' \in \mcP$ and $e$ is greedy w.r.t.\ $u$, remove $e$ from $Y$ and assign one token to $u$.
        \item[2.] As long as $Y$ contains a link $e= uv$ such that $u \in V(P)$ and $v \in V(P')$ for some $P, P' \in \mcP$ and $e$ is quasi-greedy w.r.t.\ $u$, remove $e$ from $Y$ and assign one token to $u$. Furthermore, remove $xy \in Y$, where $x,y \in V^*(P')$ are the endvertices of $P'$, and assign one additional token to the vertices of $P'$ and reassign the tokens of $P'$ such that $t(x) = t(y) = 1$.
        \item[3.] As long as $Y$ contains a path $P \in \mcP$ with endvertices $xy \in V^*(P)$ such that $t(P) \geq 1$ and $xy \in Y$, remove $xy \in Y$ and assign one additional token to the vertices of $P$ and reassign the tokens of $P$ such that $t(x), t(y) \geq 1$.
        \item[4.] If there is a path $P \in \mcP$ such that $t(P) \geq 2$, redistribute the tokens such that each endvertex of $P$ receives at least one token.
    \end{itemize}
    Let $Y_1, Y_2, Y_3$ be the set of edges after the first, second and third step, respectively.
    We now show that after the first step we have $|\delta_{Y_1}(P)| \leq 2$ for each path $P \in \mcP$.
    Assume this is not true, i.e., there is some $P \in \mcP$ such that $|\delta_{Y_1}(P)| \geq 3$.
    It is easy to see that the interior vertex of $P$ can not be incident to a link $e$ from $\delta_{Y_1}(P)$ as this is a greedy link.
    But then there must be an endpoint $v \in V^*(P)$ that is incident to at least $2$ links.
    Now observe that such a link incident to $v$ is a greedy link.
    
    Hence, we have $|\delta_{Y_1}(P)| \leq 2$ for each path $P \in \mcP$ and therefore $|\delta_{Y_3}(P)| \leq 2$.
    Observe that if $|\delta_{Y_3}(P)| = 2 - i$ for $i \in \{0, 1, 2\}$, then $t(P) \geq i$, i.e., we have assigned at least $i$ tokens to vertices of $P$, which follows from our token assignment rule.
    Similarly, if $\delta_{Y_3}(u) = 0$ for some endvertex of some path $P \in \mcP$, then $t(u) \geq 1$.
    Furthermore, we show that we have the following properties: 
    \begin{itemize}
        \item[(i)] If there is some path $P \in \mcP$ such that $t(P) \geq 2$, then no vertex $u \in V(P)$ is incident to some link from $Y_3$.
        \item[(ii)] If there is some path $P \in \mcP$ such that $t(P) = 1$, then there is one endvertex $x \in V^*(P)$ such that $|\delta_{Y_3}(x)| =1$ and one endvertex $y \in V^*(P)$ such that $|\delta_{Y_3}(y)| =0$.
        \item[(iii)] If there is some path $P \in \mcP$ such that $t(P) = 0$, then $|\delta_{Y_3}(P)| = 2$.
        \item[(iv)] If there is some path $P \in \mcP$ with $x, y \in V^*(P)$ such that $xy \in Y_3$, then $t(P) = 0$.
        \item[(v)] If there is a link $u v \in Y_3$, where $u$ is an endvertex of some path $P \in \mcP$ and $v$ is an interior vertex of some path $P' \in \mcP$, $P \neq P'$, then $t(P') = 0$ and $x y \in Y_3$, where $x, y \in V^*(P')$.
    \end{itemize}

    To see that (i) is true, observe that if there is a path $P$ such that $t(P) \geq 2$ and $|\delta_{Y_3}(P)| \geq 1$, then $P$ is incident to a greedy edge, a contradiction as $Y_2$ does not contain greedy edges.
    For (ii), consider some path $P \in \mcP$ such that $t(P) = 1$. Observe that $|\delta_{Y_3}(P)| = 1$ in this case. 
    Assume that $x y \in Y_3$. Then it is easy to see that $e \in \delta_{Y_3}(P)$ is a greedy edge, a contradiction.
    Hence, the unique edge $e \in \delta_{Y_3}(P)$ must be incident to some endvertex of $P$ and hence we obtain (ii).
    Property~(iii) clearly follows from previous observations.
    It can be checked that (iv) follows from (i)-(iii).
    Finally, to see (v), note that if the link $xy \notin Y_3$, then $uv$ is a greedy edge, a contradiction.
    Hence, we have $xy \in Y_3$ and by the previous properties we have $t(P') = 0$.

    We now further modify $Y_3$ as follows.
    \begin{itemize}[1.]
        \item[5.] As long as $Y_3$ contains a link $e= ux'$ such that $u \in V^*(P)$ and $x' \in V(P')$ for some $P, P' \in \mcP$, $P \neq P'$, and $u v \in Y_3$, where $u \neq v \in V^*(P)$, we replace $ux'$ by its shadow $x' w$, where $w$ is the interior vertex of $P$.
    \end{itemize}
    Let $Y_5$ be the set of edges after this fifth step. 
    We do one additional sixth step. 
    Let $H = (V, Y_5)|\mathcal{P}$ be the graph on edge set~$Y_5$ after contracting each path of $\mcP$ to a single vertex.
    Note that the maximum degree in $H$ is $2$ and hence it is a disjoint union of paths and cycles.
    \begin{itemize}[1.]
        \item[6.] If in $H$ there is a cycle $C$ such that each vertex $u \in V(C)$ in $H$ corresponds to a path $P \in \mcP$ such that $x y \in Y_5$, where $x, y \in V^*(P)$, then remove from $Y_5$ all links that correspond to the links in~$C$ and additionally for each $u \in V(C)$ that corresponds to a path $P \in \mcP$ remove $x y$ from $Y_5$, where $x, y \in V^*(P)$ and add one token to both $x$ and $y$. 
    \end{itemize}
    
    Note that in the sixth step the number of added tokens is equal to the number of removed links, which is exactly twice the number of paths incident to $C$.
    We now set $X'$ to be the set of links remaining after this sixth step. 
    We claim that $X'$ is a union of disjoint tracks. 
    We then define $X$ to be the set of tracks defined by $X'$ for the corresponding instance of \tpp. 
    Furthermore, we will show that $|Y| = |X'| + \sum_{P \in \mcP} t(P)$, and finally $|X| \geq 2 |\mcP| - |Y|$.

    We first show that the link set of $X'$ corresponds to a union of disjoint tracks. 
    Let $H' = (V, X)|\mathcal{P}$ 
    and note that the maximum degree in $H'$ is $2$ and hence it is a disjoint union of paths and cycles.
    Consider some path or cycle $A$ in $H'$. 
    We show that we can decompose $A$ into disjoint tracks. 
    By doing this for each path or cycle in $H'$, we are done.
    Assume first that $A$ is a path and consider $A$ as the corresponding link set in $G$.
    $A$ must start with some link $uv$ in some endvertex $u$ of some path $P \in \mcp$, by (ii).
    If $v$ is some endvertex of some other path, then this single link defines the first track.
    Else, if $v$ is some interior vertex of some other path $P'$, then by (v) we have $xy \in X'$, where $x, y \in V^*(P')$.
    By (iv) and (iii), there must be an additional link outgoing from $P'$ in $X'$ and we proceed as before, depending if this link hits an endvertex of some path or an interior vertex.
    Eventually, this procedure ends if we hit another endvertex of some path. 
    It can be easily observed that this sequence of links defines a track.
    This finishes the case in which $A$ is a path.
    The case that $A$ is a cycle is similar. 
    Here, we simply start at some vertex that is incident to a link of $A$ and additionally an endvertex of some path incident to $A$. 
    Such a vertex has to exist since otherwise $A$ corresponds to a cycle $C$ with exactly the properties of Step 6, a contradiction.
    Hence, $X'$ can be decomposed into disjoint tracks. Let $X$ be the set of these tracks.

    We now prove $|Y| = |X'| + \sum_{P \in \mcP} t(P)$.
    Observe that $X'$ was obtained from $Y$ by removing links and for each removed link we have added exactly $1$ token to some vertex $u \in V$. 
    Since $t(P) = \sum_{v \in V(P)} t(v)$, we simply have $|Y| = |X'| + \sum_{P \in \mcP} t(P)$.

    Finally, we show that $|X| \geq 2 |\mcP| - |Y|$.
    Let $V_1^* \subseteq V^*(\mcp)$ be the set of endvertices that are incident to $X'$, i.e., part of one of the tracks in $X$. Additionally, let $V_2^* = V^*(\mcp) \setminus V_1^*$ be the remaining set of endvertices.
    Note that $|X'| = \sum_{T \in X} (|V^*(T)| -1) = |V^*_1| - |X|$.
    Further, note that by feasibility of $Y$ we have that $\sum_{P \in \mcp} t(P) \geq V^*_2$.
    Hence, we obtain 
    $$ |Y| = |X'| + \sum_{P \in \mcP} t(P) = |V^*_1| - |X| + \sum_{P \in \mcP} t(P) \geq |V^*_1| - |X| + |V_2^*| = V^*(\mcp) - |X| = 2 |\mcp| - |X| \ , $$
    and therefore we get $|X| \geq 2 |\mcp| - |Y|$.
    This finishes the proof.
\end{proof}

We also note that \tpp is \NP-hard, by reducing it to the \NP-hard $P_2$-Packing problem~\cite{kirkpatrick1978completeness}. 
In the $P_2$-Packing problem we are given a graph $G$ and the task is to select a maximum number of vertex-disjoint $P_2$'s, i.e., paths that contain exactly 2 edges.

\begin{proposition}
    \label{prop:hardness-tpp}
    The Track Packing problem is \NP-hard, even on shadow-complete instances, in which each endvertex of a path is connected to a vertex of a different path.
\end{proposition}
\begin{proof}
    Let $G'=(V', E')$ be an instance of the $P_2$-packing problem, where $V' = \{ x_1, x_2, ..., x_n \}$.
    From $G'$ we construct a shadow-complete instance $G=(V, \mcP \cup L)$ of TPP, in which each endvertex of a path is connected to a vertex of a different path, as follows. 
    For each $x_i$, $1 \leq i \leq n$, we add three vertices $u_i, v_i, w_i$ to $G$ and define the path $P_i = u_i w_i v_i$ with endvertices $u_i$ and $v_i$ and interior vertex $w_i$.
    Furthermore, we additionally add two special paths $P_{n+1}= u_{n+1} w_{n+1} v_{n+1}$ and $P_{n+2}=u_{n+2} w_{n+2} v_{n+2}$, where again the endvertices are $u_{n+1}, v_{n+1}$ and $u_{n+2}, v_{n+2}$, respectively, and the interior vertices are $w_{n+1}$, and $w_{n+2}$, respectively.
    We set $\mcP = \{P_1, P_2, ..., P_n, P_{n+1}, P_{n+2}\}$ and define $V = \bigcup_{1 \leq i \leq n+2} \{u_i, v_i, w_i\}$. 
    
    The link set $L$ is defined as follows. For each $x_i x_j \in E'$, $i \neq j$, we add the links $v_i w_j$ and $v_j w_i$ to~$L$, as well as the link $w_i w_j$. Furthermore, for each $1 \leq i \leq n$ we add $v_i u_i$ to~$L$.
    Additionally, for $1 \leq i \leq n$, we add $u_i w_{n+1}$, $w_i w_{n+1}$, $u_i w_{n+2}$, and $w_i w_{n+2}$. Also, we add the seven links $v_{n+1} v_{n+2}$, $u_{n+1} u_{n+2}$, $v_{n+1} w_{n+2}$, $w_{n+1} v_{n+2}$, $w_{n+1} w_{n+2}$, $u_{n+1} w_{n+2}$, and $w_{n+1} u_{n+2}$.
    
    This finishes the construction of $G$. Note that indeed $G$ is shadow-complete. Further, observe that for each $1 \leq i \leq n$ we have that $u_i$ is only incident in $L$ to $v_i$ and the vertices $V(P_{n+1})$ and $V(P_{n+2})$.
    Also, note that each endvertex of a path is connected to a vertex of a different path.

    We now show that from a solution $X'$ to $G'$ consisting of vertex-disjoint paths, we can construct in polynomial-time a solution $X$ to $G$ consisting of disjoint tracks such that $|X'|+2 = |X|$, and vice versa. This then proves the \NP-hardness of TPP, since the $P_2$-packing problem is \NP-hard~\cite{kirkpatrick1978completeness}.

    First, consider some solution $X'$ to $G'$. We compute a solution $X$ to $G$ such that $|X| = |X'|+2$.
    For each $P_2$ $x_i x_j x_k \in X'$ we add the track with link set $\{v_i w_j, v_j u_j, w_j v_k \}$ to $X$. Note that since any two $P_2$'s in $X'$ are vertex-disjoint, also any two tracks added so far to $X$ are vertex-disjoint.
    Furthermore, we add the two (disjoint) tracks $v_{n+1} v_{n+2}$ and $u_{n+1} u_{n+2}$ of length 1 to $X$.
    Again, note that these two additional tracks are not intersecting any previous tracks added to $X$.
    Hence, $X$ is a feasible solution to $G$.
    Furthermore, we clearly have $|X| = |X'|+2$, which proves the statement.
    
    Next, consider some solution $X$ to $G$. We compute a solution $X'$ to $G'$ such that $|X'|+2 = |X|$.
    We first show that withoutloss of generality we can assume that the two tracks $v_{n+1} v_{n+2}$ and $u_{n+1} u_{n+2}$ of length 1 are contained in $X$.
    If not, observe that there are at most two tracks incident to $V(P_{n+1}) \cup V(P_{n+2})$, since $\{v_{n+1}, v_{n+2}, u_{n+1}, u_{n+2}\}$ is only incident to $V(P_{n+1}) \cup V(P_{n+2})$.
    Let $T_1$ and $T_2$ be the (at most) two tracks incident to (potentially a subset of) $V(P_{n+1})$ and $V(P_{n+2})$. Setting $Y = (X \setminus \{T_1, T_2\}) \cup \{ v_{n+1} v_{n+2}, u_{n+1} u_{n+2}\}$ yields a new set of disjoint tracks with the desired property and $|Y| \geq |X|$.
    We next show that without loss of generality we can assume that each track in $X$ consists of 3 links, except for the tracks $v_{n+1} v_{n+2}$ and $u_{n+1} u_{n+2}$.
    Observe that there are no edges between endvertices of distinct paths in $\mcP$, except for the links $v_{n+1} v_{n+2}$ and $u_{n+1} u_{n+2}$. 
    Hence, there are no other tracks consisting of only 1 link.
    Consider a track $T$ of at least 5 links (recall that a track always has an odd number of links).
    Hence, $|Q(T)| \geq 3$. Recall that $Q(T)$ is a path, where the endvertices of $Q(T)$ are endvertices of some paths in $\mcP$ and each interior vertex of $Q(T)$ is an interior vertex of some path in $\mcP$.
    Let the first two links in the path $Q(T)$ be $v_i w_j$ and $w_j w_k$.
    By construction of $G$, if $w_j w_k \in L$, then also $v_j w_k \in L$ and $v_k w_j \in L$.
    Now, observe that $T' = \{v_i w_j, v_j u_j, w_j v_k \}$ is a track in $G$ such that $V(T') \subseteq V(T)$.
    Hence, by replacing each track $T$ with its corresponding track $T'$ of length 3, we can construct a new set of vertex disjoint tracks of same cardinality, such that each track consists of exactly 3 links, except for the tracks $v_{n+1} v_{n+2}$ and $u_{n+1} u_{n+2}$.
    Now, for each track $T = \{v_i w_j, v_j u_j, w_j v_k \}$ containing 3 links we add the path $x_i x_j x_k$ to $X'$. Note that such a track excludes the two special tracks $v_{n+1} v_{n+2}$ and $u_{n+1} u_{n+2}$.
    Since $X$ consists of disjoint tracks and each $u_i$ is only adjacent to $v_i$ in $G \setminus (V(P_{n+1}) \cup V(P_{n+2})$ for each $1 \leq i \leq n$, also all $P_2$'s in $X'$ must be vertex-disjoint. Furthermore, by construction, we have $|X'|+2=|X|$, which proves our statement.
    \end{proof}

Note that due to the above relationship between \tecpc and \tpp established in \Cref{lem:relaxation:relationship:packing-to-covering} and \Cref{lem:relaxation:relationship:covering-to-packing}, \Cref{prop:hardness-tpp} also implies \NP-hardness for \tecpc.

\begin{proposition}
    \label{prop:hardness-2ECPC}
    The 2-Edge Cover with Path Constraints is \NP-hard.
\end{proposition}

\subsection{Obtaining a good Starting Solution via a factor-revealing LP}
In this section we obtain a good starting solution via approximation algorithms for \tpp.
The bound then is obtained using a factor-revealing LP.
Finally, we show that this solution satisfies all three Invariants \ref{invariant:credit}-\ref{invariant:block-size}.

First, we describe three easy approximation algorithms for \tpp.
These algorithms are obtained by viewing \tpp as a \emph{set-packing problem}.
In a set-packing problem, we are given a set of elements $U$ and a set of subsets $\mathcal{S} \subseteq 2^U$ of $U$.
The task is to compute a set $X \subseteq \mathcal{S}$ of maximum cardinality such that for any two sets $S, S' \in X$ we have $S \cap S' = \emptyset$.
For some instance $I$ of Set Packing let $k \coloneqq \max_{S \in \mathcal{S}} |S|$ be the maximum size of a set in $\mathcal{S}$.
The state-of-the-art polynomial-time approximation algorithms for set-packing problem achieve an approximation factor of $\frac{3}{1 + k + \varepsilon}$ for any constant $\varepsilon > 0$~\cite{cygan2013improved, furer2014approximating}.
\tpp can be viewed as a set-packing problem by considering each $v \in V^*(\mcP)$ as an element and each track as a set, which exactly contains those endvertices to which the track is incident to.
However, since the number of tracks can be exponential, we restrict only to tracks of certain constant size.

Throughout this section we fix an optimum solution $\OPT_P$ for the instance of TPP and assume that $\OPT_P = \bigcup_{i=1}^{n} O_i$, where $O_i$ is the set of tracks $T$ in the optimum solution such that $|Q(T)|=i$.
We define $\omega_i = |O_i|$.
Hence, $|\OPT_P| = \sum_{i=1}^{n} \omega_i$.
We additionally let $\omega_0$ be the number of vertices of $V^*(\mcP)$ that are not incident to any track in $\OPT_P$.
Since each endvertex of $V^*(\mcP)$ can be incident to at most one track and each track $T$ with $|Q(T)|=i$ is incident to $2i$ distinct endvertices of $V^*(\mcP)$, we have that $\omega_0 + \sum_{i=1}^{n} 2 i \cdot \omega_i = 2 |\mcP|$.
Finally, let $\omega_1^e \leq \omega_1$ be the number of tracks $T$ contained in $\OPT_P$ such that $|Q(T)|=1$ and $T$ corresponds to an expensive link.

\paragraph*{Algorithm A.}
In Algorithm~A we first only consider tracks of size 1, which are just links from $\bar{L}$.
In the first step we compute a maximum number of disjoint tracks of size 1, i.e., we compute a maximum matching $A_1 \subseteq \bar{L}$, which can be done in polynomial time.

For the second step, let $R_A$ be the set of endvertices of $V^*(\mcp)$ that are not incident to some link of $A_1$. 
Let $I_{R_A}$ be the resulting set-packing instance where we have one element for each vertex in $R_A$ and consider only those tracks $T$ that are only incident to vertices from $R_A$ and additionally satisfy $Q(T) = 2$ (i.e., $Q(T)$ contains exactly 2 links).
We then compute a solution $A_2$ to $I_{R_A}$ using the $\frac{3}{5+\varepsilon}$-approximation algorithm from~\cite{furer2014approximating} and output the solution $S_A = A_1 \cup A_2$.

\paragraph*{Algorithm B.}
Algorithm~B is similar to Algorithm~A, but we compute a different matching in the first step.
We aim at 'guessing' the value $\omega_1^e$, i.e., how many tracks of size 1 are contained in the optimum solution $\OPT_P$ that correspond to an expensive link. 
There are at most $|\mcP|$ many guesses, hence we can simply try all possible values.

Given a guess of $\omega_1^e$, we compute a subset $B_1 \subseteq \bar{L}$ of maximum size, such that $B_1$ contains at most $\omega_1^e$ many expensive links.
This can be done by guessing the size $k$ of the maximum matching with at most $\omega_1^e$ many expensive links and giving each expensive link in $\bar{L}$ a cost of $1$ and any other link of $\bar{L}$ a cost of 0 and computing a minimum-cost matching of size exactly $k$. 
The maximum $k$ for which the cost of the solution is at most $\omega_1^e$ is then the desired solution.
This gives us the edge set $B_1$, which corresponds to a maximum set of disjoint tracks of length 1 with at most $\omega_1^e$ many expensive links.
Note that a minimum-cost matching of size exactly $k$ in a weighted graph with $n$ vertices (and non-negative weights) can be found by reducing it to a minimum-cost perfect matching problem by adding $n - 2k$ dummy vertices and connecting each dummy vertex to every original vertex with an edge of cost $0$.

The second step in Algorithm~B is identical to the second step of Algorithm~B:
Let $R_B$ be the set of endvertices of $V^*(\mcp)$ that are not incident to some link of $B_1$. 
Let $I_{R_B}$ be the resulting set-packing instance where we have one element for each vertex in $R_B$ and consider only those tracks $T$ that are only incident to vertices from $R_B$ and additionally satisfy $Q(T) = 2$ (i.e., $Q(T)$ contains exactly 2 links).
We then compute a solution $B_2$ to $I_{R_B}$ using the $\frac{3}{5+\varepsilon}$-approximation algorithm from~\cite{furer2014approximating} and output the solution $S_B = B_1 \cup B_2$.

\paragraph*{Algorithm C.}
In the third algorithm, we directly consider the corresponding set-packing problem in which we only consider tracks $T$ with $|Q(T)| \leq 2$ (i.e., tracks $T$ such that $|Q(T)| = 1$ or $|Q(T)| = 2$).
We then compute a solution $S_C$ to this set-packing instance using the $\frac{3}{5+\varepsilon}$-approximation algorithm from~\cite{furer2014approximating} and output $S_C$. We set $S_B = C_1 \cup C_2$, where $C_i$ is the set of tracks $T$ such that $|Q(T)| = i$.

\paragraph*{Final Algorithm.}
The final solution is then computed as follows. 
For a solution $H = (V, E(\mathcal{P}) \cup S)$ we defined $\cost(H) \coloneqq \credits(H) + |S|$.
For each guess $i \in \{0, 1, ..., |\mcP|\}$ of $\omega_1^e$ in Algorithm~B we have a solution $H_B^i = (V, E(\mathcal{P}) \cup S_B^i)$ and for Algorithm A and C we have one solution $H_A = (V, E(\mathcal{P}) \cup S_A)$ and $H_C = (V, E(\mathcal{P}) \cup S_C)$, respectively.
For each of the $|\mcP| + 3$ many solutions, we then compute their respective cost as defined above and output the solution with lowest cost.
We let $H = (V, E(\mathcal{P}) \cup S)$ be this final solution.

We note that in fact one of the solutions computed in Algorithm B will be identical to the solution computed in Algorithm A. 
This is true if the guess of $\omega_1^e$ is $0$ and hence in Algorithm B, we simply compute a maximum matching of $\bar{L}$. 
However, for notational purposes it will be more convenient to also consider Algorithm A with its solution $S_A$.

\paragraph*{Bounding the cost of the solution.}
We now bound $\cost(H)$.
We bound the number of links in the optimum solution as follows.

\begin{lemma}
    \label{lem:starting-solution:bound:opt}
    We have $\opt_C \geq 2|\mcp| - \sum_{i \geq 1} \omega_i$. 
\end{lemma}

\begin{proof}
    By Lemmas~\ref{lem:relaxation:relationship:packing-to-covering} and~\ref{lem:relaxation:relationship:covering-to-packing} we have that $\opt_C = 2 |\mcp| - \opt_P$.
    By definition, we have that $\opt_P = \sum_{i \geq 1} \omega_i$ and therefore obtain the result.
\end{proof}

Let $S_B$ be the solution from Algorithm B which corresponds to the correct guess of $\omega_1^e$. From now on, we only need this solution from Algorithm B and do not consider the solutions for the other guesses of $\omega_1^e$.
Let $S_A = A_1 \cup A_2$, where $A_i$ is the set of tracks $T$ from $S_A$ such that $|Q(T)|=i$ and define $\alpha_i = |S_i|$ for $i=1, 2$.
Similarly, we define $S_B = B_1 \cup B_2$ and $S_C = C_1 \cup C_2$ and set $\beta_i = |B_i|$ and $\gamma_i = |C_i|$ for $i=1, 2$.
Furthermore, let $\alpha_0, \beta_0, \gamma_0$ be the number of vertices from $V^*(\mcP)$ that are not incident to some track of $S_A$, $S_B$ or $S_C$, respectively. 
Note that such vertices correspond to leaves in $S_A$ and $S_B$, respectively.
Additionally, let $\alpha_0^e \leq \alpha_0$, $\beta_0^e \leq \beta_0$, and $\gamma_0^e \leq \gamma_0$ be the number of such leaves as above that are expensive w.r.t.\ $S_A$, $S_B$ and $S_C$, respectively.
Finally, for $j= 0, 1, 2$ let $\alpha_1^j$ be the number of tracks $T = u v$ of size 1 in $S_A$ such that exactly $j$ vertices of $\{u, v\}$ intersect some track of size 1 in $\OPT_P$.
Similarly, for $j= 0, 1, 2$ we define $\beta_1^j$.
Note that $\alpha_1 = \alpha_1^0 + \alpha_1^1 + \alpha_1^2$ and $\beta_1 = \beta_1^0 + \beta_1^1 + \beta_1^2$.
First, we prove the following bound on the cost of a solution by giving it certain credits.

\begin{restatable}[]{lemma}{lemstartingsolutionboundcredits}
    \label{lem:starting-solution:bound:credits}
    Let $S$ be a solution computed by Algorithm A, B, or C.
    If each track of length $1$ receives a credit of $\frac 34$, each track of length $3$ a credit of $2$, each lonely leaf vertex a credit of $\frac 32$ if it is not expensive and $\frac 74$ if it is expensive, and each degenerate path an additional credit of $\frac 12$, then we can redistribute the credits such that the credit rules (A)-(E) are satisfied.
    Furthermore, Invariants~\ref{invariant:degree-lonely-vertex} and~\ref{invariant:block-size} are satisfied.
\end{restatable}

\begin{proof}
    Let $H= (V, E(\mathcal{P}) \cup S)$. 
    First, consider some connected component $C$ which is 2EC. 
    If $C$ contains at least one track of size 3, then the component is large and the total sum of credits assigned to this component is at least 2.
    If $C$ contains at least 3 tracks of size 1, then the component is also large and the total sum of credits assigned to this component is also at least 2.
    Else, the component consists of exactly two tracks of size 1 and the component is small.
    If none of these paths is degenerate, then the total sum of credits assigned to this component is exactly $\frac 32$, as desired.
    Otherwise, this component receives a credit of at least 2.
    This proves part (E).

    Next, consider a complex component $C$ containing at least 1 bridge. 
    Observe that each block in $C$ is simple, since it must be part of some track. 
    Hence, this proves (D).
    Furthermore, note that this also proves \Cref{invariant:block-size}.
    Additionally, observe that after contracting each block of $C$ to a single vertex we obtain a path. Furthermore, in $H^C$ no two simple blocks are adjacent to each other, since we only consider tracks containing 1 or 3 links. Hence, \Cref{invariant:degree-lonely-vertex} is satisfied.

    Finally, we prove that the credits are satisfied for the complex component $C$.
    There are 2 lonely leaf vertices in $C$, both receiving a credit of $\frac 32$, plus an additional credit of $\frac 14$ if the leaf is expensive.
    We redistribute this credit such that each lonely leaf vertex receives a credit of $1$ (or $1.25$ if it is expensive), and additionally the complex component $C$ receives a credit of $1$. 
    The first part proves (A)(i) and the latter part proves (C).
    Each bridge that corresponds to a track of size 1 keeps its credit of $\frac 34$.
    Each bridge that corresponds to a track of size 3 keeps a credit of $\frac 34$, leaving another credit of $\frac 12$, which we distribute to the 2 lonely leaf vertices of $C$ equally, i.e., each lonely leaf vertex receives an additional credit of $\frac 14$ if $C$ contains at least one track of size 3.
    This proves~(B).
    Therefore, each lonely vertex receives an additional credit of $\frac 14$ if $C$ contains a track of size 3 (and therefore a simple block). Observe that this proves (A)(ii).
\end{proof}

Next, we show that also \Cref{invariant:credit} is satisfied. To do so, we establish the following lemmas.

\begin{restatable}[]{lemma}{lemstartingsolutionboundABC}
    \label{lem:starting-solution:bound:ALGA-B-C}
    If we assign to $H_A$, $H_B$ and $H_C$ credits according to \Cref{lem:starting-solution:bound:credits}, then we have $\cost(H_A) \leq  3 |\mcp| - \frac 54 \alpha_1 - \alpha_2  + \frac 14 \alpha_0^e + \frac{1}{2}\varepsilon' |\mcP|$, $\cost(H_B) \leq 3 |\mcp| - \frac 54 \beta_1 - \beta_2  + \frac 14 \beta_0^e + \frac{1}{2}\varepsilon' |\mcP|$ and $\cost(H_C) \leq 3 |\mcp| - \frac 54 \gamma_1 - \gamma_2  + \frac 14 \gamma_0^e + \frac{1}{2}\varepsilon' |\mcP|$, where $\varepsilon' |\mcP|$ is the number of degenerate paths in $G$.
\end{restatable}

\begin{proof}
    We prove the statement for $H_A$; for $H_B$ and $H_C$ it is the same.
    Note that the number of lonely leaf vertices in $H_A$ is bounded by $2|\mcp| - 2 \alpha_1 - 4 \alpha_2$.
    Furthermore, we assign a credit of $\frac 34$ for each of the $\alpha_1$ many tracks containing one link, a credit of $2$ for each of the $\alpha_2$ many tracks containing three links and a credit of $\frac{3}{2}$ for each lonely leaf vertex of $H_A$ and an additional credit of $\frac{1}{4}$ for each expensive lonely leaf vertex.
    Finally, observe that we give an additional credit of $\frac{1}{2}$ to each degenerate path.
    Hence, we can bound $\cost(H_A)$ by
    \begin{align*}\cost(H_A) = |S_A| + \credits(H_A) & \ \leq \alpha_1 + 3 \alpha_2 + \frac 32 (2|\mcp| - 2 \alpha_1 - 4 \alpha_2) + \frac 34 \alpha_1 + 2 \alpha_2 + \frac 14 \alpha_0^e + \frac{1}{2}\varepsilon' |\mcP| \\
    & \ = 3 |\mcp| - \frac 54 \alpha_1 - \alpha_2  + \frac 14 \alpha_0^e + \frac{1}{2}\varepsilon' |\mcP|  \ \ , 
    \end{align*}  
    as desired. 
\end{proof}

Next, we obtain some further constraints on the variables we defined earlier.

\begin{restatable}[]{lemma}{lemstartingsolutionboundABCparameters}
    \label{lem:starting-solution:bound:ALGA-B-C-parameters}
    The following inequalities are valid:
    \begin{multicols}{2}
    \begin{enumerate}
        \item $\alpha_1 \geq \omega_1, \beta_1, \gamma_1$ , 
        \item $\alpha_1 \geq \omega_1 + \alpha_1^0$ , 
        \item $\alpha_1 \geq \omega_1 + \beta_1^0$ ,
        \item $\beta_1 \geq \omega_1$ ,
        \item $\alpha_1 \geq \alpha_0^e$ ,
        \item $\beta_1 \geq \beta_0^e$ ,
        \item $\gamma_1 \geq \gamma_0^e$ ,
        \item $\omega_0 \geq \beta_0^e$ ,
        \item $\alpha_1 = \alpha_1^0 + \alpha_1^1 + \alpha_1^2$ ,
        \item $\beta_1 = \beta_1^0 + \beta_1^1 + \beta_1^2$ ,
        \item $\alpha_2 \geq \frac{3}{5 + \eps} \cdot (\omega_2 - 2 \alpha_1^0 -\alpha_1^1)$ , 
        \item $\beta_2 \geq \frac{3}{5 + \eps} \cdot (\omega_2 - 2 \beta_1^0 -\beta_1^1)$ , and 
        \item $\gamma_1 + \gamma_2 \geq \frac{3}{5+\varepsilon}\cdot(\omega_1 + \omega_2)$ .
    \end{enumerate}
    \end{multicols}
\end{restatable}

\begin{proof}
We prove the inequalities one by one.
Recall that $A_1, B_1, C_1, O_1 \subseteq \bar{L}$ are simply matchings.
Since in Algorithm A we compute a maximum matching of $\bar{L}$, the first set of inequalities follows.
The second inequality follows since $\alpha_1^0$ denotes the number of matching edges in $A_1$ that are disjoint from matching edges in $O_1$.
Similarly, we obtain the third inequality as those matching edges contributing to $\beta_1^0$ and $\omega_1$ together form a matching, by the definition of $\beta_1^0$.
For the fourth inequality, observe that $O_1$ is a matching with $\omega_1^e$ many expensive edges. Since in Algorithm B we compute a maximum matching with at most $\omega_1^e$ many expensive edges, Inequality 4 has to be satisfied.

We next prove Inequality 5; Inequalities 6 and 7 follow similarly.
Recall the definition of $\alpha_0^e$: This is the number of expensive lonely leaf vertices in $H_A$. 
We further recall the definition of an expensive link and expensive lonely vertex:
Consider two paths $P, P' \in \mathcal{P}$ with endvertices $u, v$ and $u', v'$, respectively, and assume $u u' \in E(H_A)$.
If $N_G(v) \subseteq (V(P) \cup V(P')) \setminus \{v'\}$, we call the link $u u'$ an expensive link.
If $v$ is also a leaf in $H_A$, we say that $v$ an expensive leaf vertex.
Hence, if there is an expensive lonely leaf vertex in $H_A$, then there must be an expensive link $uu'$ in $H_A$, which corresponds to some link contributing to $\alpha_1$.
Hence, $uu'$ can create at most two expensive leaf vertices, i.e., the vertices $v$ and $v'$.
This already implies $\alpha_0^e \leq 2 \alpha_1$. 
We now strengthen this inequality.
Assume that $v$ and $v'$ are both expensive leaf vertices in $H_A$.
By the definition of a lonely leaf vertex, we have $N_G(v) \subseteq (V(P) \cup V(P')) \setminus \{v'\}$ and $N_G(v') \subseteq (V(P) \cup V(P')) \setminus \{v\}$.
Hence, any feasible solution to the \PAP instance $G$ must include at least two links from within $G[V(P) \cup V(P')]$: one link incident to $v$ and one incident to $v'$. Since $v v' \notin L$, by the definition of an expensive link, these two links must be distinct.  
Furthermore, note that since $G$ is structured, there must be a link $e_v$ from $v$ to $V(P')$ in $L$ and a link $e_{v'}$ from $v'$ to $V(P)$ in $L$. 
Now observe that $X = (V(P) \cup V(P'), E(P) \cup E(P') \cup e_v \cup e_{v'} \cup u u')$ is 2-edge connected, and hence $X$ is a strongly $\frac{3}{2}$-contractible subgraph.
This contradicts our assumption on $G$ being structured.

Inequality 8 follows from the definition of Algorithm B. 
The number of expensive links in $\OPT_P$ is exactly the number of expensive links in $S_B$. 
Since each expensive link can contribute to at most one expensive lonely leaf vertex (see previous argumentation for Inequality 5), we have that the number of expensive leaf vertices in $S_B$ is equal to the number of expensive links in $\OPT_P$. 
It remains to show that for each expensive link in $\OPT_P$ there must be one unique lonely leaf vertex contributing to $\omega_0$.
Consider two paths $P, P' \in \mathcal{P}$ with endvertices $u, v$ and $u', v'$, respectively, and assume $u u'$ is an expensive link such that $N_G(v) \subseteq (V(P) \cup V(P')) \setminus \{v'\}$.
Assume $u u'$ is a track contained in the optimum solution. 
Assume for a contradiction that $v$ is part of some other track $T$ in $\OPT_P$.
But then $T$ must be incident to either $u$ or $u'$, since $N_G(v) \subseteq (V(P) \cup V(P')) \setminus \{v'\}$. But then $T$ is not vertex-disjoint to the track $u u'$.

The Inequalities 9 and 10 are simply the definitions of the corresponding variables.

We next prove Inequality 11; Inequality 12 follows similarly.
In Algorithm A, in order to compute $A_2$, we compute a $\frac{3}{5 + \eps}$-approximate solution to the Set-Packing instance defined by removing all elements corresponding to vertices incident to $A_1$.
Consider the tracks of $\OPT_P$ which are in $O_2$. 
A track from $A_1$ can be incident to at most 2 distinct tracks of $O_2$, since these are simply links.
However, by the definition of $\alpha_1^j$, a track from $A_1$ contributing to $\alpha_1^j$ can be incident to at most $2-j$ distinct tracks of $O_2$. 
Hence, we obtain the desired inequality.

Finally, consider Inequality 13. This inequality simply follows from the fact that we compute a $\frac{3}{5 + \eps}$-approximate solution to the Set-Packing instance.
\end{proof}

Recall that $H= (V, E(\mathcal{P}) \cup S)$, where $S$ is the best solution output by Algorithms A, B, and C w.r.t.\ their cost.
Note that the following lemma proves that $H$ satisfies \Cref{invariant:credit}.

\begin{restatable}[]{lemma}{lemstartingsolutionboundfinal}
    \label{lem:starting-solution:bound:final}
    For $\varepsilon \leq 0.001$ we have $\cost(H) \leq \apxrn \cdot \OPT(G)$.
\end{restatable}

\begin{proof}
    
    Let $\varepsilon' |\mcP| $ be the number of degenerate paths. Note that $\OPT(G) \geq |\mcP|$, and therefore $\varepsilon' |\mcP| \leq \varepsilon' \OPT(G)$.
    We prove the statement by giving a factor-revealing LP, which proves $\cost(H) \leq \rho \cdot \OPT(G)$ if the polyhedron is empty for a specific value of $\rho$. The parameter $\rho$ will then give the desired approximation ratio.
    The factor-revealing LP is as follows.
    We use the same variables as before, but for simplicity we divide all constraints by $|\mcp|$.
    For ease of exposition, we do not replace each variable with a new variable, which is simply the old variable divided by $|\mcp|$. Instead, we simply use the same variable names.
    
    We use the constraints we obtained by the preceding discussion and lemmas.
    That is, we use the constraints obtained from \Cref{lem:starting-solution:bound:opt}, \Cref{lem:starting-solution:bound:credits}, \Cref{lem:starting-solution:bound:ALGA-B-C}, \Cref{lem:starting-solution:bound:ALGA-B-C-parameters}, and constraints that we naturally obtained from the definitions of the variables.
    Additionally, we simplify the constraints on the optimum solution value as follows. 
    From \Cref{lem:starting-solution:bound:opt} we get $\opt_C \geq 2|\mcp| - \sum_{i \geq 1} \omega_i$. 
    Furthermore, by definition, we have $\omega_0 + \sum_{i=1}^{n} 2 i \cdot \omega_i = 2$ (recall that we divide all numbers by $|\mcP|$).
    Since in all the constraints we used there is no other bound on $\gamma_3$, both inequalities can be simplified to $\opt_C \geq 2 - \omega_1 - \omega_2 - \omega_3$ and $\omega_0 + 2 \omega_1 + 4 \omega_2 + 6 \omega_3 \leq 2$.
    We then get the following polyhedron.
    \begin{align*}
		   \opt_C & \geq 2 - \omega_1 - \omega_2 - \omega_3 \\
         \cost(H) & \geq \rho \cdot \opt_C\\
         \cost(H) & \leq  3 - \frac 54 \alpha_1 - \alpha_2  + \frac 14 \alpha_0^e + \frac{1}{2} \varepsilon' \\
         \cost(H) & \leq  3 - \frac 54 \beta_1 - \beta_2  + \frac 14 \beta_0^e + \frac{1}{2} \varepsilon'\\
         \cost(H) & \leq  3 - \frac 54 \gamma_1 - \gamma_2  + \frac 14 \gamma_0^e + \frac{1}{2} \varepsilon'\\
         \omega_0 + 2 \omega_1 + 4 \omega_2 + 6 \omega_3 & \leq 2\\
         \alpha_1 & \geq \omega_1, \beta_1, \gamma_1 \\ 
         \alpha_1 & \geq \omega_1 + \alpha_1^0 \\ 
         \alpha_1 & \geq \omega_1 + \beta_1^0 \\
         \beta_1 & \geq \omega_1 \\
         \alpha_1 & \geq \alpha_0^e \\
         \beta_1 & \geq \beta_0^e \\
         \gamma_1 & \geq \gamma_0^e \\
         \omega_0 & \geq \beta_0^e \\
         \alpha_1 & = \alpha_1^0 + \alpha_1^1 + \alpha_1^2 \\
         \beta_1 & = \beta_1^0 + \beta_1^1 + \beta_1^2 \\
         \alpha_2 & \geq \frac{3}{5 + \eps} \cdot (\omega_2 - 2 \alpha_1^0 -\alpha_1^1) \\ 
         \beta_2 & \geq \frac{3}{5 + \eps} \cdot (\omega_2 - 2 \beta_1^0 -\beta_1^1) \\
         \gamma_1  + \gamma_2 & \geq \frac{3}{5+\varepsilon}\cdot(\omega_1 + \omega_2) \\
         \omega_0, \omega_1, \omega_2, \omega_3 & \geq 0 \\ 
         \alpha_0^e, \alpha_1, \alpha_1^0, \alpha_1^1, \alpha_1^2, \alpha_2 & \geq 0 \\
         \beta_0^e, \beta_1, \beta_1^0, \beta_1^1, \beta_1^2, \beta_2 & \geq 0 \\ 
         \gamma_0^e, \gamma_1, \gamma_1^0, \gamma_1^1, \gamma_1^2, \gamma_2 & \geq 0 \ .
	\end{align*}
    By plugging this into a standard solver and setting $\rho \geq \apxrn$, $\eps \leq 0.001$ and $\varepsilon' \leq 0.0001$, we obtain that the polyhedron is empty and hence we obtain the result.
\end{proof}

\section{Bridge Covering}
\label{sec:bridge-covering}
In the previous section we showed that we can compute an initial solution $S$ such that $H=(V, E(\mathcal{P}) \cup S)$ satisfies Invariants~\ref{invariant:credit}-\ref{invariant:block-size}.
In this section, we show that we can transform $H$ into a solution that does not contain bridges and maintains Invariants~\ref{invariant:credit}-\ref{invariant:block-size}.
In particular, in this section we show the following lemma. 
\lembridgecoveringmain*

An outline of this section was given in \Cref{sec:overview:bridge-covering}. 
For convenience we restate the definitions and lemmas.
Throughout this section we assume that $H$ contains some complex component $C$, as otherwise there is nothing to prove.
We prove \Cref{lem:bridge-covering:main} by providing the following lemma below. 
It is straightforward to see that exhaustively applying this lemma proves \Cref{lem:bridge-covering:main}.

\begin{restatable}[]{lemma}{lembridgecoveringmainiterative}
    \label{lem:bridge-covering:main-iterative}
    Let $G$ be some structured instance of \PAP and $H=(V, E(\mcP) \cup S)$ be a subgraph of $G$ satisfying the Invariants~\ref{invariant:credit}-\ref{invariant:block-size} and containing at least one bridge.
    In polynomial time we can compute a subgraph $H'=(V, E(\mcP) \cup S')$ of $G$ satisfying Invariants~\ref{invariant:credit}-\ref{invariant:block-size} such that $H'$ contains fewer connected components or the number of connected components in $H'$ is at most the number of connected components of $H$ and $H'$ contains fewer bridges than $H$.
\end{restatable}

\begin{algorithm}[t]
\DontPrintSemicolon
\caption{Bridge-Covering}
\label{alg:bridge-covering}
Compute a starting solution $S$ according to \Cref{lem:main:starting-solution-satisfies-invariants} and initialize $H = (V, E(\mathcal{P}) \cup S)$.\\
\While{$H$ contains at least one bridge}{
    \If{$H$ contains two lonely leaf vertices $u$ and $v$ of different complex components and $uv \in L$}{
        \Return $H \cup uv$. \label{alg:bridge-covering:connecting-leaves}
    }
    \If{$H$ contains a complex component with a simple block vertex}{
        compute $R$ and $F$ according to \Cref{lem:bridge-covering:simple} and \Return $H \cup R \setminus F$. \label{alg:bridge-covering:simple}
    }
    \ElseIf{$H$ contains a non-trivial complex component with a non-expensive lonely leaf vertex}{
        compute $R$ and $F$ according to \Cref{lem:bridge-covering:non-expensive} and \Return $H \cup R \setminus F$. \label{alg:bridge-covering:non-expensive}
    }
    \ElseIf{$H$ contains a complex component with a non-simple block vertex}{
        compute $R$ and $F$ according to \Cref{lem:bridge-covering:large-block} and \Return $H \cup R$.\label{alg:bridge-covering:block}
    }
    \ElseIf{$H$ contains a non-trivial complex component containing at least one bridge that is in $L$}{
        compute $R$ according to \Cref{lem:bridge-covering:no-large-component} and \Return $H \cup R$. \label{alg:bridge-covering:no-large}
    }
    \Else{
        compute $R$ according to \Cref{lem:bridge-covering:no-large-component:only-trivial} and \Return $H \cup R$. \label{alg:bridge-covering:only-trivial}
    }
}
\Return $H$
\end{algorithm}

It remains to prove \Cref{lem:bridge-covering:main-iterative}.
Our algorithm is summarized in \Cref{alg:bridge-covering}.
Before we proceed with the main contribution of this section, in which we cover the bridges, we need one other lemma: if there is a link $e \in L$ connecting two lonely leaf vertices of different complex components, then we can add $e$ to $H$ while maintaining the Invariants~\ref{invariant:credit}-\ref{invariant:block-size} and $H'$ has fewer connected components than~$H$.

\begin{restatable}[]{lemma}{lembridgecoveringlinksbetweenleafs}
    \label{lem:bridge-covering:links-between-leafs}
    Let $u, v$ be lonely leaf vertices of different complex components. 
    If $uv \in L$, then $H' = H \cup uv$ satisfies the Invariants~\ref{invariant:credit}-\ref{invariant:block-size} and $H'$ has fewer connected components than $H$.
\end{restatable}

\begin{proof}
    Let $H =(V, E(\mcP) \cup S)$ and let $H' =(V, E(\mcP) \cup S')$, where $S' = S \cup uv$.
    Clearly, $H'$ has exactly one component less than $H$. Furthermore, it is clear that \Cref{invariant:block-size} and \Cref{invariant:degree-lonely-vertex} are satisfied in $H'$. It remains to prove \Cref{invariant:credit}.
    Since $u$ and $v$ are lonely leaf vertices of different complex components, say $C_u$ and $C_v$, respectively, $u$ and $v$ each receive a credit of at least 1 and $C_u$ and $C_v$ each receive a credit of 1 in $H$. 
    In $H'$, neither $u$ nor $v$ is a leaf vertex anymore and $C_u$ and $C_v$ form one complex component. 
    However, the newly introduced bridge $uv$ needs a credit of $\frac 34$. 
    Hence, $\credits(H) - \credits(H') \geq \frac 94$. Clearly, we have $|S'| = |S|+1$. Combining both bounds, we have $\cost(H') \leq \cost(H)$. 
    Therefore, since $H$ satisfied \Cref{invariant:credit}, also $H'$ satisfies it, and hence we obtain the result.
\end{proof}

In the remainder we assume that the conditions of \Cref{lem:bridge-covering:links-between-leafs} do not apply.

Intuitively, as long as our current solution $H$ contains complex components, i.e., it contains bridges, our goal is to add certain links to $H$ in order to 'cover' the bridges and, at the same time, maintain the Invariants~\ref{invariant:credit}-\ref{invariant:block-size}.
To do so, it will be convenient to add so-called pseudo-ears to $H$, defined as follows. Similar definitions have been used in~\cite{cheriyan2023improved} and~\cite{garg2023matching}.

\defpseudoear*

An example for a pseudo-ear is given in \Cref{fig:bridge-covering}.
We first show that a pseudo-ear $Q^{uv}$ can be found in polynomial time, if one exists.

\begin{lemma}
    \label{lem:bridge-coverin:pseudo-ear:find-polytime}
    Let $C$ be a complex component of $H$, let $u, v \in V(C)$ be distinct vertices of $C$.
    Then we can find a pseudo-ear $Q^{uv}$ in polynomial time, if one exists.
\end{lemma}

\begin{proof}
    Simply compute $G^C$ and delete all vertices of $C$, except for $u$ and $v$.
    Let $G'$ be this resulting graph.
    Now observe that any simple path from $u$ to $v$ in $G'$ corresponds to a pseudo-ear $Q^{uv}$ .
\end{proof}

Pseudo-ears are useful in the following sense.
If, for some pseudo-ear $Q^{uv}$, its witness path $W^{uv}$ contains 'enough' credit, then we can add $Q^{uv}$ to $H$ to obtain the graph $H'$.
Afterwards, $H'$ is a graph with less bridges that also maintains all invariants.
We formalize this in the following two lemmas.

\begin{restatable}[]{lemma}{lembridgecoveringaddingpseudoears}
    \label{lem:bridge-covering:adding-pseudo-ears}
    Let $Q^{uv}$ be a pseudo-ear and let $W^{uv}$ be its witness path in $C$.
    Assume that $W^{uv}$ contains one of the following:
    \begin{itemize}
        \item One vertex that is a non-simple block vertex and another vertex that is either a non-simple block vertex or a lonely leaf vertex.
        \item Two bridges of $S$ and either one vertex that is a lonely leaf vertex or one non-simple block vertex.
        \item Two lonely leaf vertices and at least one bridge of $S$.
        \item Three bridges of $S$.
    \end{itemize}
    Then $H' = H \cup Q^{uv}$ satisfies the Invariants~\ref{invariant:credit}-\ref{invariant:block-size} and has strictly less bridges than $H$.
\end{restatable}

\begin{proof}
    Observe that in $H'$, all vertices in $W^{uv}$ are contained in a newly formed 2EC block, together with all vertices of $Q^{uv}$. 
    Hence, the number of bridges in $H'$ is strictly less than the number of bridges in $H$, since $u \neq v$.
    Furthermore, this also implies that all lonely vertices in $H'$ have degree at most 2, implying \Cref{invariant:degree-lonely-vertex}.
    
    Next, we prove that $H'$ satisfies \Cref{invariant:block-size}.
    If $W^{uv}$ contains at least one non-simple block vertex, then \Cref{invariant:block-size} follows since $H$ satisfied \Cref{invariant:block-size}.
    Therefore, assume that $W^{uv}$ does not contain any non-simple block vertices.
    If $W^{uv}$ contains at least two bridges of $S$ and a lonely leaf vertex, then $W^{uv}$ must also contain at least two distinct paths $P, P' \in \mathcal{P}$.
    Similarly, \Cref{invariant:block-size} is satisfied if $W^{uv}$ contains at least three bridges.
    Hence, assume that $W^{uv}$ contains two lonely leaf vertices and one bridge of $S$.
    Then, $W^{uv}$ must also contain at least two distinct paths $P, P' \in \mathcal{P}$, implying \Cref{invariant:block-size}.

    Finally, we prove that $H'$ satisfies \Cref{invariant:credit}.
    It satisfies to show that $\credits(H \cup Q^{uv}) + |Q^{uv}| \leq \credits(H)$.
    First, note that each vertex in $G^C$ corresponds to a component in $H$, which has a credit of at least 1. 
    Therefore, the number of links in $Q^{uv}$ is exactly the number of distinct components visited by $Q^{uv}$.
    Hence, by adding $Q^{uv}$ to $H$ the number of component credits is reduced by the number of links in $Q^{uv}$ minus $1$, i.e., by $|Q^{uv}| -1$.
    Further note that the newly created block $B$ containing the vertices of $W^{uv}$ also receives a credit of one. 
    Therefore, we need to argue that 2 additional credits are reduced by adding $Q^{uv}$ to $H$.
    If $W^{uv}$ contains two block vertices, two lonely leaf vertices or one of each, then by adding $Q^{uv}$ to $H$ the number of credits is additionally reduced by 2 and we are done.
    If $W^{uv}$ contains either one block vertex or one lonely leaf vertex, then $W^{uv}$ must contain 2 bridges.
    Then these credits sum up to at least $\frac{5}{2}$, so that by adding $Q^{uv}$ to $H$ the number of credits is additionally reduced by 2.
    Finally, if $W^{uv}$ contains at least three bridges, then this credit sums up to at least $\frac 9 4$, which is at least 2.
    Therefore, $H'$  satisfies \Cref{invariant:credit}.
\end{proof}

\begin{restatable}[]{lemma}{lembridgecoveringaddingtwopseudoears}
    \label{lem:bridge-covering:adding-pseudo-ears2}
    Let $Q^{uv}$ and $Q^{xy}$ be pseudo-ears that are vertex-disjoint w.r.t.\ $H^C$ (except for potentially $v$ and $y$) and let $W^{uv}$ and $W^{xy}$ be their witness paths in $C$, respectively, and assume that $W^{uv} \cap W^{xy} \neq \emptyset$.
    If $W^{uv} \cup W^{xy}$ contains two distinct non-simple blocks, two distinct lonely leaf vertices, or one non-simple block and one lonely-leaf vertex, and additional contains 2 bridges of $S$, then $H' = H \cup Q^{uv} \cup Q^{xy}$ satisfies the Invariants~\ref{invariant:credit}-\ref{invariant:block-size} and has strictly less bridges than $H$.
    Furthermore, if $W^{uv} \cup W^{xy}$ contains 1 bridge of $S$ and two distinct lonely leaf vertices of which at least one receives a credit of at least $\frac 54$, then $H' = H \cup Q^{uv} \cup Q^{xy}$ satisfies the Invariants~\ref{invariant:credit}-\ref{invariant:block-size} and has strictly less bridges than $H$.
\end{restatable}

\begin{proof}
    The proof is similar to the previous proof.
    First, observe that after adding $Q^{uv} \cup Q^{xy}$, the bridges in $W^{uv} \cup W^{xy}$ will be in one single block in $H \cup Q^{uv} \cup Q^{xy}$, since $Q^{uv} \cap Q^{xy} \neq \emptyset$. 
    This also already implies that the number of bridges in $H'$ has decreased compared to $H$.
    Furthermore, similar to before it is easy to see that $H'$ satisfies \Cref{invariant:block-size} and \Cref{invariant:degree-lonely-vertex}.

    Finally, we prove that $H'$ satisfies \Cref{invariant:credit}.
    Assume first that $W^{uv} \cup W^{xy}$ contains 2 bridges.
    It satisfies to show that $\credits(H \cup Q^{uv} \cup Q^{xy}) + |Q^{uv}| + |Q^{xy}| \leq \credits(H)$.
    Recall that $Q^{xy}$ and $Q^{xy}$ are vertex disjoint w.r.t.\ $G^C$.
    First, note that each vertex in $G^C$ corresponds to a component in $H$, which has a credit of at least 1. 
    Therefore, the number of links in $Q^{uv} \cup Q^{xy}$ is exactly the number of distinct components visited by $Q^{uv}$ minus 1.
    Hence, by adding $Q^{uv} \cup Q^{xy}$ to $H$ the number of component credits is reduced by the number of links in $Q^{uv} \cup Q^{xy}$ minus $2$, i.e., by $|Q^{uv}| + |Q^{xy}| -2$.
    Further note that the newly created block $B$ containing the vertices of $W^{uv} \cup W^{xy}$ also receives a credit of one. 
    Therefore, we need to argue that 3 additional credits are reduced by adding $Q^{uv} \cup Q^{xy}$ to $H$.
    Since $W^{uv} \cup W^{xy}$ contains two bridges in $S$ and two block vertices, two lonely leaf vertices or one of each, then by adding $Q^{uv} \cup Q^{xy}$ to $H$ the number of credits is additionally reduced by 3. 
    Therefore, $H'$  satisfies \Cref{invariant:credit}.

    Finally, if $W^{uv} \cup W^{xy}$ contains only 1 bridge, then one of the lonely leaf vertices receives a credit of at least $\frac 54$. Similar to the previous case it can easily be checked that $H'$ satisfies \Cref{invariant:credit}.
\end{proof}

Using pseudo-ears, we now show that in polynomial time we can find link sets $R \subseteq L \setminus S$ and $F \subseteq S$, such that $(H \cup R) \setminus F$ has fewer bridges than $H$ and $(H \cup R) \setminus F$ also satisfies all the invariants.
This then proves \Cref{lem:bridge-covering:main-iterative}.
We split this proof into several lemmas. 
First, we prove that whenever $H$ contains a simple block, then we can find in polynomial time an edge set $R$, such that $H \cup R$ has fewer bridges than $H$ and $H \cup R$ also satisfies all the invariants.

\begin{restatable}[]{lemma}{lembridgecoveringsimple}
    \label{lem:bridge-covering:simple}
    Assume that $H^C$ contains a simple block vertex $b$.
    Then there are two sets of links $R \subseteq L \setminus S$ and $F \subseteq S$ such that $(H \cup R) \setminus F$ has less bridges than $H$ and 
    $(H \cup R) \setminus F$ satisfies the Invariants~\ref{invariant:credit}-\ref{invariant:block-size}.
    Moreover, $R$ and $F$ can be found in polynomial time.
\end{restatable}

\begin{proof}
    We make a case distinction on how the simple block vertex is connected to the remaining part of the complex component.
    
    \textbf{Case 1: There is a simple block $b$ that causes a lonely leaf vertex $u$ to receive an additional credit of $\frac{1}{4}$ according to rule A(ii).}
    Let $b$ be the first non-simple block reachable from $u$ in $H[C]$ and let $P \in \mcP$ be the path belonging to $u$.
    Since $G$ is 2EC, there must be some pseudo-ear $Q^{uv}$ starting at $u$ and ending at some vertex $v \in V(C)$ distinct from $u$.
    Furthermore, since $G$ is structured, there must be some pseudo-ear $Q^{uv}$ starting at $u$ and ending at some vertex $v \in V(C)$ distinct from $u$ and $v$ belongs a path $P' \in \mcp$ with $P' \neq P$.
    This is true since otherwise $G$ is not structured; in particular $G$ must contain a path separator.
    
    \textbf{Case 1.1: The unique path from $u$ to $b$ in $H^C$ contains only one link.}
    Let $Q^{uv}$ be a pseudo-ear starting at $u$ and ending at some vertex $v \in V(C)$ distinct from $u$ and $v$ belongs a path $P' \in \mcp$ with $P' \neq P$.
    We now prove that setting $H' = H \cup Q^{uv}$ satisfies all invariants and has less bridges than $H$.
    Clearly, $H'$ has less bridges than $H$. 
    Furthermore, clearly the created block subsumes $b$ and $P$, and hence contains at least two paths, proving \Cref{invariant:block-size}.
    Furthermore, this implies \Cref{invariant:degree-lonely-vertex}.
    To see \Cref{invariant:credit}, we need to show $\credits(H \cup Q^{uv}) + |Q^{uv}| \leq \credits(H)$.
    Observe that the number of components in $H'$ is equal to the number of components in $H$ minus $|Q^{uv}|-1$. 
    Hence, by adding $Q^{uv}$ to $H$ the number of component credits is reduced by the number of links in $Q^{uv}$ minus $1$, i.e., by $|Q^{uv}| -1$.
    Further note that the newly created block $B$ containing the vertices of $W^{uv}$ also receives a credit of one. 
    Therefore, we need to argue that 2 additional credits are reduced by adding $Q^{uv}$ to $H$.
    This credit is obtained by observing that $u$ has a lonely leaf credit of $\frac{5}{4}$ (which is not needed in $H'$) and $H'$ contains at least one less bridge in $S$ compared to $H$. This sums up to a total credit of at least $2$.

    \textbf{Case 1.2: The unique path from $u$ to $b$ in $H^C$ contains at least two links.}
    If there is a pseudo-ear $Q^{uv}$ starting at $u$ and ending at $v$ such that $W^{uv}$ contains at least 2 bridges, then according to \Cref{lem:bridge-covering:adding-pseudo-ears} we are done.
    Hence, assume that this is not the case.

    \textbf{Case 1.2.1: The vertex $u$ is expensive.}
    We now have that any pseudo-ear $Q^{uv}$ starting at $u$ and ending at $v$ satisfies that $W^{uv}$ contains only the unique bridge incident to $P$ in $H$.
    We add $Q^{uv}$ to $H$ to obtain $H'$. 
    Afterwards, $u$ is in a block $b_u$ in $H'$ which might not satisfy \Cref{invariant:block-size}.
    Hence, from $b_u$ in $H'$ we again argue that there must be a pseudo-ear $Q^{b_u x}$ covering at least one bridge from $S$ (since $G$ is structured).
    Let $H'' = H \cup Q^{uv} \cup Q^{b_u x}$.
    Now it can easily be checked that $H''$ satisfies \Cref{invariant:degree-lonely-vertex} and \Cref{invariant:block-size}.
    To see that \Cref{invariant:credit} is satisfied, observe that the number of components in $H''$ is equal to the number of components in $H$ minus $|Q^{uv}| + |Q^{b_u x}| -2$. Additionally, the newly created block receives a credit of 1.
    Hence, we need to 'find' a credit of 3. To do so, observe that $u$ initially had a credit of $\frac{3}{2}$ (as it satisfied A(i) and A(ii)). 
    Furthermore, $H''$ contains at least 2 fewer bridges of $L$ compared to $H$, which sums up together with the lonely leaf credit of $u$ to a total credit of at least $3$.
    This proves \Cref{invariant:credit}.

    \textbf{Case 1.2.2: The vertex $u$ is not expensive.}
    Let $v$ be the other endvertex of $P$ distinct from $u$ and let $e = v v'$ be the unique link incident to $v$ from $H \cap L$. Let $P' \in \mcP$ be the path to which $v'$ is incident and let $u'$ be the other distinct endvertex of $P'$.
    Since $u$ is non-expensive, $v v'$ is non-expensive. Hence, by the definition, $u$ is either incident to $u'$ or incident to some other path $P'' \in \mcP$ with $P \neq P'' \neq P'$.
    In the former case, if $u u' \in L$, it can easily be observed that adding $u u'$ to $H$ results in a solution with fewer bridges and also satisfying all invariants.
    Hence, assume that $u u' \notin L$.
    Thus, $u$ is incident to some path $P'' \in \mcP$ with $P \neq P'' \neq P'$.
    Recall that there is no pseudo-ear $Q^{ux}$ starting at $u$ and ending at $x$ such that $W^{ux}$ contains at least 2 bridges, as then we are done.
    Additionally, recall that we have applied \Cref{lem:bridge-covering:links-between-leafs} exhaustively and hence there is no link between $u$ and any lonely leaf vertex of a complex component distinct from $C$.
    We now claim that that this contradicts the fact that $G$ is structured.
    Any pseudo-ear starting in $u$ and ending at some vertex $x \in V(C)$ satisfies that $x \in V(P) \cup V(P')$. 
    Let $x$ be the neighbor of $u$ on $P$ ($v = x$ might be possible).
    Recall that $P^{vx}$ is the subpath of $P$ from $v$ to $x$.
    We claim that $Z = V(P') \cup V(P^{vx})$ is a $P^2$-separator.
    Clearly, $Z$ is a separator as it separates $u$ and $P''$ from $b$.
    Hence, let $(V_1, V_2)$ be any bipartition of $V \setminus Z$ such that there are no edges between $V_1$ and $V_2$ and such that $u \in V_1$.
    Observe that there is an edge $e \in \delta(Z) \cap E(\mcP)$.
    Hence, if we can show that $\opt((G \setminus V_2)|Z) \geq 4$ and $\opt((G \setminus V_1)|Z) \geq 3$ or  $\opt((G \setminus V_2)|Z) \geq 3$ and $\opt((G \setminus V_1)|Z) \geq 4$, we have shown that $Z$ is a $P^2$-separator.
    Clearly, either $\opt((G \setminus V_2)|Z) \geq 4$ or $\opt((G \setminus V_1)|Z) \geq 4$, as $G$ is structured.
    It is easy to see that $\opt((G \setminus V_1)|Z) \geq 3$, since there must be at least 4 vertices of $V^*(\mcp)$ contained in $V_2$.
    Hence, if $\opt((G \setminus V_2)|Z) \geq 3$, we are done.
    Therefore, assume $\opt((G \setminus V_1)|Z) \geq 4$. We show $\opt((G \setminus V_2)|Z) \geq 3$.
    This is clearly true if there are at least 4 vertices of $V^*(\mcp)$ contained in $V_1$.
    Hence, assume that there are only 3 such vertices (there must be at least 3 such vertices since $V_1$ contains $u$ and another path $P''$).
    Hence, $P''$ is a trivial complex component.
    But since $u$ is not incident to any endvertex of $P''$, we have $\opt((G \setminus V_1)|Z) \geq 3$, implying that $Z$ is a $P^2$-separator, a contradiction to $G$ being structured.
    This finishes the first case.

    \textbf{Case 2:  No lonely leaf vertex receives an additional credit of $\frac{1}{4}$ according to rule A(ii).}
    Hence, in $H^C$, from $b$ we can not reach a lonely leaf vertex without going through a non-simple block.
    Therefore, since each simple block in $C$ has degree 2 in $H^C$, there are exactly 2 non-simple blocks that are reachable in $H^C$ from $b$ without going through another non-simple block.
    Let these two non-simple blocks be $b_1$ and $b_2$.
    Let $Q_1$ be the path in $H^C$ from $b_1$ to $b$ and let $Q_2$ be the path in $H^C$ from $b_2$ to $b$.
    Note that on $Q_1$ and $Q_2$ there might be simple blocks distinct from $b$.
    However, due to \Cref{invariant:degree-lonely-vertex}, these simple-blocks are not adjacent in $H^C$.

    \textbf{Case 2.1: On $Q_1$ there is another simple block distinct from $b$.}
    Let $b'$ be the first simple block that is reachable from $b$ in $H^C$ when traversing $Q_1$ from $b$ to $b_1$. Let this subpath from $b$ to $b'$ be $Q_1'$.
    Note that by the construction of $H$, $b$ and $b'$ are not adjacent.
    Therefore, there must exist two distinct paths $P, P' \in \mcP$ on $Q_1'$ such that two endvertices of $P$ and $P'$ are connected by a link $e$ in $H$.
    Let $f, f'$ be the unique links incident to $P$ and $P'$ distinct from $e$.
    Observe that $G[V \setminus (V(P) \cup V(P')]$ must be connected, as otherwise $P$ and $P'$ together with $e$ form a $P^2$-separator, a contradiction to the fact that $G$ is structured.
    Now it is easy to see that there must be a pseudo-ear such that its witness path contains at least 3 links, namely $e, f$, and $f'$. Hence, by \Cref{lem:bridge-covering:adding-pseudo-ears}, we can add this pseudo-ear, satisfy all invariants and obtain a solution with fewer bridges.
    Hence, we can assume that on $Q_1$ there is no other simple block other than $b$.
    Similarly, this holds for $Q_2$.

    \textbf{Case 2.2: On $Q_1$ there is no other simple block distinct from $b$.}
    First, assume that $b_1$ is incident to a link $e$ on $Q_1$. Let $P \in \mcP$ be the path to which $e$ is incident to such that $P$ is not part of the block $b_1$.
    Furthermore, $P$ must be incident to one more link, say $e'$.
    Now observe that since $G$ is structured, $P$ can not be a separator and hence there must be a pseudo-ear such that its witness path contains $b_1$, $e$ and $e'$. Hence, by \Cref{lem:bridge-covering:adding-pseudo-ears}, we can add this pseudo-ear, satisfy all invariants and obtain a solution with fewer bridges.
    Similarly this holds for $Q_2$.
    
    We now assume that $b_1$ is not incident to a link on $Q_1$. Let $P_1 \in \mcP$ be the path to which $b_1$ is incident to on $Q_1$. Let $v_1, w_1$ be the endvertices of $P_1$. Without loss of generality, we have $u_1 \in V(b)$ and $v_1 \notin V(b)$, where $V(b)$ denotes the vertices of $G$ that belong to the block $b$. Let $P_1'$ be the subpath of $P_1$ which is not contained in $b_1$.
    Assume that the number of links on $Q_1$ is at least 2. 
    Hence, $v_1$ must be incident to a link $e$ that is incident to an endvertex $v_2$ of some path $P_2 \in \mcP$ with $P_2 \neq P_1$ and $P_2$ does not belong to a simple block, since in this case $Q_1$ contains at least 2 links and only one simple block (namely $b$).
    Furthermore, the other endvertex of $P_2$, say $u_2$, is also incident to a link $e'$ on $Q_1$ with $e \neq e'$.
    Now it is easy to see that there must be a pseudo-ear satisfying \Cref{lem:bridge-covering:adding-pseudo-ears}: Observe that $G[V \setminus (V(P_1') \cup V(P_2)]$ must be connected, as otherwise $P_1'$ and $P_2$ together with $e$ form a $P^2$-separator, a contradiction to the fact that $G$ is structured.
    Now it is easy to see that there must be a pseudo-ear such that its witness path contains at least 2 bridges, namely $e$ and $e'$ and contains at least one block, namely $b_1$.
    Similarly the same arguments holds if $Q_2$ contains at least two links.
    
    Finally, in the remaining case, we have that the number of links on $Q_1$ and $Q_2$ is exactly 1 and $Q$ contains only one simple block, which is $b$. Let $e_1$ be the link on $Q_1$ and $e_2$ be the link on $Q_2$.
    Since $G$ is structured, there must be a pseudo-ear $R_1$ such that its witness path contains $b_1$ and $e_1$ and a pseudo-ear $R_2$ such that its witness path contains $b_2$ and $e_2$.
    Now observe that $R_1$ and $R_2$ satisfy \Cref{lem:bridge-covering:adding-pseudo-ears2} and hence we can add these pseudo-ears, satisfy all invariants and obtain a solution with fewer bridges.
    This finishes the proof.
\end{proof}

Hence, after applying the above lemma exhaustively, there are no more simple blocks in $H$.
In particular, this implies that no lonely leaf vertex in $H$ receives credit according to rule A(ii).
As we will never create simple blocks, throughout the remaining parts of this section we assume that $H$ does not contain simple blocks.

Next, we deal with the case that a lonely leaf vertex is non-expensive.
\begin{restatable}[]{lemma}{lembridgecoveringnonexpensive}
    \label{lem:bridge-covering:non-expensive}
    Let $C$ be a non-trivial component and let $v$ be a lonely leaf vertex of $H^C$ that does not receive credit according to A(i), i.e., $v$ is a non-expensive leaf vertex.
    Then there are two sets of links $R \subseteq L \setminus S$ and $F \subseteq S$ such that $(H \cup R) \setminus F$ has less bridges than $H$ and 
    $(H \cup R) \setminus F$ satisfies the Invariants~\ref{invariant:credit}-\ref{invariant:block-size}.
    Moreover, $R$ and $F$ can be found in polynomial time.
\end{restatable}

\begin{proof}
We assume throughout the proof that $C$ does not contain simple blocks. Let $P \in \mcP$ be the path of which $v$ is an endpoint of and let $u$ be the other endpoint of $P$.
We divide the proof into two cases depending on whether $C$ contains a non-simple block $b$ or not. Each case itself is then divided into three cases, depending on the number of bridges that are links of $L$.

\textbf{Case 1: $C$ contains a non-simple block $b$.}
We show that there must exist a pseudo-ear starting in $v$ such that its witness path either contains 2 bridges or one block vertex.
Let $b$ be the closest non-simple block vertex to $v$ and let $Q$ be the path from $v$ to $b$ on $C$. Recall that by \Cref{invariant:degree-lonely-vertex}, each vertex on $Q$ except $v$ and $b$ has degree exactly 2 in $H$. Let $X = Q \cap L$.

\textbf{Case 1.1: $|X|= 0$.}
Then, since $G$ is structured, it is easy to see that there must be a pseudo-ear $Q^{vy}$ starting in $v$ and ending in $y$ such that its witness path $W^{vy}$ contains $b$. By \cref{lem:bridge-covering:adding-pseudo-ears}, adding $Q^{vy}$ to $H$ results in a solution with fewer bridges and satisfying the Invariants~\ref{invariant:credit}-\ref{invariant:block-size}.

\textbf{Case 1.2: $|X|= 1$.}
Let $P' \in \mcP$ be the other path distinct from $P$ that is on $Q$. Let $u'$ be the endpoint of $P'$ such that $uu' \in H$ and let $v'$ be the other endpoint of $P'$. Note that $v'$ must be part of the vertices corresponding to the block $b$.
Hence, if $vv' \in L$, clearly adding $vv'$ to $H$ results in the desired new solution.
Therefore, since $v$ is non-expensive, there must be another path $P'' \in \mcP$ distinct from $P$ and $P'$ such that $u$ is incident to some vertex of $P''$. 
Now it has to be the case that there is a pseudo-ear $Q^{vy}$ starting in $v$ and ending in $y$ such that its witness path $W^{vy}$ contains $b$. By \Cref{lem:bridge-covering:adding-pseudo-ears}, adding $Q^{vy}$ to $H$ results in a solution with fewer bridges and satisfying the Invariants~\ref{invariant:credit}-\ref{invariant:block-size}.
If such a pseudo-ear does not exist, it is not hard to see that $G$ is not structured (a contradiction): Let $x$ be the first vertex of $P'$ that is part of $b$ when traversing $Q$ from $v$ to $b$. Let $Z= Q^{vx} - \{v, x\}$ be the subpath of $Q$ from $v$ to $x$ without the vertices $v$ and $x$. Since the above witness path does not exist, $Z$ is a separator separating $v$ and $P''$ from $b$. Let $(V_1, V_2)$ be a partition of $V \setminus Z$ such that there are no edges between $V_1$ and $V_2$ and $v \in V_1$. Since \Cref{lem:bridge-covering:links-between-leafs} has been applied exhaustively, we have that $\opt((G \setminus V_2)|Z) \geq 3$. Also, by \Cref{invariant:block-size}, we have $\opt((G \setminus V_1)|Z) \geq 3$. However, since $G$ is structured, either $\opt((G \setminus V_1)|Z) \geq 4$ or $\opt((G \setminus V_2)|Z) \geq 4$ must hold. But then $Z$ is a $P^2$-separator, a contradiction to the fact that $G$ is structured.

\textbf{Case 1.3: $|X| \geq 2$.}
The case can be essentially treated the same as the previous case with $|X|=1$, except that we argue that there must exist a pseudo-ear $Q^{vy}$ starting in $v$ and ending in $y$ such that its witness path $W^{vy}$ contains at least two bridges. By \Cref{lem:bridge-covering:adding-pseudo-ears}, adding $Q^{vy}$ to $H$ results in a solution with fewer bridges and satisfying the Invariants~\ref{invariant:credit}-\ref{invariant:block-size}. If such a pseudo-ear does not exist, similar to before we can argue that there must be a separator $Z$ that is a $P^2$-separator, a contradiction to the fact that $G$ is structured.
This finishes the first case.

\textbf{Case 2: $C$ does not contain any blocks.}
By \Cref{invariant:degree-lonely-vertex}, $C$ is a path.
Since $C$ is non-trivial, $C$ contains at least one bridge in $L$. Let $X = C \cap L$.
Let $u$ be the other endvertex of $P$ distinct from $v$.
Let $P'$ be the other path to which $u$ is incident and let $u' \in V(P')$ be the endvertex of $P'$ to which $u$ is incident. Let $v'$ be the other endvertex of $P'$.

\textbf{Case 2.1: $|X| = 1$.}
If there is a pseudo-ear starting in $v$ and ending in $v'$, we are done according to \Cref{lem:bridge-covering:adding-pseudo-ears}. Hence, let us assume this is not the case.
Since $v$ is non-expensive, either $vv' \in L$ or there is some path $\hat{P}$ distinct from $P$ and $P'$ to which $v$ is incident. The former case can not occur by our previous assumption. 
Thus, assume $vv' \notin L$. Hence, there is some path $\hat{P}$ distinct from $P$ and $P'$ to which $v$ is incident.
Assume that $v'$ is expensive. Since $G$ is structured, there must exist two pseudo-ears $Q^{vy}$ and $Q^{v'z}$ starting in $v$ and $v'$, respectively, such that their witness paths overlap. Now that the witness paths also cover the bridge $uu'$ and hence the condition of \Cref{lem:bridge-covering:adding-pseudo-ears2} is satisfied, implying that adding $Q^{vy}$ and $Q^{v'z}$ to $H$ can be done as desired.
Hence, we assume that also $v'$ is non-expensive. Since $vv' \notin L$, there must be another path $\bar{P} \in \mcP$ distinct from $P, P'$, and $\hat{P}$. $\bar{P}$ must be distinct from $\hat{P}$, as otherwise there is a pseudo-ear connecting $v$ and $v'$, a contradiction.
But now it can be easily observed that $Z = V(P) \cup V(P') \setminus ( v + v')$ is a separator.
Similarly to previous cases, in can easily be observed that $Z$ is a $P^2$-separator, a contradiction to $G$ being structured. Hence, there must exist a pseudo-ear connecting $v$ and $v'$.

\textbf{Case 2.2: $|X| = 2$.}
Let $P'' \in \mcP$ be the other path on $C$ distinct from $P$ and $P'$ with endvertices $u'', v''$ such that $v' v'' \in H$.
Again, we assume that there are no pseudo-ears starting in $v$ (or $u''$) satisfying the conditions of either \Cref{lem:bridge-covering:adding-pseudo-ears} or \Cref{lem:bridge-covering:adding-pseudo-ears2}, as then we are done.
This already implies that $vv' \notin L$, as there is a pseudo-ear starting in $u''$ such that its witness path contains $v'$, contradicting our previous assumption.
Hence, since $v$ is non-expensive, there is some path $\hat{P}$ distinct from $P, P'$ and $P''$ to which $v$ is incident.
Let $z \in V(C)$ be the vertex such that the witness path of the pseudo-ear $Q^{u''z}$ (if there is a pseudo-ear to $z$) is as long as possible.
If $v'u'' \in L$, we modify $H$ to obtain $H'$ by removing $v' v''$ from $H$ and adding $v' u''$ if the following is satisfied: There is a $y \in V(C)$ such that there is a pseudo-ear $Q^{v''y}$ such that $y$ is closer to $v$ than $z$. If this is not the case, we simply do nothing, but still set $H' = H$. We now work with $H'$ instead of $H$ and observe that this modification does not affect the number of components, bridges or any invariant.
(For simplicity, we 'ignore' the case that modifying $H$ to $H'$ might increase the credit by $\frac 14$ due to making $v''$ an expensive lonely leaf vertex, as we do not use this credit anyways and after in this step all vertices of $C$ are in a single 2-edge connected block.)
For convenience, let us assume that the previous if-condition was not true and we did not swap $v'v''$ with $v' u''$. The other case is analogous.
Let $x$ be the vertex on $P'$ that is furthest away from $v$ on $C$ but still reachable by some pseudo-ear starting in $v$. Additionally, by our assumptions, we know that $z \in V(P')$. Let $P'_z$ be the subpath of $P'$ from $u'$ to $z$.
We now claim that $Z = V(P) \cup V(P'_z) \setminus (v + z)$ is a $P^2$-separator, which is a contradiction to $G$ being structured.
Note that $x \in Z$.
Clearly, $Z$ is a separator, since $v$ is separated from $u''$.
Let $(V_1, V_2)$ be a partition of $V \setminus Z$ such that there are no edges between $V_1$ and $V_2$ and $v \in V_1$. Since \Cref{lem:bridge-covering:links-between-leafs} has been applied exhaustively, we have that $\opt((G \setminus V_2)|Z) \geq 3$. Also, since $u''$ and $v''$ are both not incident to $Z$, we have that, we have $\opt((G \setminus V_1)|Z) \geq 3$. However, since $G$ is structured, either $\opt((G \setminus V_1)|Z) \geq 4$ or $\opt((G \setminus V_2)|Z) \geq 4$ must hold.
This implies that $Z$ is a $P^2$-separator, a contradiction.

\textbf{Case 2.3: $|X| \geq 3$.}
Let $P', P'', P''' \in \mcP$ be the closest three paths on $C$ to $v$ with endpoints $u', v', u'', v'', u''', v'''$, respectively, such that $uu', v'v'', u''u''' \in L \cap H$.
Again, we assume that there are no pseudo-ears starting in $v$ (or $v'''$) satisfying the conditions of either \Cref{lem:bridge-covering:adding-pseudo-ears} or \Cref{lem:bridge-covering:adding-pseudo-ears2}, as then we are done.
If $vv' \in L$, we consider $H' = H + vv'$. In $H'$, there must be a pseudo-ear $Q^{xy}$ starting in some vertex $x \in V(P) \cup V(P')$ and ending in some vertex $y \in V(P''')$, as otherwise $P''$ would be a separator.
But then adding $Q^{xy}$ to $H'$ results in a solution with fewer bridges and it can be easily checked that all invariants are satisfied, since in total at least 3 bridges of $S$ are covered.
Hence, since $v$ is non-expensive, there is some path $\hat{P}$ distinct from $P, P', P''$ and $P'''$ to which $v$ is incident. Similarly to previous cases, it can be observed that there must exist a pseudo-ear satisfying \Cref{lem:bridge-covering:adding-pseudo-ears}, as otherwise $Z = V(P) \cup V(P') - v$ is a $P^2$-separator, which is a contradiction. This finishes the proof.
\end{proof}

We apply the above lemma exhaustively and therefore assume that $H$ does not have non-trivial complex components that have a non-expensive lonely leaf vertex.
Next, we assume that there is a complex component containing a (non-simple) block vertex.

\begin{restatable}[]{lemma}{lembridgecoveringlargecomponent}
    \label{lem:bridge-covering:large-block}
    Let $u$ be a non-simple block vertex of $H^C$.
    Then there are two sets of links $R \subseteq L \setminus S$ and $F \subseteq S$ such that $(H \cup R) \setminus F$ has less bridges than $H$ and 
    $(H \cup R) \setminus F$ satisfies the Invariants~\ref{invariant:credit}-\ref{invariant:block-size}.
    Moreover, $R$ and $F$ can be found in polynomial time.
\end{restatable}

\begin{proof}
    Throughout this proof we assume that none of the previous lemmas applies and hence there are no simple block vertices in $H$ and all lonely leaf vertices of non-trivial complex components are expensive.
    Assume first that $u$ is a leaf of $C$.
    Since $G$ is 2EC, there must be some pseudo-ear $Q^{uv}$ starting at $u$ and ending at some vertex $v \in V(C)$ distinct from $u$.
    If there is a pseudo-ear $Q^{uv}$, such that $Q^{uv}$ satisfies one of the conditions in \Cref{lem:bridge-covering:adding-pseudo-ears}, then we are done.
    Hence, assume that for all $v \in V(C)$, for which some pseudo-ear $Q^{uv}$ exists, the pseudo-ear $Q^{uv}$ does not satisfy the conditions in \Cref{lem:bridge-covering:adding-pseudo-ears}. 
    Of all such $v \in V(C)$ let $w$ be the vertex such that its corresponding witness path $W^{uw}$ has the maximum number of edges.
    All vertices of $W^{uw}$ must have degree at most 2, as otherwise, by \Cref{invariant:degree-lonely-vertex}, one of the vertices would be a block vertex, a contradiction to our assumption.
    Hence, $w$ is unique, a lonely vertex and we have that $V(C)-V(W^{uw})$ contains at least one vertex (since it must contain a leaf of $C$).
    Furthermore, $W^{uw}$ contains no blocks, no lonely leaf vertex and at most one bridge of $S$.
    Let $V_u$ be the set of components that are reachable via a simple path from $u$ in $G^C \setminus (V(C) - u)$.
    Note that $Q \coloneqq W^{uw} - u$ is a separator in $G$, since it disconnects the vertices in $V_u$ from $V(C)-V(Q)$.
    Let $(V_1, V_2)$ be a partition of $V(G) \setminus V(Q)$ such that there are no edges between $V_1$ and $V_2$ in $V(G) \setminus V(Q)$ and such that $u \in V_1$ (here we mean the vertices in $V(G)$ corresponding to the contracted block vertex $u$) and $V(C)-V(Q) \subseteq V_2$.
    
    Note that $Q$ contains at most one bridge of $S$, as we assumed that no condition of \Cref{lem:bridge-covering:adding-pseudo-ears} is applicable for $Q^{uv}$.
    Further, note that $Q$ must contain at least 1 bridge in $S$. Otherwise, $Q$ must be a sub-path of some path path $P \in \mathcal{P}$, a contradiction to the fact that $G$ is structured (contradicting Property P4).
    
    Hence, $Q$ must contain exactly one bridge in $S$.
     Next, we show that one of the endpoints of $Q$ must have an edge $e \in E(\mathcal{P})$ such that $e \in \delta(Q)$.
    Assume the contrary.
    Note that there are exactly two bridges in $E(H) \cap \delta(Q)$: $e_1$, which is incident to $u$, and $e_2$, which is incident to $w$ (and going to $V_2$). 
    Hence, both these edges must be in $S$.
    But then, $Q^{uw}$ contains exactly two bridges of $S$, one contained contained in $Q$ and $e_1$. This is a contradiction to our previous assumption.
    
    Therefore, one of the endpoints of $Q$ must have an edge $e^* \in E(\mathcal{P})$ such that $e^* \in \delta(Q)$.
    Next, assume that $V_2$ contains a block vertex.
    We show that then $Q$ is a $P^2$-separator, contradicting that $G$ is structured.
    Observe that by \Cref{invariant:block-size}, we have that $\opt((G \setminus V_i))|Q) \geq 3$ for $i \in \{1, 2 \}$. 
Furthermore, since $G$ is structured, we have that $\opt(G) \geq 12$ and hence $\opt((G \setminus V_i))|Q) \geq 4$ for some $i \in \{1, 2 \}$.
Further, by the existence of $e^*$, we have that $Q$ is a $P^2$-separator, a contradiction to the fact that $G$ is structured.

Hence, $V_2$ does not contain a block vertex and therefore $H[V_2]$ is the union of paths, by \Cref{invariant:degree-lonely-vertex}.
If there are two paths $P_1, P_2 \in \mathcal{P}$ such that $V(P_1) \cup V(P_2) \subseteq V_2$, then also $\opt((G \setminus V_1))|Q) \geq 3$ and hence, by the same argument as above, $Q$ would be a $P^2$-separator, a contradiction.

If $H[V_2]$ is disconnected, then we claim that $\opt((G \setminus V_1))|Q) \geq 3$. 
We know that $\opt((G \setminus V_1))|Q) \geq 2$ since two vertices in $H[V_2]$ must be leaves in this case.
However, if $\opt((G \setminus V_1))|Q) = 2$, then there must be two vertices $x, y \in V(\mathcal{P})$ of distinct paths $P_x, P_y \in \mathcal{P}$ in $V_2$ such that $xy \in L$.
But this contradicts the fact that we have used \Cref{lem:bridge-covering:links-between-leafs} exhaustively before.
Now, as before, we have $\opt((G \setminus V_1))|Q) \geq 3$ and therefore $Q$ is a $P^2$-separator, a contradiction.

Hence, it follows that $H[V_2]$ is a path and contains either one or no link of $S$.
If $H[V_2]$ contains no link of $S$, then we claim that there is an edge $f \in L$ such that $H \cup Q^{uw} + f$ has fewer bridges than $H$ and satisfies the Invariants~\ref{invariant:credit}-\ref{invariant:block-size}.
Let $z$ be the leaf of $V_2$.
Since $H[V_2]$ contains no link of $S$, $H[V_2] \subseteq P$ for some $P \in \mathcal{P}$. 
Further, there must be a link $f \in L$ incident to $z$ to some other path $P' \in \mathcal{P}$. 
Let $H' = (V, E(\mathcal{P}) \cup S')$, where $S' =S \cup Q^{uw} + f$.
Since $Q$ is a separator, this edge must go to $Q$ such that all bridges of $C$ are in a 2EC component in $H'$ (recall that $C$ is a path in $H^C$).
It remains to show that the Invariants~\ref{invariant:credit}-\ref{invariant:block-size} are satisfied.
We first prove that \Cref{invariant:credit} is satisfied.
First, observe that $|S'| = |Q^{uw}| + 1$.
Next, note that all components incident to $Q^{uw}$ are in a single connected component in $H'$. 
Moreover, in $H'$ all bridges of $C$ are now in a single 2EC component, for which we need a component credit of 2, as either it is a single 2EC component (with credit 2) or a complex component, for which we need a component credit of 1 and a block credit of 1.
However, in $H'$, the leaf $z$ is not a leaf anymore, the unique bridge $W^{uw} \cap S$ is not a bridge in $H'$ and the block vertex $u$ is also contained in the newly created component.
Hence, we have $\credits(H') \leq \credits(H) - (|Q^{uw}| -1 -1) - \frac{5}{4} - \frac{3}{4} + 1 = \credits(H) -1$.
Hence, we obtain $\credits(H') + |S'| \leq \credits(H) + |S|$ and therefore $H'$ satisfies \Cref{invariant:credit} as $H$ satisfies \Cref{invariant:credit}.
Next, observe that each lonely vertex that is incident to $f$ or some link of $W^{uw}$, is in a 2EC component in $H'$.
Hence, $H'$ satisfies \Cref{invariant:degree-lonely-vertex}, since $H$ satisfies \Cref{invariant:degree-lonely-vertex}.
Since $u$ was a block vertex in $H$ and is now part of the block created by adding $f$ and $Q^{uw}$, \Cref{invariant:block-size} is also satisfied, since $H$ satisfies \Cref{invariant:block-size}.

Finally, we assume that $H[V_2]$ is a path and contains exactly one edge of $S$.
Let $a \in V_2$ be the leaf in $H[V_2]$ and let $b c $ be the unique bridge in $H[V_2]$ of $S$.
Let $P_1, P_2 \in \mathcal{P}$ be the distinct paths such that $a$ and $b$ are the endpoints of $P_1$ and $c$ is one of the endpoints of $P_2$. Observe that $V_2 \subseteq V(P_1) \cup V(P_2)$.

Recall that if $\opt((G \setminus V_1))|Q) \geq 3$, then $Q$ is a $P^2$-separator, a contradiction to the fact that $G$ is structured. 
Hence, $\opt((G \setminus V_1))|Q) = 2$.
Let $F$ be any such optimal solution of $\OPT(G \setminus V_1))|Q)$.
We claim that $H' = (H - bc) \cup Q^{uw} \cup F$ is 2EC and satisfies the Invariants~\ref{invariant:credit}-\ref{invariant:block-size}.
Let $H' = (V, E(\mathcal{P}) \cup S')$, i.e., $S' = (S - bc) \cup Q^{uw} \cup F$.
Note that $|S'| = |S| + |Q^{uw}| + 1$, since we removed $bc$ and $|F| = 2$ (if $bc \in F$, the statement follows as well since $bc \in E(H) \cap S$).
Further, note that all bridges of $C$ are now in a 2EC component of $H'$. 
Therefore, since $\credits(H') \leq \credits(H) - (|Q^{uw}| -1) - \frac{5}{4} - 2 \cdot \frac{3}{4} = \credits(H) - |Q^{uw}| - \frac{7}{4}$ we have that $\credits (H') + |S'| \leq \credits(H) + |S|$ and hence \Cref{invariant:credit} is satisfied. 
Furthermore, it is easy to check that also \Cref{invariant:degree-lonely-vertex} and \Cref{invariant:block-size} are satisfied.
This finishes the case that $u$ is a leaf of $C$.

Finally, we assume that $u$ is not a leaf of $C$.
This case is very similar to the previous one.
Let $u v$ be some bridge incident to $u$ in $H$ and let $C_v$ be the connected component of $H^C[V(C)] - u$ that contains $v$.
Further, let $C_u = H^C[V(C)] \setminus C_v$.

In order to compute a pseudo-ear that covers the bridge $uv$, we now contract $C_u$ to a single vertex $u'$ and then compute a pseudo-ear starting at $u'$ and ending at some vertex in $w \in V(C_v)$ in this contracted graph.
Now observe that $u'$ is a leaf in this new graph and hence we can follow the exact same steps as in the case that $u$ was a leaf in the first place.
\end{proof}

Again, we apply the above lemma exhaustively and can assume that $H$ contains no non-trivial complex component that contains a block and that each leaf is an expensive leaf.
We next consider the case that $C$ is a non-trivial component.

\begin{restatable}[]{lemma}{lembridgecoveringnolargecomponent}
    \label{lem:bridge-covering:no-large-component}
    Let $C$ be a non-trivial complex component such that each vertex of $C$ is lonely.
    Then there are two sets of links $R \subseteq L \setminus S$ and $F \subseteq S$ such that $(H \cup R) \setminus F$ has less bridges than $H$ and 
    $(H \cup R) \setminus F$ satisfies the Invariants~\ref{invariant:credit}-\ref{invariant:block-size}.
    Moreover, $R$ and $F$ can be found in polynomial time.
\end{restatable}

\begin{proof}
Throughout this proof we assume that all previous lemmas have been applied exhaustively. In particular, there are no complex components containing blocks and all lonely leaf vertices of non-trivial complex components are expensive. 
    Since each vertex of $C$ is lonely, \Cref{invariant:degree-lonely-vertex} implies that $C$ is a path.
    Let $C = P_1, e_1, P_2, e_2, ..., $ $e_{k-1}, P_k$, where $P_i \in \mathcal{P}$ for $i \in [k]$, $P_i \neq P_j$ for $i, j \in [k],i \neq j$ and let $u_i, v_i$ be the endpoints of $P_i$ in the order they appear when traversing $C$ from $P_1$ to $P_k$.
    Observe that $e_i = v_i u_{i+1}$ is the link connecting $P_i$ and $P_{i+1}$ in $C$.
    For some $x \in V(C)$ let $C^+(x)$ be the sub-path of $C$ starting at $x$ and ending at $v_k$ and let $C^-(x)$ be the sub-path of $C$ starting at $u_1$ and ending at $x$.

    Let $w_1$ be the furthest vertex (w.r.t.\ number of edges) from $u_1$ on $C$ such that there is a pseudo-ear $Q^{u_1 w_1}$ starting in $u_1$ and ending at some vertex $w_1$.
    If $w_1 \in \bigcup_{3 \leq i \leq k} V(P_i)$, then $Q^{u_1 w_1}$ satisfies the Invariants~\ref{invariant:credit}-\ref{invariant:block-size} by \Cref{lem:bridge-covering:adding-pseudo-ears}.

    Hence, we assume $w_1 \in V(P_2)$, since there must be a link $u_1 a$ for some $a \notin V(P_1)$ and $G$ does not contain path separators (as $G$ is structured).
    Analogously, we have the same for $v_k$ as for $u_1$: Let $w_k$ be the furthest vertex (w.r.t. \ number of edges) from $v_k$ on $C$ such that there is a pseudo-ear $Q^{v_k w_k}$ starting in $v_k$ and ending at some vertex $w_k$.
    Similarly, $w_k \in  \bigcup_{1 \leq i \leq k-2} V(P_i)$ would lead to the desired set $R$, and hence we assume that $w_k \in V(P_{k-1})$ (since $v_k$ must have a link to some other path, as $G$ is structured).
    We make a case distinction on $k$.

    \textbf{Case 1: $k \leq 3$.}
    We argue that if $k \leq 3$, then we can find the desired set $R$ as in the lemma statement.
    Note that $k \geq 2$ since $C$ is nontrivial.
    
    \textbf{Case 1.1: $k =2$.} 
    If there is some pseudo-ear $Q^{u_1 w_1}$ such that $w_1 = v_2$ (the endpoint of $C$), then $R= Q^{u_1 w_1}$ is the desired set of edges.
    Otherwise, we set $R = Q^{u_1 w_1} \cup Q^{v_2 w_2}$ (note that in this case $Q^{u_1 w_1}$ and $Q^{v_2 w_2}$ must be disjoint).
    In any of the two cases, we claim that $H' = H \cup R$ satisfies the Invariants~\ref{invariant:credit}-\ref{invariant:block-size}.
    To see this, first observe that $C$ is now in a 2EC block (or component) in $H'$ and hence $H'$ satisfies the Invariants~\ref{invariant:degree-lonely-vertex} and~\ref{invariant:block-size}.
    Let $S' = S \cup R $.
    Observe that $|S'| = |S| + |R|$ and $\credits(H') \leq \credits(H) - (|R| + 2) - 2 \cdot \frac{5}{4} - \frac{3}{4} < \credits(H) - |R|$, since all components incident to $R$ are merged to a single component and the credits of the two leaves of $C$ and the bridge of $C \cap S$ is not needed anymore. 
    Hence, $\credits(H') + |S'| \leq \credits(H) + |S|$.
    This finishes case $k=2$.

    \textbf{Case 1.2: $k = 3$.}
    First, assume that (i) there is one pseudo-ear starting at $u_1$ and ending at $v_3$ or (ii) there are two disjoint pseudo-ears $Q^{u_1 w_1}$ and $Q^{v_3 w_3}$ such that $C^+(w_1) \cap C^-(w_3) \neq \emptyset$.
    In case (i), we set $R = Q^{u_1 v_3}$ and in case (ii) we set $R = Q^{u_1 w_1} \cup Q^{v_3 w_3}$.
    Similar to before, it is easy to see that in both cases $H \cup R$ satisfies Invariants~\ref{invariant:credit}-\ref{invariant:block-size}.
    Hence, assume that the pseudo-ear $Q^{u_1 v_3}$ does not exist and for all pseudo-ears $Q^{u_1 w_1}$ and $Q^{v_3 w_3}$ we have that $C^+(w_1) \cap C^-(w_3) = \emptyset$.
    Since $w_1 \in V(P_2)$ and $w_3 \in V(P_2)$ and $P_2$ can not be a separator (since $G$ is structured), there must be a pseudo-ear $Q^{xy}$ starting at some vertex $ x \in C^-(w_1)$ and ending at some vertex $y \in C^+(w_3)$ that is disjoint from $Q^{u_1 w_1}$ and $Q^{v_3 w_3}$.
    Let $R = Q^{u_1 w_1} \cup Q^{x y} \cup Q^{v_3 w_3}$.
    We claim that $H' = H \cup R$ satisfies the Invariants~\ref{invariant:credit}-\ref{invariant:block-size}.
    To see this, first observe that $C$ is now in a 2EC block (or component) in $H'$ and hence $H'$ satisfies the Invariants~\ref{invariant:degree-lonely-vertex} and~\ref{invariant:block-size}.
    Let $S' = S \cup R $.
    Observe that $|S'| = |S| + |R|$ and $\credits(H') \leq \credits(H) - (|R| + 2) - 2 \cdot \frac{5}{4} - 2 \cdot \frac{3}{4} = \credits(H) - |R|$, since all components of $R$ are merged to a single component and the credits of the two leaves of $C$ and the two bridges of $C$ are not needed anymore. 
    Hence, $\credits(H') + |S'| \leq \credits(H) + |S|$.
    This finishes case $k=3$.

    \textbf{Case 2: $4 \leq k \leq 5$.}
    Let $x_2 \in C^-(w_1)$ and $w_2 \in V(C)$ be the closest vertex to $v_k$ on $C$ such that there is a pseudo-ear $Q^{x_2 w_2}$ starting at $x_2$ and ending at $w_2$.
    Similar to before, it is easy to see that if $y_2 \in \bigcup_{4 \leq i \leq k} V(P_i)$, that then $R = Q^{u_1 w_1} \cup Q^{x_2 w_2}$ satisfies Invariants~\ref{invariant:credit}-\ref{invariant:block-size}.
    Hence, assume $w_2 \in V(P_3)$ (since $P_2$ is not a separator, as $G$ is structured).
    Similarly, let $x_{k-1} \in C^+(w_k)$ and $w_{k-1} \in V(C)$ be the closest vertex to $u_1$ on $C$ such that there is a pseudo-ear $Q^{x_{k-1} w_{k-1}}$ starting at $x_{k-1}$ and ending at $w_{k-1}$.
    Similar to before, it is easy to see that if $w_{k-1} \in \bigcup_{1 \leq i \leq k-3} V(P_i)$, that then $R = Q^{v_k w_k} \cup Q^{x_{k-1} w_{k-1}}$ satisfies Invariants~\ref{invariant:credit}-\ref{invariant:block-size}.
    Hence, assume $w_{k-1} \in V(P_{k-2})$ (since $P_{k-1}$ is not a separator, as $G$ is structured).

    \textbf{Case 2.1: $k =4$.}
    Observe that $w_2 \in V(P_3)$ and $w_{k-1} = w_3 \in V(P_2)$.
    If $C^+(w_3) \cap C^-(w_1) \neq \emptyset$, then one can observe that $R = Q^{u_1 w_1} \cup Q^{x_3 w_3} \cup Q^{v_4 w_4}$ is the desired set of links.
    It is easy to see that in $H' = H \cup R$ all bridges of $C$ (from $H$) are now contained in a 2-edge-connected component (or block).
    Similar to the previous case, it is easy to see that $H'$ satisfies the credit invariant.
    Analogously, one can show that if $C^+(w_4) \cap C^-(w_2) \neq \emptyset$, then $R = Q^{u_1 w_1} \cup Q^{x_2 w_2} \cup Q^{v_4 w_4}$ is the desired set of links.
    Hence, assume that $C^+(w_3) \cap C^-(w_1) = \emptyset$ and $C^+(w_4) \cap C^-(w_2) = \emptyset$.
    Let $Q_2$ be the sub-path of $P_2$ from $w_3$ to $v_2$, let $Q_3$ be the sub-path of $P_3$ from $w_2$ to $u_3$ and let $Q = Q_2 \cup Q_3 + v_2 u_3$.
    Observe that $Q$ is a separator of $G$ (otherwise there would be a different pseudo-ear).
    Let $(V_1, V_2)$ be a partition of $V(G) \setminus V(Q)$ such that there are no edges between $V_1$ and $V_2$ and $V(P_1) \subseteq V_1$ and $V(P_k) \subseteq V_2$.
    Since $G$ is structured and hence $\opt(G) \geq 12$, we have that either $\opt(G \setminus V_2) | Q) \geq 4$ or $\opt(G \setminus V_1) | Q) \geq 4$.
    If both are $ \geq 4$, then $Q$ is a $P^2$-separator, a contradiction.
    
    Hence, without loss of generality, assume that $\opt(G \setminus V_2) | Q) \geq 4$.
    We have that $3 \geq \opt(G \setminus V_1)| Q) \geq 2$.
    If $\opt(G \setminus V_1)| Q) = 2$, then we claim that $H' = (V, E(\mathcal{P}) \cup S')$ satisfies the statement of the lemma, where $S' = (S - v_{k-1} u_k) \cup Q^{u_1 w_1} \cup Q^{x_2 w_2} \cup \OPT(G \setminus V_1)| Q)$.
    First, observe that all edges of $C$ are now contained in a 2EC component (or block) in $H'$, since we are in the case that $C^+(w_3) \cap C^-(w_1) = \emptyset$ and $C^+(w_4) \cap C^-(w_2) = \emptyset$.
    Therefore, Invariants~\ref{invariant:degree-lonely-vertex} and~\ref{invariant:block-size} are easy to see.
    Observe that $|S'| \leq |S| + |Q^{u_1 w_1}| + |Q^{x_2 w_2}| + 1$ and $\credits(H') = \credits(H) - (|Q^{u_1 w_1}| + |Q^{x_2 w_2}| -3) - 2 \cdot \frac{5}{4} - 3 \cdot \frac{3}{4} \leq \credits(H) - (|Q^{u_1 w_1}| + |Q^{x_2 w_2}|) - 1$ and hence \Cref{invariant:credit} follows.
    Finally, if $\opt(G \setminus V_1)| Q) = 3$, note that there is an edge $e \in \delta(Q)$ such that $e \in E(\mathcal{P})$ and $e$ is incident to an endpoint of $Q$.
    Hence, $Q$ is a $P^2$-separator, a contradiction.
    This finishes case $k=4$.

    \textbf{Case 2.2: $k =5$.}
    Observe that $w_2 \in V(P_3)$ and $w_{k-1} = w_4 \in V(P_3)$.
    If $C^-(w_2) \cap C^+(w_4) \neq \emptyset$, similar to before we have that $R = Q^{u_1 w_1} \cup Q^{u_2 w_2} \cup Q^{v_4 w_4} \cup Q^{v_5 w_5}$ is the desired set of edges and $H \cup R$ satisfies the Invariants~\ref{invariant:credit}-\ref{invariant:block-size}.
    Hence, assume $C^-(w_2) \cap C^+(w_4) = \emptyset$.

    Let $x_3 \in C^-(w_2)$ such that there is a pseudo-ear $Q^{x_3 w_3}$ to some vertex $w_3 V^+(w_4)$ such that $w_3$ is closest to $v_5=v_k$.
    Similar to before, one can easily show that if $C^-(w_3) \cap C^+(w_5) \neq \emptyset$, then $R = Q^{u_1 w_2} \cup Q^{x_2 w_2} \cup Q^{x_3 w_3} \cup Q^{v_5 w_5}$ is the desired set of edges such that $H \cup R$ satisfies all invariants and has fewer bridges than $H$.
    Hence, assume that $C^-(w_3) \cap C^+(w_5) = \emptyset$.
    Similar to the case $k=4$ let $Q_3$ be the sub-path of $P_3$ from $w_4$ to $v_3$ and let $Q_4$ be the sub-path of $P_4$ from $w_3$ to $u_4$.
    Define $Q = Q_3 \cup Q_4 \cup v_3 u_4$.
    Observe that $Q$ is a separator, as otherwise this would contradict our assumptions on the pseudo-ears.
    Again, similar to the previous case, let $(V_1, V_2)$ be a partition of $V(G) \setminus V(Q)$ such that there are no edges between $V_1$ and $V_2$ and $V(P_1) \subseteq V_1$ and $V(P_k) \subseteq V_2$.
    Since $G$ is structured and hence $\opt(G) \geq 12$, we have that either $\opt(G \setminus V_2) | Q) \geq 4$ or $\opt(G \setminus V_1) | Q) \geq 4$.
    Assume w.l.o.g.\ that $\opt(G \setminus V_2)| Q) \geq \opt(G \setminus V_1)| Q)$.
    Since $G$ is structured, we have that $\opt(G \setminus V_2)| Q) \geq 4$.
    With the same arguments as used in case $k=4$ it directly follows that if either (i) $\opt(G \setminus V_1)| Q) \geq 4$ or (ii) $\opt(G \setminus V_1)| Q) = 3$ this yields a contradiction, since then $Q$ is a $P^2$-separator. In case (ii), $Q$ is a $P^2$-separator since there is an edge $e \in \delta(Q)$ such that $e \in E(\mathcal{P})$ and $e$ is incident to an endpoint of $Q$.

    Hence, let us assume that $\opt(G \setminus V_1)| Q) = 2$.
    Setting $S' = (S - v_4 u_5) \cup Q^{u_1 w_2} \cup Q^{x_2 w_2} \cup Q^{x_3 w_3} \cup \opt(G \setminus V_1)| Q)$, we claim that $H = (V, E(\mathcal{P}) \cup S')$ has fewer bridges than $H$ and satisfies all invariants.
    To see this, first observe that all edges in $C$ are now contained in a 2EC components (or block) in $H'$, as $C^-(w_3) \cap C^+(w_5) = \emptyset$.
    Hence, Invariants~\ref{invariant:degree-lonely-vertex} and~\ref{invariant:block-size} are maintained for $H'$.
    Furthermore, we have that $|S'| = |Q^{u_1 w_2}| + |Q^{x_2 w_2}| + |Q^{x_3 w_3}| + 1$ and $\credits(H') \leq \credits(H) - (|Q^{u_1 w_2}| + |Q^{x_2 w_2}| + |Q^{x_3 w_3}| - 4) - 2 \cdot \frac{5}{4} - 4 \cdot \frac{3}{4} \leq \credits(H) - (|Q^{u_1 w_2}| + |Q^{x_2 w_2}| + |Q^{x_3 w_3}|) -1$.
    Therefore also \Cref{invariant:credit} is maintained for $H'$.
    This finishes case $k=5$.

    \textbf{Case 3: $k \geq 6$.}
    Recall that $w_2 \in V(P_3)$ and $w_{k-1} \in V(P_{k-2})$.
    Let $x_3 \in C^-(w_2)$ and $w_3 \in V(C)$ be the closest vertex to $v_k$ on $C$ such that there is a pseudo-ear $Q^{x_3 w_3}$ starting at $x_3$ and ending at $w_3$.
    Similar to before, it is easy to see that if $w_3 \in \bigcup_{5 \leq i \leq k} V(P_i)$, that then $R = Q^{u_1 w_1} \cup Q^{x_2 w_2} \cup Q^{x_3 w_3}$ satisfies Invariants~\ref{invariant:credit}-\ref{invariant:block-size}.
    Hence, assume $w_3 \in V(P_4)$ (since $P_2$ is not a separator, as $G$ is structured).
    Similarly, let $x_{k-2} \in C^+(w_{k-1})$ and $w_{k-2} \in V(C)$ be the closest vertex to $u_1$ on $C$ such that there is a pseudo-ear $Q^{x_{k-2} w_{k-2}}$ starting at $x_{k-2}$ and ending at $w_{k-2}$.
    Similar to before, it is easy to see that if $w_{k-2} \in \bigcup_{1 \leq i \leq k-4} V(P_i)$, that then $R = Q^{v_k w_k} \cup Q^{x_{k-1} w_{k-1}} \cup Q^{x_{k-2} w_{k-2}}$ satisfies Invariants~\ref{invariant:credit}-\ref{invariant:block-size}.
    Hence, assume $w_{k-2} \in V(P_{k-3})$ (since $P_{k-1}$ is not a separator, as $G$ is structured).
    
    Let $Q_3$ be the sub-path of $P_3$ from $w_2$ to $v_3$, let $Q_4$ be the sub-path of $P_4$ from $w_3$ to $u_4$ and let $Q = Q_3 \cup Q_4 + v_3 u_4$.
    Observe that $Q$ is a separator by the choice of $w_3$.
    Let $(V_1, V_2)$ be a partition of $G \setminus V(Q)$ such that there are no edges between $V_1$ and $V_2$ in $G$ and $u_1 \in V_1$ and $v_k \in V_2$.
    Observe that $\opt(G \setminus V_1)| Q) \geq 3$ and $\opt(G \setminus V_2)| Q) \geq 3$, as $V_1$ and $V_2$ each contain at least 5 distinct endpoints of paths in $\mathcal{P}$.
    If $\opt(G \setminus V_1)| Q) = 3$, we set $S' = (S \setminus \{ v_1 u_2, v_2 u_3 \}) \cup \opt(G \setminus V_1)| Q)$ and observe that $H' = (V, E(\mathcal{P}) \cup S')$ has fewer bridges than $H$.
    Furthermore, \Cref{invariant:degree-lonely-vertex} is satisfied and also \Cref{invariant:block-size} is satisfied since the newly created block contains at least the two paths $P_1$ and $P_2$.
    To see that \Cref{invariant:credit} is satisfied, observe that $|S'| = |S| +1$ and $\credits(H') \leq \credits(H) + 1 - \frac{5}{4} - 2 \cdot \frac{3}{4} \leq \credits(H) - 1$.
    The case $\opt(G \setminus V_2)| Q) = 3$ is analogous.

    In the remaining case we have that $\opt(G \setminus V_1)| Q) \geq 4$ and $\opt(G \setminus V_2)| Q) \geq 4$. But then $Q$ is a $P^2$-separator, a contradiction.
    This finishes the proof of the lemma.    
\end{proof}

Finally, in the remaining case, we have that each complex component is a trivial component.
The following lemma deals with this case.

\begin{restatable}[]{lemma}{lembridgecoveringnolargecomponentonlypathcomplex}
\label{lem:bridge-covering:no-large-component:only-trivial}
    If every complex component of $H$ is a trivial component, then there are two sets of links $R \subseteq L \setminus S$ and $F \subseteq S$ such that $(H \cup R) \setminus F$ has less bridges than $H$ and 
    $(H \cup R) \setminus F$ satisfies the Invariants~\ref{invariant:credit}-\ref{invariant:block-size}.
    Moreover, $R$ and $F$ can be found in polynomial time.
\end{restatable}

\begin{proof}
    Let $P \in \mathcal{P}$ be the path such that $C = P$ and let $u_1, v_1$ be the endpoints of $P$.
    Note that since $P$ can not be a separator (as $G$ is structured), it directly follows that there must be a pseudo-ear $Q^{u_1 v_2}$ from $u_1$ to $v_1$ visiting some other component $C' \neq C$:
    Note that there must be some edge $u_1 z$, where $z \in V(G) \setminus V(C)$, since $G$ is structured.
    Since $C$ is a single path, $C$ can not be a separator and hence there must be a pseudo-ear $Q^{u_1 v_1}$ starting at $u_1$, ending at $v_1$ and containing $z$.
    
    If $C'$ is a 2EC component, we claim that $R = Q^{u_1 v_1}$ is the desired set of links to obtain $H' = H \cup R$.
    Let $S' = S \cup R$.
    First, observe that all bridges of $C$ are no bridges anymore in $H'$.
    Since $C'$ is a 2EC component and $H$ satisfies \Cref{invariant:block-size}, also $H'$ satisfies \Cref{invariant:block-size}.
    For \Cref{invariant:credit}, observe that $|S'| = |S| + |R|$ and $\credits(H') \leq \credits(H) - (|R| -1) - 2 \leq \credits(H) - |R| -1$.
    Finally, it is also easy to see that \Cref{invariant:degree-lonely-vertex} follows.

    If $C'$ is also a single path, say $P' \in \mathcal{P}$ with $P' \neq P$, then \Cref{invariant:block-size} may not be satisfied.
    To deal with this, let $u_2$ and $v_2$ be the endpoints of $P'$ and let $b$ be the block in $H'$ containing all vertices of $C$.
    
    If there is a pseudo ear $Q^{u_2 v_2}$ we set $H'' = H' \cup Q^{u_2 v_2}$ and $S'' = S \cup Q^{u_2 v_2}$.
    We now prove that $H''$ has fewer bridges and satisfies all invariants.
    Clearly, $H''$ has fewer bridges than $H$ and also satisfies \Cref{invariant:block-size} and \Cref{invariant:degree-lonely-vertex}.
    It remains to prove \Cref{invariant:credit}.
    The number of components in $H''$ compared to $H$ has been reduced by at least $|Q^{u_1 v_1}| + |Q^{u_2 v_2}| -2$. Furthermore, the vertices $u_1, v_1, u_2, v_2$ are no leaves anymore, but are contained in the block $b$, which itself needs a credit of 1.
    We now have that $|S''| = |S| \cup |Q^{u_1 v_1}| \cup |Q^{u_2 v_2}|$ and $\credits(H') \leq \credits(H) - (|Q^{u_1 v_1}| + |Q^{u_2 v_2}|-2) +1 -4$ and hence we have that $\cost(H') \leq \cost(H)$, implying that \Cref{invariant:credit} holds.
    
    Otherwise, $Q^{u_2 v_2}$ does not exist. 
    Since $G$ is structured, there must be a pseudo ear $Q^{u_2 y}$ starting in $u_2$ and ending in some other vertex $y$ such that its witness path contains $b$.
    Similarly, there must be a pseudo ear $Q^{v_2 z}$ starting in $v_2$ and ending in some other vertex $z$ such that its witness path contains $b$. Note that since $Q^{u_2 v_2}$ does not exist in this case, these pseudo-ears are disjoint (except for potentially $y$ and $z$).  
    Let $H'' = H' \cup Q^{u_2 y} \cup Q^{v_2 z}$ and set $S'' = S \cup Q^{u_1 v_1} \cup Q^{u_2 y} \cup Q^{v_2 z}$.
    We now prove that $H''$ has fewer bridges and satisfies all invariants.
    Clearly, $H''$ has fewer bridges than $H$ and also satisfies \Cref{invariant:block-size} and \Cref{invariant:degree-lonely-vertex}.
    It remains to prove \Cref{invariant:credit}.
    The number of components in $H''$ compared to $H$ has been reduced by at least $|Q^{u_1 v_1}| + |Q^{u_2 y}| + |Q^{v_2 z}| -3$. Furthermore, the vertices $u_1, v_1, u_2, v_2$ are no leaves anymore, but are contained in the block $b$, which itself needs a credit of 1.
    We now have that $|S''| = |S| \cup |Q^{u_1 v_1}| \cup |Q^{u_2 y}| \cup |Q^{v_2 z}|$ and $\credits(H') \leq \credits(H) - (|Q^{u_1 v_1}| + |Q^{u_2 y}| + |Q^{v_2 z}| -3) +1 -4$ and hence we have that $\cost(H') \leq \cost(H)$, implying that \Cref{invariant:credit} holds.
\end{proof}

Observe that the above lemmas imply \Cref{lem:bridge-covering:main-iterative}, which finishes the bridge-covering step. 
\Cref{alg:bridge-covering} summarizes these steps.

\section{Gluing Algorithm}
\label{sec:gluing}

In this section we prove the following lemma.

\lemgluingmain*

Throughout this section we assume that we are given a solution $S \subseteq L$ such that $H=(V, E(\mcP) \cup S)$ contains no bridges and satisfies Invariants~\ref{invariant:credit}-\ref{invariant:block-size} as in the above lemma statement.
Our goal is to iteratively decrease the number of components of $H$ so as to maintain Invariants~\ref{invariant:credit}-\ref{invariant:block-size} and without introducing bridges.
Formally, we prove the following lemma.
Iteratively applying it clearly proves \Cref{lem:gluing:main}.

\begin{restatable}[]{lemma}{lemgluingmainiterative}
    \label{lem:gluing:main-iterative}
    Let $G$ be some structured instance of \PAP and $H=(V, E(\mcP) \cup S)$ be a solution satisfying the Invariants~\ref{invariant:credit}-\ref{invariant:block-size} that contains no bridges.
    In polynomial-time we can compute a solution $H'=(V, E(\mcP) \cup S')$ satisfying the Invariants~\ref{invariant:credit}-\ref{invariant:block-size} such that $H'$ contains no bridges and $H'$ has fewer connected components than $H$.
\end{restatable}

The rest of this section is devoted to proving \Cref{lem:gluing:main-iterative}.
Note that since each component of $H$ is 2-edge-connected, from now on only credit according to Rule~(E) is assigned, as we do not have bridges and complex components anymore.
Each (2EC) component is either small or large.
Recall that a 2EC component of $H$ is small if it consists of precisely 2 paths $P_1, P_2 \in \mathcal{P}$ such that the endpoints of $P_1$ are connected to the endpoints of $P_2$ using exactly two links. 
A 2EC component is large if it is not small, i.e., it contains at least three links.
Furthermore, according to the credit scheme, each large component and each small component containing a degenerate path receives a credit of 2 and each small component that does not contain a degenerate path receives a credit of $\frac{3}{2}$.
In the context here, a small component consisting of two paths $P_1, P_2 \in \mcP$ with endpoints $u_1, v_1$ and $u_2, v_2$, respectively, contains a degenerate path if either 
$N(u_1), N(v_1) \subseteq V(P_1) \cup V(P_2)$ or 
$N(u_2), N(v_2) \subseteq V(P_1) \cup V(P_2)$.

We define the \emph{component graph} $\Tilde{G}$ by contracting each component of $H$ to a single vertex.
$\Tilde{G}$ might not be a simple graph and can contain parallel edges and self-loops.
For convenience, we remove self-loops but keep parallel edges.
To distinguish between vertices of $\Tilde{G}$ and $G$, we call the vertices of $\Tilde{G}$ \emph{nodes} instead of vertices.
Note that $\Tilde{G}$ is a 2-edge connected multi-graph and not necessarily 2-node connected.
We identify edges of $\Tilde{G}$ with edges in $G$.
Hence, if we add edges from $\Tilde{G}$ to $H$, we add the corresponding edges from $G$ to $H$.

We will add certain cycles of $\Tilde{G}$ to $H$ in order to decrease the number of components. 
Note that adding a cycle $K$ of $\Tilde{G}$ to $H$ decreases the number of components by at least one. 
Furthermore, it is easy to see that the newly created component in $H \cup K$ that contains the nodes/components incident to $C$ is 2-edge connected. 
Additionally, since $H$ satisfied \Cref{invariant:degree-lonely-vertex} and \Cref{invariant:block-size}, these invariants are clearly also satisfied for $H \cup K$.
Hence, whenever we add a cycle $K$ of $\Tilde{G}$ to $H$, it only remains to prove \Cref{invariant:credit}.
In the following, we show which cycles can be safely added to $H$ such that \Cref{invariant:credit} is satisfied. Furthermore, we show how we can find such cycles and prove that they always exist (in structured graphs).

Let $K$ be a cycle in $\Tilde{G}$.
We say a small component $C_S$ is \emph{shortcut} w.r.t.\ $K$ if in $G[V(C_S)]$ there exists a Hamiltonian path from $u$ to $v$ containing exactly one link, where $u$ and $v$ are the vertices incident to $K$ when $C_S$ is expanded.
Note that if $C_S$ is shortcut by $K$, then $u$ and $v$ must be endpoints of the two paths $P_1$ and $P_2$ of $C_S$, respectively.

Let $K$ be a cycle in $\Tilde{G}$ that shortcuts some small component $C_S$.
We say that we \emph{augment} $H$ along $K$ by adding all links from $K$ to $H$ and for each small component $C_S$ that is shortcut by $K$, we replace all links in $H[V(S)]$ by the Hamiltonian path from $u$ to $v$ containing exactly one link, where $u$ and $v$ are the vertices incident to $K$ when $C_S$ is expanded. We denote this operation by $H \Dot{\Delta} K$.

We now define \emph{good cycles}. 
In \Cref{fig:gluing} there are two examples of good cycles.

\begin{definition}
    \label{def:good-cycle}
    A cycle $K$ in $\Tilde{G}$ is called \emph{good} if one of the following conditions is satisfied:
    \begin{itemize}
        \item $K$ is incident to at least $2$ components having a credit of $2$, i.e., that are large or small and contain a degenerate path.
        \item $K$ is incident to a small component $C_S$ that is shortcut. 
    \end{itemize}
\end{definition}

First, we show that if $K$ is a good cycle of $\Tilde{G}$, then $H \Dot{\Delta} K$ is a solution as desired in \Cref{lem:gluing:main-iterative}.

\begin{restatable}[]{lemma}{lemgluingaddinggoodcycles}
\label{lem:gluing:adding-good-cycles}
    Let $K$ be a good cycle in the component graph $\Tilde{G}$ and let $H' = H \Dot{\Delta} K$.
    Then $H'$ satisfies the Invariants~\ref{invariant:credit}-\ref{invariant:block-size}, contains no bridges and $H'$ has fewer connected components than $H$.
\end{restatable}

\begin{proof}
    We first show that $H'$ does not contain bridges.
    By the definition of $H \Dot{\Delta} K$, all components incident to $K$ are now in a single 2EC-component.
    This is obvious if $K$ does not contain small components that are shortcut.
    If $K$ also contains small components that are shortcut, then the statement follows since for each such small component we have a Hamiltonian path from $u$ to $v$, where $u$ and $v$ are the vertices incident to $K$ when $C_S$ is expanded.
    Hence, $H'$ can not have any bridges.
    Furthermore, it is easy to see that the number of components in $H'$ is at most the number of components in $H$ minus 1 as $K$ has length at least 2.

    It remains to show that $H'$ satisfies \Cref{invariant:credit}.
    Let $k_\ell$ be the number of components incident to $K$ that have a credit of 2 and let $k_s$ be the number of (small) components incident to $K$ that have a credit of $\frac 32$.
    First, assume that $K$ contains at least 2 components having a credit of 2, i.e., $k_\ell \geq 2$.
    Then we have that $\credits(H') \leq \credits(H) - 2 k_\ell - \frac{3}{2} k_s + 2 \leq \credits(H) - 2 (k_\ell -1) - k_s$, since the credits of all components incident to $K$ are not needed anymore, except for the newly created 2EC component, which receives a credit of~2.
    Let $S' = H' \cap L$. 
    Since $|S'| \leq |S| + k_\ell + k_s$ and $k_\ell \geq 2$, we have $\credits(H') + |S'| \leq \credits(H) - 2 (k_\ell -1) - k_s + |S| + k_\ell + k_s \leq \credits(H) + |S|$.
    
    Next, assume that $K$ contains at least one small component $C_S$ that is shortcut.
    Then we have that $\credits(H') \leq \credits(H) - 2 k_\ell - \frac{3}{2} k_s + 2$, since the credits of all components incident to $K$ are not needed anymore, except for the newly created 2EC component, which receives a credit of~2.
    Furthermore, since $H[C_S]$ contains precisely two links and $H'[C_S]$ only contains one link, we have that $|S'| \leq |S| + k_\ell + k_s - 1$.
    Now, since $k_\ell + k_s \geq 2$, we have $\credits(H') + |S'| \leq \credits(H) - 2 k_\ell - \frac{3}{2} k_s + |S| + k_\ell + k_s -1 \leq \credits(H) + |S| - k_\ell - \frac{1}{2} \cdot k_s + 1 \leq \credits(H) + |S|$.
\end{proof}

Second, we show that in polynomial time we can find a good cycle in $\Tilde{G}$, if one exists.

\begin{restatable}[]{lemma}{lemgluingfindinggoodcycles}
\label{lem:gluing:finding-good-cycles}
    In polynomial time we can find a good cycle in the component graph $\Tilde{G}$, if one exists.
\end{restatable}

\begin{proof}
    We simply check both possibilities for a good cycle.
    In order to find a good cycle in $\Tilde{G}$ that contains two components receiving a credit of 2, we do the following:
    We check every pair of such components $C_1$, $C_2$ and check if there are two edge-disjoint paths from $C_1$ to $C_2$.
    If there is a good cycle in $\Tilde{G}$ that contains 2 large components, then this algorithm will find such a good cycle in polynomial time.

    In the other case, we want to find a good cycle $K$ that contains a small component $C$ that is shortcut w.r.t\ $K$.
    Recall that if $C$ is shortcut w.r.t.\ $K$, then the edges incident to $C$ must be endpoints $u$ and $v$ of the two paths $P_1$ and $P_2$ of $C$, such that there is a Hamiltonian path from $u$ to $v$ in $H[C]$ with exactly one link.
    Therefore, for each small component $C$ and each such pair $u$ and $v$ as above, we simply check if there is a path from $u$ to $v$ in the component graph of $G' = G[V(G) \setminus (V(H) \setminus \{u, v \})]$, i.e., in $\Tilde{G'}$.
    If there is a good cycle $K$ that shortcuts some small component $C$, then this algorithm will find a good cycle in polynomial time.
\end{proof}

Now, it remains to show that if $G$ is structured, then there is always a good cycle in the component graph $\Tilde{G}$.

\begin{restatable}[]{lemma}{lemgluingexistencegoodcycles}
\label{lem:gluing:existence-good-cycles}
    If $G$ is a structured instance, then there is a good cycle in the component graph $\Tilde{G}$.
\end{restatable}

\begin{proof}
    If $H$ only contains components that have a credit of $2$, then since $G$ is 2EC, there must be a cycle of length two in the component graph.
    It is easy to see that this is a good cycle.

    Otherwise, if $H$ contains a small component $C$ that has a credit of $\frac{3}{2}$, i.e., it does not contain a degenerate path, then we show that there must be a good cycle incident to $C$ in the component graph $\Tilde{G}$.
    In particular, we show that there is a cycle $K$ in $\Tilde{G}$ such that $C$ is shortcut.
    Since $C$ is small, it consists of two distinct paths $P_1, P_2 \in \mathcal{P}$ with endpoints $u_1, v_1$ and $u_2, v_2$, respectively, such that there are links $u_1 u_2$ and $v_1 v_2$ in $H$.
    Furthermore, since $C$ does not contain a degenerate path, we have that there exists some $z_1 \in N(u_1) \cup N(v_1)$ and some $z_2 \in N(u_2) \cup N(v_2)$ with $z_1, z_2 \in V \setminus (V(P_1) \cup V(P_2))$.
    
    First, we show that there must be two vertices $x, y \in \{u_1, u_2, v_1, v_2 \}$ such that there is a Hamiltonian path between $x$ and $y$ in $G[V(C)]$ with exactly one link and $x$ and $y$ each have a link to $V(G) \setminus V(C)$. 
    If $u_1$ and $u_2$ have a link to $V(G) \setminus V(C)$ then it is easy to see that the desired Hamiltonian path is $P_1 v_1 v_2 P_2$. 
    Similarly, if $v_1$ and $v_2$ have a link to $V(G) \setminus V(C)$, then it is easy to see that the desired Hamiltonian path is $P_1 u_1 u_2 P_2$.
    Hence, assume that neither case occurs.
    Since $z_1$ and $z_2$ exist, we know that either both $v_1$ and $u_2$ have a link to $V(G) \setminus V(C)$ or both $u_1$ and $v_2$ have a link to $V(G) \setminus V(C)$.
    Since the two cases are symmetric, let us assume that $v_1$ and $u_2$ have a link to $V(G) \setminus V(C)$.
    If $u_1 v_2 \in L$, then $P_1 u_1v_2 P_2$ is the desired Hamiltonian Path. 
    Hence, let us assume that $u_1 v_2 \notin L$.
    But then $N(u_1) \subseteq V(P_1) \cup V(P_2) - v_2$ and $N(v_2) \subseteq V(P_1) \cup V(P_2) - u_1$.
    Now it can be easily observed that any feasible solution must pick at least two links from $G[V(P_1) \cup V(P_2)]$ and hence $C$ is a contractible cycle, a contradiction to the fact that $G$ is structured.

    Hence, let $x$ and $y$ be the desired vertices as above and let $e_x$ and $e_y$ be the links incident to $x$ and $y$, respectively, which go to $V(G) \setminus V(C)$.
    We claim that there is a good cycle using these two edges.
    Let $C_x$ and $C_y$ be the components distinct from $C$ to which $e_x$ and $e_y$ are incident to.
    If $C_x = C_y$, then we directly have a good cycle, since $C$ can be shortcut using $e_x$ and $e_y$.
    Else, if $C_x \neq C_y$, we claim that there is a simple path $Q^{xy}$ from $C_x$ to $C_y$ in $\Tilde{G} \setminus c$, where $c$ is the vertex in $\Tilde{G}$ corresponding to the component $C$ of $H$.
    If not, then this means that $C$ is a separator in $G$ such that there is a partition $(V_1, V_2)$ of $V(G) \setminus V(C)$ such that there are no edges between $V_1$ and $V_2$.
    However, since in $H$ each connected component contains at least two distinct paths $P, P' \in \mathcal{P}$, we have that $\opt((G \setminus V_i)|C) \geq 3$.
    But then $C$ is a $C^2$-separator, a contradiction to the fact that $G$ is structured.
    Hence, $K = e_x Q^{xy} e_y$ is a good cycle in $\Tilde{G}$.
    Moreover, note that all computational steps in this proof can be done in polynomial time.
\end{proof}

The algorithm for the gluing step is quite easy: As long as we have more than one component in our current solution $H$, we simply compute a good cycle $C$ and augment $H$ along $C$.
It is straightforward to see that the above three lemmas imply \Cref{lem:gluing:main-iterative}, which in turn proves \Cref{lem:gluing:main}.

\section{Preliminaries on FAP and Proof of \Cref{thm:FAP:main} }
\label{sec:preliminaries:FAP-PAP}
In this section we prove \Cref{thm:FAP:main} using \Cref{thm:PAP:main} and results from~\cite{grandoni2022breaching}.
We first give an overview on how the result for FAP is obtained in~\cite{grandoni2022breaching}. 
With these insights, we show how we can use the algorithm from \Cref{thm:PAP:main} to obtain \Cref{thm:FAP:main}.

\paragraph*{Overview of result for FAP from~\cite{grandoni2022breaching}.}
Consider some instance $G=(V, F \cup L)$ of FAP, where $n_{\rm comp}$ is the number of connected components of $(V, F)$.
Note that $\opt(G) \geq n_{\rm comp}$.
The algorithm in~\cite{grandoni2022breaching} is the first algorithm that achieves a better-than-2 approximation for FAP.
We briefly highlight the key ideas of their algorithm, which actually consists of two algorithms: One, which performs well if the number of connected components $n_{\rm comp}$ in the forest is small (and hence $(V, F)$ is \emph{almost} a tree) and another, which performs well if $n_{\rm comp}$ is large (and hence $(V, F)$ is \emph{far away} from being a tree).
Therefore, in the first case, it makes sense to augment the forest $F$ with some cleverly chosen links $S' \subseteq L$ such that $(V, F \cup S')$ is a spanning tree of~$G$.
We then use an approximation algorithm for TAP on the resulting instance.
Hence, one can obtain the following lemma, proven in~\cite{grandoni2022breaching}. 
This is an oversimplification of their actual algorithm, which needs several key ingredients to work as described.

\begin{lemma}[Lemma 3.1 in~\cite{grandoni2022breaching}]
\label{lem:tap}
Let $\eps > 0$ be a constant. Given an instance $(V, F \cup L)$ of FAP, we can compute in polynomial time a solution of size at most $n_{\rm comp} + (1 + \ln(2 - \frac{n_{\rm comp}}{\opt} + \eps) \cdot \opt$.
\end{lemma}

Note that this algorithm performs well if $\opt$ is sufficiently larger than $ n_{\rm comp}$.
Hence, in light of \Cref{lem:tap} it makes sense to design an approximation algorithm that performs well if $\opt$ is close to $n_{\rm comp}$.
Therefore, the authors of~\cite{grandoni2022breaching} introduced a $(\rho, K)$-approximation algorithm, which is an algorithm that computes a solution of cost at most $\rho \cdot \opt + K \cdot (\opt - n_{\rm comp})$.
Together with \Cref{lem:tap}, any $(\rho, K)$ approximation for FAP with $\rho < 2$ and constant $K$ implies a better-than-2 approximation algorithm for FAP.

Instead of designing a $(\rho, K)$-approximation algorithm directly for FAP (which might be quite difficult), they give a reduction to PAP in the following sense.

\begin{lemma}[Lemma 3.4 in~\cite{grandoni2022breaching}]
\label{lem:red-FAP-PAP}
Given a polynomial time $(\rho, K)$-approximation algorithm for PAP for some constant $\rho \geq 1$ and $K \geq 0$, there is a polynomial-time $(\rho, K + 2(\rho - 1))$-approximation algorithm for FAP.
\end{lemma}

Therefore, it suffices to give such an approximation algorithm for PAP.

\begin{lemma}[Lemma 3.5 in~\cite{grandoni2022breaching}]
\label{lem:old-PAP}
There is a $(\frac{7}{4}, \frac{7}{4})$-approximation algorithm for PAP.
\end{lemma}

Hence, by \Cref{lem:red-FAP-PAP} this implies a $(\frac{7}{4}, \frac{13}{4})$-approximation for FAP.
By setting $n_{\rm comp} = \delta \cdot  \opt$ for some $\delta \in [0, 1]$ and computing the minimum of $\{ \delta + 1 + \ln (2 - \delta), \frac{7}{4} + \frac{13}{4}(1- \delta) \} + \varepsilon$ one obtains the worst-case approximation ratio of $1.9973$ for FAP.
Similarly, one obtains the worst-case approximation ratio of $1.9913$ for PAP (by substituting $\frac{13}{4}$ with $\frac{7}{4}$ in the calculation).

\paragraph*{Improved Approximation Algorithm for FAP.}

We now show how to obtain \Cref{thm:FAP:main}.
The algorithm is the same as for \Cref{thm:PAP:main}, but we use a slightly different bound.
We did not optimize the bound and there might be further room for improvement.
We show that we can derive a different bound on the approximation ratio of the starting solution while still satisfying the credit scheme and all invariants (yet, for \Cref{invariant:credit} we have a different bound).
In particular, we show that we can also bound the cost of the starting solution $H$ by $\cost(H) \leq \frac 74 \cdot \opt + (\opt - n_{\rm comp})$. Using \Cref{lem:bridge-covering:main} and \Cref{lem:gluing:main}, we can then obtain a feasible solution $H = (V, E(\mcP) \cup S)$ with $|S| \leq \frac 74 \cdot \opt + (\opt - n_{\rm comp})$, i.e., the solution is also a $(\frac 74, 1)$-approximation for \PAP.
In fact, to obtain the bound on the starting solution, we only need Algorithm~B from \Cref{sec:initial-solution}.
Before we state this lemma, we need one additional lemma which is a slight adaptation of Lemma~5.3 from~\cite{grandoni2022breaching}.
\begin{lemma}[See also Lemma~5.3 from~\cite{grandoni2022breaching}]
\label{lem:FAP:helper}
    The number of lonely leaf vertices in the solution computed by Algorithm~B is at most $4 \cdot (\opt - n_{\rm comp})$.
\end{lemma}

\begin{proof}
    The proof is almost identical to the proof of Lemma~5.3 from~\cite{grandoni2022breaching}.

Let $\text{OPT} \subseteq L$ be an optimal solution and let $M_{\text{OPT}} \subseteq \text{OPT} \cap \bar{L}$ be a maximal matching in $\text{OPT}$ between the endvertices of paths that contains no links from $\hat{L}$.
Recall that $B_1 \subseteq \bar{L}$ is a matching such that $|B_1| \geq |M_{\text{OPT}}|$, where $B_1$ is defined as in Algorithm~B.
Thus it suffices to show that at most $4 \cdot (\text{opt} - n_{\text{comp}})$ leaves of the forest $(V, E(\mcP))$ have no incident edge in the matching $M_{\text{OPT}}$.

First, we observe that every connected component of the forest $(V, E(\mcP))$ has at least two edges incident to it that have their other endpoint in a different connected component. 
Because links in $\hat{L}$ have both endpoints in the same connected component of the forest, this implies
\begin{equation}
\label{eq1}
\text{opt} \geq n_{\text{comp}} + \text{opt}_{\text{bad}},
\end{equation}
where $\text{opt}_{\text{bad}}$ denotes the number of links in $\text{OPT} \cap \hat{L}$. 
Let $V_{\text{leaf}} \subseteq V$ denote the set of leaves of the forest $(V, E(\mcP))$. Because every connected component of this forest is a path, we have
\[
2 \cdot n_{\text{comp}} = |V_{\text{leaf}}| = 2 \cdot |M_{\text{OPT}}| + \left| \{ v \in V_{\text{leaf}} : M_{\text{OPT}} \cap \delta (v) = \emptyset \} \right|.
\]
Using the maximality of $M_{\text{OPT}}$ and the fact that every leaf of $(V, F)$ is the endpoint of a link in $\text{OPT}$, this implies
\begin{align*}
\text{opt} & \ \geq |M_{\text{OPT}}| + \left| \{ v \in V_{\text{leaf}} : M_{\text{OPT}} \cap \delta (v) = \emptyset \} \right| - \text{opt}_{\text{bad}} \tag{2} \\
& \ = n_{\text{comp}} + \frac{1}{2} \left| \{ v \in V_{\text{leaf}} : M_{\text{OPT}} \cap \delta (v) = \emptyset \} \right| - \text{opt}_{\text{bad}}.
\end{align*}
Adding up the Inequalities~(1) and~(2), we obtain
\[
2 \cdot (\text{opt} - n_{\text{comp}}) \geq \frac{1}{2} \left| \{ v \in V_{\text{leaf}} : M_{\text{OPT}} \cap \delta (v) = \emptyset \} \right|.
\]
Thus, the number of unmatched leaves is at most $4 \cdot (\text{opt} - n_{\text{comp}})$.
\end{proof}

\begin{restatable}[]{lemma}{lemstartingsolutionboundcreditsfap}
    \label{lem:starting-solution:bound:credits:FAP}
    Let $S$ be a solution computed by Algorithm B.
    If each track of length $1$ receives a credit of $\frac 34$, each track of length $3$ a credit of $2$, each lonely leaf vertex a credit of $\frac 54$ if it is not expensive and $\frac 32$ if it is expensive, each degenerate path an additional credit of $\frac 12$, and we additionally have a credit of $\opt - n_{\rm comp}$, then we can redistribute the credits such that the credit rules (A)-(E) are satisfied.
    Furthermore, Invariants~\ref{invariant:degree-lonely-vertex} and~\ref{invariant:block-size} are satisfied.
\end{restatable}

\begin{proof}
    Note that the only difference between the statement of this lemma and \Cref{lem:starting-solution:bound:credits} is that here each lonely leaf vertex receives a credit of $\frac 54$ if it is not expensive and $\frac 32$ if it is expensive, while in \Cref{lem:starting-solution:bound:credits}, it receives an additional $+ \frac 14$ credit in both cases. 
    Furthermore, in the lemma statement here, there is an additional credit of $\opt - n_{\rm comp}$.
    Hence, following the proof of \Cref{lem:starting-solution:bound:credits}, we can redistribute all credits such that all invariants are satisfied and also the credit assignment works out, except for a missing $\frac 14$ credit for each lonely leaf vertex.
    We show that the additional credit of $\opt - n_{\rm comp}$ compensates for this. Indeed, \Cref{lem:FAP:helper} implies that the number of lonely leaf vertices in the solution computed by Algorithm~B is bounded by $4 \cdot (\opt - n_{\rm comp})$, and hence the additional $\opt - n_{\rm comp}$ credits compensate for the loss of $\frac 14$ credit for each lonely leaf vertex. This finishes the proof.
\end{proof}

Next, we give a different bound on the cost of our Algorithm. Here, we only need to consider the solution that is returned by Algorithm~B.

\begin{restatable}[]{lemma}{lemstartingsolutionboundcreditsfap2}
    \label{lem:starting-solution:bound:credits:FAP2}
    Let $S \subseteq L$ be a solution computed by Algorithm B and let $H= (V, E(\mcP) \cup S)$.
    We have $\cost(H) \leq (\frac{7}{4} + \varepsilon) \cdot \opt + (\opt - n_{\rm comp})$.
\end{restatable}

\begin{proof}
    Using the definitions established in \Cref{sec:initial-solution}, and following the proof of \Cref{lem:starting-solution:bound:ALGA-B-C}, we first bound $\cost(H)$ as follows.
    Note that the number of lonely leaf vertices in $H$ is bounded by $2|\mcp| - 2 \beta_1 - 4 \beta_2$.
    Furthermore, we assign a credit of $\frac 34$ for each of the $\beta_1$ many tracks containing one link, a credit of $2$ for each of the $\beta_2$ many tracks containing three links and a credit of $\frac{5}{4}$ for each lonely leaf vertex of $H_A$ and an additional credit of $\frac{1}{4}$ for each expensive lonely leaf vertex.
    Furthermore, observe that we give an additional credit of $\frac{1}{2}$ to each degenerate path.
    Finally, we have an additional credit of $\opt - n_{\rm comp}$.
    Hence, we can bound $\cost(H)$ by
    \begin{align*}\cost(H) & \ = |S| + \credits(H) \\
    & \ \leq \beta_1 + 3 \beta_2 + \frac 54 (2|\mcp| - 2 \beta_1 - 4 \beta_2) + \frac 34 \beta_1 + 2 \beta_2 + \frac 14 \beta_0^e + \frac{1}{2}\varepsilon' |\mcP| + \opt - n_{\rm comp} \\
    & \ = \frac 52 |\mcp| - \frac 34 \beta_1 + \frac 14 \beta_0^e + \frac{1}{2}\varepsilon' |\mcP| + \opt - n_{\rm comp}  \ \ . 
    \end{align*}  
    We could now just put this inequality together with the other inequalities into a solver as we did in \Cref{sec:initial-solution}, and obtain the desired result. However, in this case we can even simply compute the bound directly, which we do in the following.
    
    Observe that $\beta_2$ is canceled out in the above inequality. 
    Hence, the number of tracks containing 3 links in $S$ does not matter.
    Furthermore, for obtaining a valid bound on $\opt$, this means that the number of tracks of size 3 in the optimum solution can be chosen freely and hence $\omega_3$ is no longer needed.
    Following the bounds we established in \Cref{sec:initial-solution}, we have the following additional inequalities:
    \begin{enumerate}
        \item $\opt \geq 2 - \omega_1 - \omega_2$,
        \item $\omega_0 + 2\omega_1 + 4 \omega_2 \leq 2$,
        \item $\beta_1 \geq \omega_1$, and
        \item $\omega_0 \geq \beta_0^e$
    \end{enumerate}
    Note that all variables are non-negative.
    Hence, since the second inequality is always satisfied with equality in an optimum solution, we have 
    $$ \omega_2 = \frac 12 - \frac 12 \omega_1 - \frac 14 \omega_0 \ .$$
    Plugging this into the first inequality, and using the third and fourth inequallity, we obtain 
    $$ \opt \geq 2 - \omega_1 - \omega_2 = \frac 32 - \frac 12 \omega_1 + \frac 14 \omega_0 \geq \frac 32 - \frac 12 \beta_1 + \frac 14 \beta_0^e \ .$$
    Now, we have the following inequalities:
    \begin{enumerate}
        \item $\opt \geq \frac 32 - \frac 12 \beta_1 + \frac 14 \beta_0^e$, and
        \item $\cost(H) \leq \frac 52 |\mcp| - \frac 34 \beta_1 + \frac 14 \beta_0^e + \frac{1}{2}\varepsilon' |\mcP| + \opt - n_{\rm comp}$
    \end{enumerate}
    Leaving out the additive $\frac{1}{2}\varepsilon' |\mcP| + \opt - n_{\rm comp}$ in the latter bound, the ratio of $\frac{\cost(H) - \frac{1}{2}\varepsilon' |\mcP| - \opt + n_{\rm comp}}{\opt}$ is maximized for $\beta_0^e = 0$ and $\beta_1 = 1$.
    Hence, since $\frac{1}{2}\varepsilon' |\mcP| \leq \frac 12 \varepsilon \opt$, we get 
    $$\cost(H) \leq (\frac{7}{4} + \varepsilon) \cdot \opt + (\opt - n_{\rm comp}) \ ,$$
    proving our claim.
\end{proof}

We are now ready to prove \Cref{thm:FAP:main}.

\begin{proof}[Proof of \Cref{thm:FAP:main}]
\Cref{lem:starting-solution:bound:credits:FAP2} together with \Cref{lem:structured:main}, \Cref{lem:bridge-covering:main}, and \Cref{lem:gluing:main} imply a $(\frac 74 + \varepsilon, 1)$-approximation for \PAP, which together with \Cref{lem:red-FAP-PAP} implies a $(\frac 74 + \eps, \frac{10}{4} + \varepsilon)$-approximation for FAP.
This bound together with \Cref{lem:tap} implies a $\apxrf$-approximation for FAP by computing the minimum of $\{ \delta + 1 + \ln (2 - \delta), \frac 74 + \frac{10}{4}(1- \delta) \} + \varepsilon$, which is attained for $\delta \approx 0.902$.
\end{proof}

\section{Conclusion}
\label{sec:conclusion}

In this we paper we designed a $\apxr$-approximation for the Path Augmentation Augmentation problem, implying a $\apxrf$-approximation for the Forest Augmentation using previous results of~\cite{grandoni2022breaching}.
We established several key ingredient: (1) a $(\frac{7}{4}+ \varepsilon)$-approximation preserving reduction to structured instances of PAP; (2) a new relaxation of the problem together with a related packing problem, which we used to obtain a good starting solution using a factor-revealing LP; (3) leveraging concepts of other connectivity augmentation problems to turn the starting solution into a feasible solution. 

There might be room for improvement, in particular in bounding the cost of the starting solution. Obtaining a better bound on the cost of the starting solution, while still maintaining the credit invariants, directly gives an improved approximation ratio for \PAP, and therefore also FAP. However, to improve the approximation ratio below $\frac{7}{4}$ requires some new insights.

The previous algorithm for FAP and PAP introduced in~\cite{grandoni2022breaching} heavily relied on techniques used for the Tree Augmentation Problem (TAP), without which a better-than-2 approximation would not be possible, neither for FAP nor for PAP.
Our algorithm for PAP is independent of these techniques. Hence, we hope that the techniques we introduced might lead to further improvements for PAP and FAP without relying on the techniques for TAP, which can be used to obtain better approximation algorithms.
In particular, we believe that our reduction to structured instances and the new relaxation 2ECPC and its corresponding packing problem TPP might be of independent interest.

\bibliography{bib.bib}

\end{document}